\documentclass[submission]{eptcs}
\providecommand{\event}{}
\usepackage[moderate,paragraphs=normal]{savetrees}

\usepackage{amsmath,amssymb}
\usepackage{algorithmic}
\usepackage{listings}
\usepackage{graphicx}
\usepackage{textcomp}
\usepackage{soul}
\usepackage{xcolor}

\usepackage{amsthm}
\usepackage{hyperref}

\usepackage{cleveref}
\usepackage{stmaryrd}

\usepackage{mathtools}
\usepackage{tikz}
\usepackage{tikz-network}
\usepackage{xspace}
\usepackage{xargs}
\usepackage{fontenc}
\usepackage{csquotes}

\usepackage{subcaption}
\usepackage{caption}

\usepackage{tikz-cd}
\usepackage[draft]{tikzit}
\usepackage{adjustbox}

\usepackage{xfrac} 
\usepackage{wrapfig}

\usepackage{microtype}

\usepackage{mathpartir}

\usepackage{aliascnt}

\usepackage{multirow}
\usepackage{bigdelim}
\usepackage{colortbl}
\usepackage{colortbl}
\makeatletter
\providecommand\CT@arc@{}
\providecommand\CT@row@color{}
\providecommand\CT@cell@color{}
\providecommand\CT@do@color{}
\makeatother
\usepackage{minted}

\usepackage[numbers]{natbib}

\makeatletter
\renewenvironment{proof}[1][\proofname]{%
  \par\pushQED{\qed}%
  \normalfont
  \topsep=2pt\partopsep=2pt\parsep=0pt\itemsep=0pt
  \trivlist
  \item[\hskip\labelsep\itshape #1\@addpunct{.}]%
}{%
  \popQED\endtrivlist\@endpefalse
}
\makeatother

\newtheoremstyle{tight}
  {0pt}   
  {1pt}   
  {}      
  {}      
  {\bfseries\scshape} 
  {.}     
  {.5em}  
  {}      

\theoremstyle{tight}

\newtheorem{theorem}{Theorem}[section]

\newaliascnt{definition}{theorem}
\newtheorem{definition}[definition]{Definition}
\aliascntresetthe{definition}

\newaliascnt{example}{theorem}
\newtheorem{example}[example]{Example}
\aliascntresetthe{example}

\newaliascnt{proposition}{theorem}
\newtheorem{proposition}[proposition]{Proposition}
\aliascntresetthe{proposition}

\newaliascnt{remark}{theorem}
\newtheorem{remark}[remark]{Remark}
\aliascntresetthe{remark}

\newaliascnt{lemma}{theorem}
\newtheorem{lemma}[lemma]{Lemma}
\aliascntresetthe{lemma}

\newaliascnt{corollary}{theorem}
\newtheorem{corollary}[corollary]{Corollary}
\aliascntresetthe{corollary}

\definecolor{e-color}{named}{orange}
\definecolor{closed-color}{named}{blue}

\newcommand{\mathcolorbox}[2]{\colorbox{#1}{\ensuremath{#2}}}

\newcommand\HypI[1]{\textbf{HypI}(#1)}
\newcommand\MdaCospans{\textbf{MHypI}(\Sigma)}
\newcommand{\Ecospans}{{\catname{CEHypI(\Sigma)}}}
\newcommand{\MdaEcospans}{{\catname{MCEHypI(\Sigma)}}}
\newcommand{\WellTypedMdaEcospans}{{\catname{MCEHypI(\Sigma)}}}
\newcommand{\hashtag}{{\#}}
\newcommand{\consistency}{{\smile}}
\newcommand{\join}{+}
\newcommand{\defeq}{\stackrel{\triangle}{=}}

\newcommand{\missing}[1]{{\color{red}\bfseries [TODO]}}

\newcommand{\defTodo}[2]{%
	\expandafter\newcommand\csname #1\endcsname[1]{%
		\todo[linecolor=#2,backgroundcolor=#2!25,bordercolor=#2,inline,size=\tiny]{\textbf{#1}: ##1}}}

\newcommand{\defTODO}[2]{%
	\expandafter\newcommand\csname #1\endcsname[1]{%
		\todo[linecolor=#2,backgroundcolor=#2!25,bordercolor=#2,inline,size=\tiny,caption={\textbf{(#1 LONG TODO)}}]{##1}}}

\newcommand{\AlekseiColor}{RoyalBlue}

\defTodo{Aleksei}{\AlekseiColor}

\defTODO{ALEKSEI}{\AlekseiColor}

\newcommand{\catname}[1]{\mathbf{#1}}

\newcommand\id{\textsf{id}}
\newcommand\sym{\textsf{sym}}

\newcommand\I{I}

\newcommand{\comonoid}{%
	\begin{tikzpicture}[baseline=0.4ex, scale=0.4, line width=0.8pt]
		\path [use as bounding box] (0,0) rectangle (0.6,0.8);
		\draw (0.2,0) -- (0.2,0.4);
		\filldraw[black] (0.2,0.4) circle (0.08);
		\draw (0,0.8) .. controls (0,0.65) and (0.1,0.45) .. (0.2,0.4);
		\draw (0.4,0.8) .. controls (0.4,0.65) and (0.3,0.45) .. (0.2,0.4);
	\end{tikzpicture}%
}
\newcommand{\counit}{%
	\begin{tikzpicture}[baseline=0.4ex, scale=0.4, line width=0.8pt]
		\path [use as bounding box] (0,0) rectangle (0.4,0.8);
		\draw (0.2,0) -- (0.2,0.4);
		\filldraw[black] (0.2,0.4) circle (0.08);
	\end{tikzpicture}%
}

\tikzstyle{termbox}=[draw=term, fill={term!10}, rounded corners, minimum size=20pt]
\tikzstyle{tallbox}=[draw=term, fill={term!10}, rounded corners, minimum width=20pt, minimum height=40pt]
\tikzstyle{termpic}=[x=1pt, y=1pt, inner sep=0pt, outer sep=0pt, thick]
\tikzstyle{box}=[shape=rectangle, text height=1.5ex, text depth=0.25ex, yshift=0.5mm, fill=white, draw=black, minimum height=12.5mm, yshift=-0.5mm, minimum width=7.5mm, font={\small}]
\tikzstyle{Z dot}=[inner sep=0mm, minimum size=2mm, shape=circle, draw=black, fill={rgb,255: red,160; green,255; blue,160}]
\tikzstyle{Z phase dot}=[minimum size=5mm, font={\footnotesize\boldmath}, shape=rectangle, rounded corners=2mm, inner sep=0.2mm, outer sep=-2mm, scale=0.8, tikzit shape=circle, draw=black, fill={rgb,255: red,160; green,255; blue,160}, tikzit draw=blue]
\tikzstyle{X dot}=[Z dot, shape=circle, draw=black, fill={rgb,255: red,220; green,0; blue,0}]
\tikzstyle{X phase dot}=[Z phase dot, tikzit shape=circle, tikzit draw=blue, fill={rgb,255: red,220; green,0; blue,0}, font={\footnotesize\color{white}\boldmath}]
\tikzstyle{hadamard}=[fill=yellow, draw=black, shape=rectangle, inner sep=0.6mm, minimum height=1.5mm, minimum width=1.5mm]
\tikzstyle{small hadamard}=[hadamard]
\tikzstyle{vertex}=[inner sep=0mm, minimum size=3pt, shape=circle, draw=black, fill=black]
\tikzstyle{vertex set}=[inner sep=0mm, minimum size=1mm, shape=circle, draw=black, fill=white, font={\footnotesize\boldmath}]
\tikzstyle{new style 0}=[draw=term, fill={term!10}, rounded corners, minimum width=25pt, minimum height=30pt]
\tikzstyle{new style 2}=[draw=term, fill={term!10}, rounded corners, minimum width=25pt, minimum height=30pt]
\tikzstyle{round box}=[fill=white, draw=black, shape=rectangle, rounded corners=5pt, font={\small}, minimum height=6mm, minimum width=6mm]
\tikzstyle{e-box}=[fill=white, draw=black, shape=rectangle, minimum width=20mm, minimum height=20mm]
\tikzstyle{empty diag string}=[fill=white, draw={rgb,255: red,165; green,165; blue,165}, shape=rectangle, minimum size=1.2 cm, dashed, thick]
\tikzstyle{empty diag red}=[fill=white, draw={rgb,255: red,237; green,105; blue,94}, shape=rectangle, minimum size=1.2 cm, dashed, thick]
\tikzstyle{empty diag yellow}=[fill=white, draw={rgb,255: red,245; green,194; blue,81}, shape=rectangle, minimum size=1.2 cm, dashed, thick]
\tikzstyle{empty diag blue}=[fill=white, draw={rgb,255: red,99; green,182; blue,240}, shape=rectangle, minimum size=1.2 cm, dashed, thick]
\tikzstyle{empty diag green}=[fill=white, draw={rgb,255: red,82; green,158; blue,86}, shape=rectangle, minimum size=1.2 cm, dashed, thick]
\tikzstyle{empty diag black}=[fill=white, draw=black, shape=rectangle, minimum size=1.2 cm, dashed, thick]
\tikzstyle{orange_red_v}=[fill={rgb,255: red,237; green,19; blue,90}, draw={rgb,255: red,237; green,19; blue,90}, shape=circle]
\tikzstyle{turquoise_v}=[fill={rgb,255: red,0; green,180; blue,206}, draw={rgb,255: red,0; green,180; blue,206}, shape=circle]
\tikzstyle{forest_v}=[fill={rgb,255: red,0; green,155; blue,85}, draw={rgb,255: red,0; green,155; blue,85}, shape=circle]
\tikzstyle{new style 1}=[fill={rgb,255: red,128; green,128; blue,128}, draw={rgb,255: red,128; green,128; blue,128}, shape=circle]
\tikzstyle{tophalfcircle}=[draw=none,minimum size=1cm,inner sep=0pt,outer sep=0pt]

\tikzstyle{new edge style 0}=[-, draw=black, line width=0.5pt]
\tikzstyle{blue 0}=[-, draw=blue, line width=0.8pt]
\tikzstyle{red 0}=[-, draw=red, line width=0.8pt]
\tikzstyle{black dash}=[-, color=black, dashed, dash pattern=on 1.5pt off 1.5pt, draw=black, line width=0.8pt]
\tikzstyle{red dash}=[-, color=red, dashed, dash pattern=on 1pt off 1.5pt, draw=red, line width=0.8pt, line cap=round]
\tikzstyle{blue dash}=[-, color=blue, dashed, dash pattern=on 1pt off 1.5pt, draw=blue, line width=0.8pt, line cap=round]
\tikzstyle{red term dash}=[-, dotted, draw={rgb,255: red,210; green,0; blue,0}, dash pattern=on 1pt off 2pt, line width=0.8pt, line cap=round]
\tikzstyle{brace edge}=[-, tikzit draw=blue, decorate, decoration={brace,amplitude=1mm,raise=-1mm}, line width=0.8pt]
\tikzstyle{gray}=[-, draw={rgb,255: red,191; green,191; blue,191}]
\tikzstyle{arrow}=[<-, draw={rgb,255: red,128; green,128; blue,128}]
\tikzstyle{double-arrow}=[draw={rgb,255: red,128; green,128; blue,128}, <->]
\tikzstyle{thick line}=[-, line width=0.8pt]
\tikzstyle{new edge style 1}=[draw=black, fill=none, ->, >=latex, line width=0.5pt]
\tikzstyle{new edge style 2}=[-, fill=gray]
\tikzstyle{grey identity}=[-, draw={rgb,255: red,191; green,191; blue,191}, color=gray, dashed, dash pattern=on 1.5pt off 1.5pt, draw=black, line width=0.8pt]
\tikzstyle{lambda_unit_string}=[-, draw={rgb,255: red,0; green,0; blue,0}, fill={rgb,255: red,238; green,238; blue,255}, thick]
\tikzstyle{fancy red dash}=[-, draw={rgb,255: red,237; green,105; blue,94}, dashed, dash pattern=on 1pt off 1.5pt, draw={rgb,255: red,237; green,105; blue,94}, line width=0.8pt, line cap=round]
\tikzstyle{fancy yellow dash}=[-, draw={rgb,255: red,245; green,194; blue,81}, dashed, dash pattern=on 1pt off 1.5pt, draw={rgb,255: red,245; green,194; blue,81}, line width=0.8pt, line cap=round]
\tikzstyle{fancy blue dash}=[-, draw={rgb,255: red,99; green,182; blue,240}, dashed, dash pattern=on 1pt off 1.5pt, draw={rgb,255: red,99; green,182; blue,240}, line width=0.8pt, line cap=round]
\tikzstyle{fancy green dash}=[-, draw={rgb,255: red,82; green,158; blue,86}, dashed, dash pattern=on 1pt off 1.5pt, draw={rgb,255: red,82; green,158; blue,86}, line width=0.8pt, line cap=round]
\tikzstyle{slotted_e_graph_header}=[-, fill=lightgray, draw=black, line width=0.75pt]
\tikzstyle{orange_red_e}=[-, draw={rgb,255: red,237; green,19; blue,90}]
\tikzstyle{turquoise_e}=[-, draw={rgb,255: red,0; green,180; blue,206}]
\tikzstyle{forest_e}=[-, draw={rgb,255: red,0; green,155; blue,85}]

\tikzstyle{hedge}=[fill=white, draw=black, shape=rectangle, rounded corners=2mm, inner sep=0.2mm, outer sep=-2mm, scale=0.8, minimum height=8mm, minimum width=8mm, tikzit category=hypergraph]
\tikzstyle{hedge interface}=[fill=white, draw=black, shape=rectangle, rounded corners=1.5mm, inner sep=0.2mm, outer sep=-2mm, scale=0.8, minimum height=5mm, minimum width=5mm, tikzit category=hypergraph]
\tikzstyle{hedge blue}=[hedge, fill={rgb,255: red,102; green,204; blue,255}, draw=black, shape=rectangle, tikzit category=hypergraph]
\tikzstyle{node}=[fill=black, draw=black, shape=circle, minimum size=1.5mm, inner sep=0mm, tikzit category=hypergraph]
\tikzstyle{red node}=[fill=red, draw=black, shape=circle, minimum size=1.5mm, inner sep=0mm, tikzit category=hypergraph]
\tikzstyle{blue node}=[fill=blue, draw=black, shape=circle, minimum size=1.5mm, inner sep=0mm, tikzit category=hypergraph]
\tikzstyle{green node}=[fill={rgb,255: red,33; green,140; blue,33}, draw=black, shape=circle, minimum size=1.5mm, inner sep=0mm, tikzit category=hypergraph]
\tikzstyle{node highlight}=[fill=black, draw=blue, thick, shape=circle, minimum size=1.5mm, inner sep=0mm, tikzit category=hypergraph]
\tikzstyle{red node highlight}=[fill=red, draw=blue, thick, shape=circle, minimum size=1.5mm, inner sep=0mm, tikzit category=hypergraph]
\tikzstyle{yellow hedge}=[hedge, fill=yellow, draw=black, shape=rectangle, tikzit category=hypergraph]
\tikzstyle{green hedge}=[hedge, fill=green, draw=black, shape=rectangle, tikzit category=hypergraph]
\tikzstyle{small box}=[fill=white, draw=black, shape=rectangle, minimum height=6mm, minimum width=6mm, tikzit category=string diagram]
\tikzstyle{vsmall box}=[fill=black, draw=black, shape=rectangle, minimum height=4mm, minimum width=1mm, tikzit category=string diagram, inner sep=0]
\tikzstyle{medium box}=[fill=white, draw=black, shape=rectangle, minimum height=11mm, minimum width=6mm, tikzit category=string diagram]
\tikzstyle{semilarge box}=[fill=white, draw=black, shape=rectangle, minimum height=16mm, minimum width=6mm, tikzit category=string diagram]
\tikzstyle{large box}=[fill=white, draw=black, shape=rectangle, minimum height=21mm, minimum width=6mm, tikzit category=string diagram]
\tikzstyle{black dot}=[fill=black, draw=black, shape=circle, minimum size=2mm, inner sep=0mm, tikzit category=string diagram]
\tikzstyle{white dot}=[fill=white, draw=black, shape=circle, minimum size=2mm, inner sep=0mm, tikzit category=string diagram]
\tikzstyle{red dot}=[fill=red, draw=black, shape=circle, minimum size=2mm, inner sep=0mm, tikzit category=string diagram]
\tikzstyle{wlabel}=[fill=none, draw=none, shape=rectangle, tikzit category=string diagram, font={\footnotesize}, inner sep=0pt, tikzit fill={rgb,255: red,102; green,204; blue,255}, tikzit draw={rgb,255: red,102; green,204; blue,255}, yshift=0.3mm]
\tikzstyle{BRchange}=[draw=black, shape=diamond, tikzit shape=circle, tikzit fill={rgb,255: red,96; green,0; blue,0}, diamond split part fill={black,red}, inner sep=-5mm, minimum width=2.7mm, minimum height=1.7mm]
\tikzstyle{RBchange}=[draw=black, shape=diamond, tikzit shape=circle, tikzit fill={rgb,255: red,165; green,0; blue,0}, diamond split part fill={red,black}, inner sep=0, minimum width=2.7mm, minimum height=1.7mm]
\tikzstyle{dummy}=[fill=none, draw=none, shape=circle, font={\small}, inner sep=1pt, tikzit draw=blue, tikzit fill=white]
\tikzstyle{node label}=[fill=none, draw=none, shape=rectangle, tikzit fill=cyan, tikzit draw=cyan, font={\scriptsize}, tikzit shape=circle, inner sep=0pt]
\tikzstyle{empty diag}=[fill=white, draw={rgb,255: red,165; green,165; blue,165}, shape=rectangle, minimum size=1.2 cm, dashed, thick]
\tikzstyle{large horizontal box}=[fill=white, draw=black, shape=rectangle, minimum height=6mm, minimum width=21mm, tikzit category=string diagram]
\tikzstyle{medium horizontal box}=[fill=white, draw=black, shape=rectangle, minimum height=6mm, minimum width=11mm, tikzit category=string diagram]
\tikzstyle{semilarge horizontal box}=[fill=white, draw=black, shape=rectangle, minimum height=6mm, minimum width=16mm, tikzit category=string diagram]
\tikzstyle{e-box-interface}=[inner sep=0mm, minimum size=1mm, shape=circle, draw=black, fill=white, font={\footnotesize\boldmath}]
\tikzstyle{very large horizontal box}=[fill=white, draw=black, shape=rectangle, minimum height=6mm, minimum width=41mm, tikzit category=string diagram]

\tikzstyle{dashed edge}=[-, dashed, very thick]
\tikzstyle{alt sort}=[-, dashed, dash pattern=on 2pt off 0.5pt, thick, draw=red]
\tikzstyle{diredge}=[->, >={Latex[length=1.5mm]}]
\tikzstyle{diredge highlight}=[->, >={Latex[length=1.5mm]}, draw=blue, thick]
\tikzstyle{boundary frame}=[-, draw={rgb,255: red,170; green,170; blue,255}, dashed, fill={rgb,255: red,238; green,238; blue,255}, thick, dash pattern=on 2pt off 0.5pt]
\tikzstyle{graph frame}=[-, draw={rgb,255: red,191; green,191; blue,191}, dashed, fill={rgb,255: red,238; green,238; blue,238}, thick, dash pattern=on 2pt off 0.5pt]
\tikzstyle{def sort}=[-]
\tikzstyle{component}=[-, draw=red, thick]
\tikzstyle{map edge}=[{|->}, >=latex, shorten <=0.5mm, shorten >=0.5mm]
\tikzstyle{hypergraph map edge}=[{|->}, draw=red, shorten <=1mm, shorten >=1mm]
\tikzstyle{cdedge}=[->]
\tikzstyle{big cdedge}=[->, very thick, >=latex]
\tikzstyle{pointer edge}=[->, draw=gray, thick]
\tikzstyle{vertical_delimiter}=[-, dashed, very thick, fill=none, draw={rgb,255: red,170; green,170; blue,255}]
\tikzstyle{e_hyperedge}=[-, draw=black, dashed, fill={rgb,255: red,255; green,255; blue,255}, thick, dash pattern=on 3pt off 1.5pt]
\tikzstyle{lambda box}=[-, draw=black, fill={rgb,255: red,255; green,255; blue,255}, thick]
\tikzstyle{lambda_unit}=[-, draw={rgb,255: red,0; green,0; blue,0},fill={rgb,255: red,238; green,238; blue,255}, thick]
\input{hypergraph.tikzdefs}

\AtBeginDocument{
  \setlength{\textfloatsep}{1pt}
  \setlength{\intextsep}{1pt}
}
\makeatletter
\def\@oddhead{}
\def\@evenhead{}
\makeatother

\author{\begin{tabular}{ccc}
  Aleksei Tiurin & Dan R. Ghica & Nick Hu\\
  {\small University of Birmingham} & {\small University of Birmingham} & {\small Kyoto University}\\
  &{\small and}&\\
  &{\small Huawei Central Software Institute}&
\end{tabular}\\
\thanks{This work was supported by EPSRC grants EP/V001612/1 and EP/Y010035/1.}
}
\title{Categorical E-Graphs for Lambda Calculi}
\def\titlerunning{Categorical E-Graphs for Lambda Calculi}
\def\authorrunning{Tiurin, Ghica \& Hu}

\begin{document}
\maketitle

\begin{abstract}
	Equality saturation is a powerful technique for program optimisation and reasoning, driven by the use of equivalence classes of terms under rewrite rules.
	These equivalence classes lie at the root of data structures called equality graphs (e-graphs).
	Despite their numerous advantages, until recently e-graphs lacked native support for variable binding, limiting their applicability to programming languages.
	We propose to address this problem from a categorical perspective by extending the interpretation of e-graphs as string diagrams, namely morphisms in symmetric monoidal semilattice-enriched categories which we additionally equip with a monoidal closed structure.
	We further define a concrete representation using hierarchical hypergraphs, and introduce a corresponding double-pushout (DPO) rewriting system.
	Finally, we establish the equivalence between term rewriting and DPO rewriting, with the combinatorial model inherently absorbing the equations of the symmetric monoidal structure.
	Our approach, specifically designed for lambda calculi, is compared and contrasted with slotted e-graphs — an alternative method for incorporating variable binding within e-graphs.
\end{abstract}
\section{Introduction}%
\label{sec:introduction}

\subsection{Motivation and contribution}

E-graphs (short for \emph{equality graphs})~\cite{EggPaper} are an efficient data structure for non-destructive rewriting, incorporating equivalence representation, used in automated theorem proving, compiler optimisation, and program analysis.
They mitigate the \emph{phase-ordering problem} of term rewriting: the way rewrites are scheduled may enable or block further rewrites due to introduction or elimination of redexes, resulting in a suboptimal result.
This is achieved by making rewrites on e-graphs non-destructive: an application of a rewrite rule, instead of replacing the left-hand side by the right-hand side, adds the latter to the e-graph so that all pre-existing redexes are protected and new ones are introduced.
Exponential growth of the e-graph is avoided by the clever use of sharing of nodes.

E-graphs were first introduced in the automated theorem proving literature in the 1980s~\cite{nelson1980techniques}, but recent years have seen renewed interest, particularly for optimisation and synthesis using the technique of \emph{equality saturation}~\cite{EqualitySaturation2009, flatt_small_2022, EggPaper}, wherein rewrites to an e-graph are repeated until a fixpoint (or timeout) is reached.
We are interested in using e-graphs in the context of functional programming languages, which builds upon an internal representation of the $\lambda$-calculus.
K{\oe}hler et al.~\cite{koehler2022sketchguided} pioneered work in this area by identifying two key, interconnected challenges—handling binding and ensuring capture-avoiding substitution—when applying equality saturation to an encoding of the $\lambda$-calculus in e-graphs.

In the textual representation of $\lambda$-terms, $\alpha$-equivalent terms may be distinct and thus cannot be shared.
De Bruijn indices~\cite{de1972lambda} syntactically identify $\alpha$-equivalent terms, but the same variable can receive different indices depending on its scope, again interfering with sharing.
As a simple example, $\lambda f . f ((\lambda x . f x) 2)$ is encoded as $\lambda . \%0 ((\lambda . \%1 \%0) 2)$, where the two occurrences of $f$ are represented differently and cannot be shared.

A simple but theoretically~\cite{AbadiExplicitSubst} and practically~\cite{anf} important extension of the $\lambda$-calculus consists of adding let-bindings or explicit substitutions to the syntax.
Terms are identified up to permutations of non-interacting substitutions, called \enquote{graph equivalence}~\cite{accattoli2014nonstandard}.
Modelling this in conventional e-graphs requires additional nodes for explicit substitutions and additional rewrite rules to capture graphical equivalence, polluting the equality saturation state space with bureaucracy rather than the domain-specific rewriting one actually cares about.

All these issues can be resolved by the use of \emph{string diagrams} in place of textual syntax.
Tiurin et al.~\cite{tiurin2025equivalencehypergraphsdporewriting} shows how e-graphs (without binding) can be given a categorical semantics using symmetric monoidal semilattice-enriched categories, which allows e-graphs to be reconstructed as string diagrams.
In this paper, we extend this technique to e-graphs with bindings, lifting to monoidal \emph{closed} categories.

\subsection{String diagrams}

String diagrams provide a topological calculus for monoidal categories, representing objects as \emph{strings} and morphisms as \emph{nodes}.
They bridge category theory---which gives semantic models to programming languages and logics---to the concrete language of graphs fundamental in compiler design; see Ghica and Zanasi~\cite{ghica2024stringdiagramslambdacalculifunctional} for the $\lambda$-calculus setting and Tiurin et al.~\cite{tiurin2025equivalencehypergraphsdporewriting} for e-graphs without bindings.
In our string diagrams, which are oriented left-to-right and bottom-to-top, the structure of symmetric monoidal categories is absorbed into the graphical syntax, with composition and tensoring represented by juxtaposition and symmetries as weaving of strings.

We require a concrete encoding of string diagrams, analogous to the de Bruijn index representation of $\lambda$-terms.
For this we use \emph{hypergraphs}~\cite{bonchi_string_2022-1}: hypergraph isomorphism rigorously expresses when string diagrams \enquote{look the same} topologically, and string-diagrammatic equational reasoning is formalised as DPO rewriting over these hypergraphs~\cite{bonchi_string_2022-1}.
For variable binding and equivalence classes we rely on \emph{functorial boxes}~\cite{10.1007/11874683_1}, represented via \emph{hierarchical hypergraphs}~\cite{fscd}.

Revisiting the earlier example in~\autoref{fig:de-brujin-string}, observe that the two occurrences of $f$ are now shared.
The scope of a $\lambda$-abstraction is a rounded box with the bound variable as a wire on the frame, the application node a half-circle
\adjustbox{scale=0.15}{\begin{tikzpicture}
		\begin{pgfonlayer}{nodelayer}
			\node [style=none] (0) at (1, 6.75) {};
			\node [style=none] (1) at (1.25, 7.75) {};
			\node [style=none] (2) at (2, 8) {};
			\node [style=none] (3) at (2.75, 7.75) {};
			\node [style=none] (4) at (3, 6.75) {};
			\node [style=none] (5) at (1.5, 6.75) {};
			\node [style=none] (6) at (2.5, 6.75) {};
			\node [style=none] (7) at (2, 9) {};
			\node [style=none] (8) at (1.5, 5.75) {};
			\node [style=none] (9) at (2.5, 5.75) {};
		\end{pgfonlayer}
		\begin{pgfonlayer}{edgelayer}
			\draw [style=lambda_unit_string] (3.center)
			to [in=60, out=-30, looseness=0.75] (4.center)
			to (0.center)
			to [in=-150, out=120, looseness=0.75] (1.center)
			to [in=-180, out=45, looseness=0.75] (2.center)
			to [in=135, out=0, looseness=0.75] cycle;
			\draw (2.center) to (7.center);
			\draw (5.center) to (8.center);
			\draw (9.center) to (6.center);
		\end{pgfonlayer}
	\end{tikzpicture}
},
and sharing a \emph{copy node} \comonoid.

\subsection{E-graphs}

E-graphs simultaneously represent several equivalent terms of an algebraic theory by generalising \emph{term graphs} to include equivalence classes of subterms.
\emph{E-classes} (dashed boxes) contain \emph{e-nodes} representing equivalent terms with the same root; e-nodes reference e-classes rather than other e-nodes.
Extracting a term amounts to choosing one e-node per e-class; different choices yield equivalent terms.
\autoref{fig:e-graph-example} shows how a series of rewrite rules are non-destructively applied to an e-graph~\cite{EggPaper}, alongside the equivalent string diagram representation using \emph{e-hypergraphs}~\cite{tiurin2025equivalencehypergraphsdporewriting}.
Applying a rewrite rule to an e-graph amounts to adding new e-nodes into the e-graph and \emph{merging} e-classes for the left-hand side and right-hand side of the rewrite rule, instantiating the variables in the rule with the e-classes according to the match.
The key difference that can be noted in string-diagrammatic representation, is that the dashed boxes that correspond to e-graph e-classes are just ordinary nodes that happen to have hierarchical structure.
These boxes have enough inputs to cater for each subdiagram inside, explicitly deleting the unnecessary wires inside each component using $\counit$.
\begin{figure*}[t!]
	\begin{minipage}{0.18\textwidth}
		\centering
		\adjustbox{scale=0.55}{
			\begin{tikzpicture}
	\begin{pgfonlayer}{nodelayer}
		\node [style=none] (0) at (-0.25, 6.75) {};
		\node [style=none] (2) at (0.5, 7.5) {};
		\node [style=none] (4) at (1.25, 6.75) {};
		\node [style=none] (5) at (0.25, 6.75) {};
		\node [style=none] (6) at (0.75, 6.75) {};
		\node [style=none] (7) at (0.5, 8.5) {};
		\node [style=none] (8) at (0.25, 6.5) {};
		\node [style=none] (9) at (0.75, 6.25) {};
		\node [style=none] (10) at (-1, 8) {};
		\node [style=none] (11) at (-0.5, 8.5) {};
		\node [style=none] (12) at (1.5, 8.5) {};
		\node [style=none] (13) at (2, 8) {};
		\node [style=none] (14) at (2, 5.5) {};
		\node [style=none] (15) at (1.5, 5) {};
		\node [style=none] (16) at (-0.5, 5) {};
		\node [style=none] (17) at (-1, 5.5) {};
		\node [style=none] (18) at (0.25, 5) {};
		\node [style=none] (27) at (-2, 10.5) {};
		\node [style=none] (28) at (2.5, 13.25) {};
		\node [style=none] (29) at (-3.25, 12.5) {};
		\node [style=none] (30) at (-2.75, 13.25) {};
		\node [style=none] (31) at (3.5, 13.25) {};
		\node [style=none] (32) at (4, 12.75) {};
		\node [style=none] (33) at (4, 4) {};
		\node [style=none] (34) at (3.5, 3.25) {};
		\node [style=none] (35) at (-2.75, 3.25) {};
		\node [style=none] (36) at (-3.25, 3.75) {};
		\node [style=vertex] (37) at (-2, 4.5) {};
		\node [style=none] (45) at (2.5, 9.25) {};
		\node [style=none] (46) at (3, 9.25) {};
		\node [style=round box] (47) at (3.25, 7.5) {$2$};
		\node [style=none] (48) at (2.5, 15) {};
		\node [style=none] (49) at (2, 10) {};
		\node [style=none] (51) at (2.75, 10.75) {};
		\node [style=none] (53) at (3.5, 10) {};
		\node [style=none] (54) at (2.5, 10) {};
		\node [style=none] (55) at (3, 10) {};
		\node [style=none] (56) at (1.75, 11.75) {};
		\node [style=none] (58) at (2.5, 12.5) {};
		\node [style=none] (60) at (3.25, 11.75) {};
		\node [style=none] (61) at (2.25, 11.75) {};
		\node [style=none] (62) at (2.75, 11.75) {};
		\node [style=none] (63) at (2.25, 11.25) {};
	\end{pgfonlayer}
	\begin{pgfonlayer}{edgelayer}
		\draw [style={lambda_unit_string}] (0.center)
			 to [in=-180, out=90] (2.center)
			 to [in=90, out=0] (4.center)
			 to cycle;
		\draw (2.center) to (7.center);
		\draw (5.center) to (8.center);
		\draw (6.center) to (9.center);
		\draw [in=-180, out=90, looseness=1.75] (10.center) to (11.center);
		\draw (11.center) to (12.center);
		\draw [in=90, out=0, looseness=1.75] (12.center) to (13.center);
		\draw (13.center) to (14.center);
		\draw [in=0, out=-90, looseness=2.25] (14.center) to (15.center);
		\draw (15.center) to (16.center);
		\draw [in=-90, out=180, looseness=1.75] (16.center) to (17.center);
		\draw (17.center) to (10.center);
		\draw [in=165, out=-90, looseness=1.25] (9.center) to (14.center);
		\draw (8.center) to (18.center);
		\draw (30.center) to (31.center);
		\draw [in=90, out=-180, looseness=1.25] (30.center) to (29.center);
		\draw (29.center) to (36.center);
		\draw [in=180, out=-90, looseness=1.75] (36.center) to (35.center);
		\draw [in=90, out=0, looseness=1.75] (31.center) to (32.center);
		\draw (32.center) to (33.center);
		\draw [in=0, out=-90, looseness=1.75] (33.center) to (34.center);
		\draw (34.center) to (35.center);
		\draw [in=-90, out=180, looseness=0.75] (33.center) to (37);
		\draw [in=270, out=-30, looseness=0.75] (37) to (18.center);
		\draw [in=90, out=-90, looseness=0.75] (27.center) to (37);
		\draw [in=90, out=-90] (46.center) to (47);
		\draw [in=90, out=-90, looseness=0.75] (45.center) to (7.center);
		\draw (48.center) to (28.center);
		\draw [style={lambda_unit_string}] (51.center)
			 to [in=90, out=0] (53.center)
			 to (49.center)
			 to [in=-180, out=90] cycle;
		\draw [style={lambda_unit_string}] (56.center)
			 to [in=-180, out=90] (58.center)
			 to [in=90, out=0] (60.center)
			 to cycle;
		\draw (45.center) to (54.center);
		\draw (46.center) to (55.center);
		\draw (28.center) to (58.center);
		\draw (62.center) to (51.center);
		\draw (61.center) to (63.center);
		\draw [in=-90, out=90, looseness=0.25] (27.center) to (63.center);
	\end{pgfonlayer}
\end{tikzpicture}
		}
		\caption{String-dia\-gram\-matic representation of $\lambda f . f ((\lambda x . f x) 2)$.}
		\label{fig:de-brujin-string}
	\end{minipage}
	\hfill
	\begin{minipage}{0.8\textwidth}
		\begin{center}
			\adjustbox{scale=0.55}{
				\input{./figures/e-graph-example-2.tikz}
			}
		\end{center}
		\caption{E-graph example (top) and its equivalent string diagram representation (bottom)~\cite{tiurin2025equivalencehypergraphsdporewriting}}
		\label{fig:e-graph-example}
	\end{minipage}
\end{figure*}

\subsection{Conventional vs.\ slotted vs.\ categorical e-graphs with binding}%
\label{sec:vs-e-graphs-with-binding}

A common approach to $\beta$-reduction in e-graphs uses special explicit-substitution nodes with associated rewrite rules~\cite{EggPaper,koehler2022sketchguided}.
In our categorical representation these nodes are unnecessary, since explicit substitution is modelled as composition.
\autoref{fig:e-graph-substitution} illustrates successive $\beta$-reductions in a conventional e-graph for $(\lambda x . (y + y) + x) 1$; \autoref{fig:e-string-substitution} shows the same reduction via string diagrams, where $\beta$-reduction is simply re-wiring when an application node meets an abstraction.
Using de Bruijn indices instead would only push the problem elsewhere: it greatly increases the complexity of encoding the reduction as a rewrite rule, requiring meta-syntactic shifting operators on the right-hand side~\cite{koehler2022sketchguided}.

\begin{wrapfigure}{r}{0.5\textwidth}
	\vspace{-8pt}
	\captionsetup{skip=2pt}
	\begin{subfigure}[b]{0.48\linewidth}
		\centering
		\adjustbox{scale=0.55}{
			\begin{tikzpicture}
	\begin{pgfonlayer}{nodelayer}
		\node [style=round box] (0) at (0.5, -0.25) {var};
		\node [style=none] (5) at (-1, -0.75) {};
		\node [style=none] (6) at (-0.5, -1.25) {};
		\node [style=none] (7) at (1.5, -1.25) {};
		\node [style=none] (8) at (2, -0.75) {};
		\node [style=none] (1) at (-1, 1.25) {};
		\node [style=none] (2) at (-0.5, 1.75) {};
		\node [style=none] (3) at (1.5, 1.75) {};
		\node [style=none] (4) at (2, 1.25) {};
		\node [style=none] (9) at (-1, 0.75) {};
		\node [style=none] (10) at (2, 0.75) {};
		\node [style=none] (11) at (0.5, 1.25) {$s_1$};
		\node [style=round box] (12) at (0.5, 4.5) {$+$};
		\node [style=none] (13) at (-1, 4) {};
		\node [style=none] (14) at (-0.5, 3.5) {};
		\node [style=none] (15) at (1.5, 3.5) {};
		\node [style=none] (16) at (2, 4) {};
		\node [style=none] (17) at (-1, 6) {};
		\node [style=none] (18) at (-0.5, 6.5) {};
		\node [style=none] (19) at (1.5, 6.5) {};
		\node [style=none] (20) at (2, 6) {};
		\node [style=none] (21) at (-1, 5.5) {};
		\node [style=none] (22) at (2, 5.5) {};
		\node [style=none] (23) at (0.5, 6) {$s_2, s_3 $};
		\node [style=round box] (24) at (0.5, 14.5) {$\lambda\; x$};
		\node [style=none] (25) at (-1, 14) {};
		\node [style=none] (26) at (-0.5, 13.5) {};
		\node [style=none] (27) at (1.5, 13.5) {};
		\node [style=none] (28) at (2, 14) {};
		\node [style=none] (29) at (-1, 16) {};
		\node [style=none] (30) at (-0.5, 16.5) {};
		\node [style=none] (31) at (1.5, 16.5) {};
		\node [style=none] (32) at (2, 16) {};
		\node [style=none] (33) at (-1, 15.5) {};
		\node [style=none] (34) at (2, 15.5) {};
		\node [style=none] (35) at (0.5, 16) {$a,b,c$};
		\node [style=none] (36) at (0.25, 4.25) {};
		\node [style=none] (37) at (0.75, 4.25) {};
		\node [style=none] (38) at (0, 1.75) {};
		\node [style=none] (39) at (1, 1.75) {};
		\node [style=none] (40) at (-0.25, 2.75) {$s_2$};
		\node [style=none] (41) at (1.25, 2.75) {$s_3$};
		\node [style=none] (42) at (0, 6.5) {};
		\node [style=none] (44) at (-1.25, 7.75) {$(s_4,s_5)$};
		\node [style=none] (45) at (0.5, 14) {};
		\node [style=none] (46) at (1, 6.5) {};
		\node [style=none] (47) at (2.25, 7.75) {$(s_6, s_7)$};
		\node [style=round box] (48) at (0.5, 9.75) {$+$};
		\node [style=none] (49) at (-1.25, 9.25) {};
		\node [style=none] (50) at (-0.75, 8.75) {};
		\node [style=none] (51) at (1.75, 8.75) {};
		\node [style=none] (52) at (2.25, 9.25) {};
		\node [style=none] (53) at (-1.25, 11.25) {};
		\node [style=none] (54) at (-0.75, 11.75) {};
		\node [style=none] (55) at (1.75, 11.75) {};
		\node [style=none] (56) at (2.25, 11.25) {};
		\node [style=none] (57) at (-1.25, 10.75) {};
		\node [style=none] (58) at (2.25, 10.75) {};
		\node [style=none] (59) at (0.5, 11.25) {$s_4, s_5, s_6, s_7$};
		\node [style=none] (60) at (0.25, 9.5) {};
		\node [style=none] (61) at (0.75, 9.5) {};
		\node [style=none] (62) at (0.5, 11.75) {};
		\node [style=none] (63) at (2, 12.75) {$(a,b,c,x)$};
		\node [style=none] (66) at (1.65, -1) {$e_0$};
		\node [style=none] (67) at (1.65, 3.75) {$e_1$};
		\node [style=none] (68) at (1.85, 9) {$e_2$};
		\node [style=none] (69) at (1.65, 13.75) {$e_3$};
	\end{pgfonlayer}
	\begin{pgfonlayer}{edgelayer}
		\draw [style=black dash, in=0, out=-90, looseness=1.75] (8.center) to (7.center);
		\draw [style=black dash] (7.center) to (6.center);
		\draw [style=black dash, in=-90, out=-180, looseness=1.75] (6.center) to (5.center);
		\draw [style=black dash] (9.center) to (5.center);
		\draw [style=black dash] (10.center) to (8.center);
		\draw [style={slotted_e_graph_header}] (1.center)
			 to [in=180, out=90, looseness=1.75] (2.center)
			 to (3.center)
			 to [in=90, out=0, looseness=1.75] (4.center)
			 to (10.center)
			 to (9.center)
			 to cycle;
		\draw [style=black dash, in=0, out=-90, looseness=1.75] (16.center) to (15.center);
		\draw [style=black dash] (15.center) to (14.center);
		\draw [style=black dash, in=-90, out=-180, looseness=1.75] (14.center) to (13.center);
		\draw [style=black dash] (21.center) to (13.center);
		\draw [style=black dash] (22.center) to (16.center);
		\draw [style={slotted_e_graph_header}] (17.center)
			 to [in=180, out=90, looseness=1.75] (18.center)
			 to (19.center)
			 to [in=90, out=0, looseness=1.75] (20.center)
			 to (22.center)
			 to (21.center)
			 to cycle;
		\draw [style=black dash, in=0, out=-90, looseness=1.75] (28.center) to (27.center);
		\draw [style=black dash] (27.center) to (26.center);
		\draw [style=black dash, in=-90, out=-180, looseness=1.75] (26.center) to (25.center);
		\draw [style=black dash] (33.center) to (25.center);
		\draw [style=black dash] (34.center) to (28.center);
		\draw [style={slotted_e_graph_header}] (34.center)
			 to (33.center)
			 to (29.center)
			 to [in=180, out=90, looseness=1.75] (30.center)
			 to (31.center)
			 to [in=90, out=0, looseness=1.75] (32.center)
			 to cycle;
		\draw [style=new edge style 1] (36.center) to (38.center);
		\draw [style=new edge style 1] (37.center) to (39.center);
		\draw [style=black dash, in=0, out=-90, looseness=1.75] (52.center) to (51.center);
		\draw [style=black dash] (51.center) to (50.center);
		\draw [style=black dash, in=-90, out=-180, looseness=1.75] (50.center) to (49.center);
		\draw [style=black dash] (57.center) to (49.center);
		\draw [style=black dash] (58.center) to (52.center);
		\draw [style={slotted_e_graph_header}] (54.center)
			 to (55.center)
			 to [in=90, out=0, looseness=1.75] (56.center)
			 to (58.center)
			 to (57.center)
			 to (53.center)
			 to [in=180, out=90, looseness=1.75] cycle;
		\draw [style=new edge style 1] (45.center) to (62.center);
		\draw [style=new edge style 1] (60.center) to (42.center);
		\draw [style=new edge style 1] (61.center) to (46.center);
	\end{pgfonlayer}
\end{tikzpicture}
		}
		\caption{Slotted e-graph}%
		\label{fig:slotted}
	\end{subfigure}
	\hfill
	\begin{subfigure}[b]{0.48\linewidth}
		\centering
		\adjustbox{scale=0.45}{
			\begin{tikzpicture}
	\begin{pgfonlayer}{nodelayer}
		\node [style=round box] (0) at (4.5, -1.75) {$+$};
		\node [style=none] (1) at (5.75, -3.25) {};
		\node [style=none] (2) at (4.25, -4.25) {};
		\node [style=none] (3) at (4.25, -2.25) {};
		\node [style=none] (4) at (4.75, -2.25) {};
		\node [style=none] (5) at (5.75, -3.25) {};
		\node [style=none] (6) at (5.25, -3.75) {};
		\node [style=none] (7) at (3.75, -3.75) {};
		\node [style=none] (8) at (3.25, -3.25) {};
		\node [style=none] (9) at (4.5, 0) {};
		\node [style=none] (10) at (4.5, -1.25) {};
		\node [style=none] (11) at (3.25, -0.5) {};
		\node [style=none] (12) at (3.75, 0) {};
		\node [style=none] (13) at (5.75, -0.5) {};
		\node [style=none] (14) at (5.25, 0) {};
		\node [style=none] (15) at (6.25, -5) {};
		\node [style=none] (16) at (5.75, -5.5) {};
		\node [style=none] (17) at (3.25, -5.5) {};
		\node [style=none] (18) at (2.75, -5) {};
		\node [style=none] (19) at (2.75, 0.5) {};
		\node [style=none] (20) at (3.25, 1) {};
		\node [style=none] (21) at (5.75, 1) {};
		\node [style=none] (22) at (6.25, 0.5) {};
		\node [style=vertex] (23) at (4.5, 2) {};
		\node [style=none] (24) at (-1.25, 2) {};
		\node [style=none] (25) at (-2.75, 2) {};
		\node [style=none] (28) at (-2, 2.75) {};
		\node [style=none] (29) at (-2.25, 2) {};
		\node [style=none] (30) at (-1.75, 2) {};
		\node [style=none] (31) at (-1.75, -5.5) {};
		\node [style=none] (32) at (-1.75, 0) {};
		\node [style=none] (33) at (-1, 4.25) {};
		\node [style=none] (34) at (-2.5, 4.25) {};
		\node [style=none] (37) at (-1.75, 5) {};
		\node [style=none] (38) at (-2, 4.25) {};
		\node [style=none] (39) at (-1.5, 4.25) {};
		\node [style=none] (40) at (1, -5.5) {};
		\node [style=none] (41) at (1, 1.25) {};
		\node [style=none] (42) at (-1.75, 6) {};
		\node [style=none] (43) at (0.25, -1.75) {$b$};
		\node [style=none] (44) at (-2.75, -1.75) {$a$};
		\node [style=none] (45) at (7.75, -5.5) {};
		\node [style=none] (46) at (8.25, 4) {};
		\node [style=none] (47) at (6.75, 4) {};
		\node [style=none] (50) at (7.5, 4.75) {};
		\node [style=none] (51) at (7.25, 4) {};
		\node [style=none] (52) at (7.75, 4) {};
		\node [style=none] (53) at (7.75, 1.25) {};
		\node [style=none] (54) at (8.5, -1.75) {$c$};
		\node [style=none] (55) at (4.5, 6) {};
		\node [style=none] (56) at (5.75, 7.5) {};
		\node [style=none] (57) at (5, 8.5) {};
		\node [style=none] (58) at (5.75, 7.5) {};
		\node [style=none] (59) at (5.25, 7) {};
		\node [style=none] (60) at (3, 7) {};
		\node [style=none] (61) at (2.5, 7.5) {};
		\node [style=none] (62) at (3.5, 14.25) {};
		\node [style=none] (63) at (3.5, 13.25) {};
		\node [style=none] (64) at (2.5, 13.75) {};
		\node [style=none] (65) at (3, 14.25) {};
		\node [style=none] (66) at (5.75, 13.75) {};
		\node [style=none] (67) at (5.25, 14.25) {};
		\node [style=none] (68) at (4.5, 7) {};
		\node [style=round box] (69) at (3.5, 12.75) {$+$};
		\node [style=none] (70) at (4.5, 9) {};
		\node [style=none] (71) at (1, 7) {$a + b$};
		\node [style=none] (72) at (5.5, 9) {};
		\node [style=none] (73) at (4, 9) {};
		\node [style=none] (76) at (4.75, 9.75) {};
		\node [style=none] (77) at (4.5, 9) {};
		\node [style=none] (78) at (5, 9) {};
		\node [style=none] (79) at (4, 10.75) {};
		\node [style=none] (80) at (4.75, 11.5) {$x+c$};
		\node [style=none] (81) at (3.75, 11.75) {};
		\node [style=none] (82) at (3.25, 11.75) {};
		\node [style=none] (83) at (3, 10) {};
		\node [style=none] (84) at (3.25, 12.25) {};
		\node [style=none] (85) at (3.75, 12.25) {};
		\node [style=none] (86) at (6.75, 5.75) {$\_ + c$};
		\node [style=none] (87) at (3.5, 15.25) {};
	\end{pgfonlayer}
	\begin{pgfonlayer}{edgelayer}
		\draw (3.center) to (2.center);
		\draw [in=-90, out=150] (1.center) to (4.center);
		\draw (10.center) to (9.center);
		\draw [in=0, out=90, looseness=1.75] (13.center) to (14.center);
		\draw (13.center) to (5.center);
		\draw [in=0, out=-90, looseness=1.75] (5.center) to (6.center);
		\draw (6.center) to (7.center);
		\draw [in=-90, out=180, looseness=1.75] (7.center) to (8.center);
		\draw (8.center) to (11.center);
		\draw [in=180, out=90, looseness=1.75] (11.center) to (12.center);
		\draw (14.center) to (12.center);
		\draw [in=0, out=-90, looseness=1.75] (15.center) to (16.center);
		\draw (16.center) to (17.center);
		\draw [in=-90, out=180, looseness=1.75] (17.center) to (18.center);
		\draw (18.center) to (19.center);
		\draw [in=180, out=90, looseness=1.75] (19.center) to (20.center);
		\draw (20.center) to (21.center);
		\draw [in=90, out=0, looseness=1.75] (21.center) to (22.center);
		\draw (22.center) to (15.center);
		\draw (9.center) to (23);
		\draw [in=165, out=-90, looseness=0.50] (2.center) to (15.center);
		\draw [style={lambda_unit_string}] (28.center)
			 to [in=90, out=0] (24.center)
			 to (25.center)
			 to [in=-180, out=90] cycle;
		\draw [in=-90, out=90, looseness=0.50] (23) to (29.center);
		\draw (31.center) to (32.center);
		\draw [in=-90, out=90, looseness=0.75] (32.center) to (30.center);
		\draw [style={lambda_unit_string}] (34.center)
			 to [in=-180, out=90] (37.center)
			 to [in=90, out=0] (33.center)
			 to cycle;
		\draw (40.center) to (41.center);
		\draw [in=-90, out=90, looseness=0.75] (41.center) to (39.center);
		\draw [in=-90, out=90, looseness=0.75] (28.center) to (38.center);
		\draw [style={lambda_unit_string}] (37.center) to (42.center);
		\draw [style={lambda_unit_string}] (46.center)
			 to (47.center)
			 to [in=-180, out=90] (50.center)
			 to [in=90, out=0] cycle;
		\draw [style={lambda_unit_string}, in=270, out=90] (45.center) to (53.center);
		\draw [style={lambda_unit_string}] (53.center) to (52.center);
		\draw [in=-90, out=0, looseness=1.25] (23) to (51.center);
		\draw [in=-105, out=90, looseness=0.50] (50.center) to (55.center);
		\draw [in=-90, out=150] (56.center) to (57.center);
		\draw (63.center) to (62.center);
		\draw [in=0, out=90, looseness=1.75] (66.center) to (67.center);
		\draw (66.center) to (58.center);
		\draw [in=0, out=-90, looseness=1.75] (58.center) to (59.center);
		\draw (59.center) to (60.center);
		\draw [in=-90, out=180, looseness=1.75] (60.center) to (61.center);
		\draw (61.center) to (64.center);
		\draw [in=180, out=90, looseness=1.75] (64.center) to (65.center);
		\draw (67.center) to (65.center);
		\draw (68.center) to (70.center);
		\draw [style={lambda_unit_string}] (76.center)
			 to [in=90, out=0] (72.center)
			 to (73.center)
			 to [in=180, out=90] cycle;
		\draw [style={lambda_unit_string}] (57.center) to (78.center);
		\draw [in=-90, out=90, looseness=1.25] (76.center) to (79.center);
		\draw [in=-90, out=90, looseness=0.75] (83.center) to (82.center);
		\draw [in=-90, out=90] (79.center) to (81.center);
		\draw (81.center) to (85.center);
		\draw [style={lambda_unit_string}] (82.center) to (84.center);
		\draw (60.center) to (83.center);
		\draw [in=-90, out=90] (55.center) to (68.center);
		\draw [in=-90, out=90, looseness=0.25] (42.center) to (60.center);
		\draw (87.center) to (62.center);
	\end{pgfonlayer}
\end{tikzpicture}
		}
		\caption{Closed e-hypergraph}%
		\label{fig:closed}
	\end{subfigure}
	\caption{Extensions of e-graphs to handle binding in $\lambda x . (a + b) + (x + c)$.}%
	\label{fig:extended-egraphs}
	\vspace{-8pt}
\end{wrapfigure}
Recent work on \emph{slotted} e-graphs~\cite{slotted-egraphs} addresses the sharing issues by incorporating variables as first-class data and making binding a built-in operation, guaranteeing freshness inside binders and differentiating bound from free variables.
An example is shown in~\autoref{fig:slotted}: each e-class is parameterised by a set of free variables and can be \enquote{invoked} by supplying concrete variable names. 
We can achieve a similar degree of sharing via closure conversion in our setting (\autoref{fig:closed}).
The key difference is that slotted e-graphs solve variable binding generically using nominal techniques applicable to arbitrary formal languages with binders (e.g.\ $\pi$-calculus), while our approach embeds into the internal language of our categorical model, using metatheoretic $\lambda$-abstraction to handle variable binding.
We provide a case study in~\autoref{sec:application}, where we argue that our approach has particular advantages for e-graphs for (variations of) $\lambda$-calculus.
Our categorical model can also help justify the congruences maintained by slotted e-graphs, as we elaborate in~\autoref{sec:categorical}.
We note that Schneider et al.~\cite{slotted-egraphs} focus on the e-graph as a data structure with a full implementation, while we focus on the foundational categorical semantics of e-graphs with binding and their combinatorial representation.
We hope our work serves as a foundation for semantics of e-graphs with bindings, a currently underdeveloped area~\cite{zhang_relational_2022,tree_automata_e_graphs}.

\begin{figure*}
	\captionsetup{skip=0pt}
	\begin{subfigure}{\linewidth}
		\adjustbox{width=\textwidth}{
			\input{./figures/e-graph-substitution.tikz}
		}
		\subcaption{E-graph explicit substitution example.}
		\label{fig:e-graph-substitution}
	\end{subfigure}
	\begin{subfigure}{\linewidth}
		\adjustbox{width=\textwidth}{
			\input{./figures/e-string-substitution-2.tikz}
		}
		\subcaption{String-diagrammatic substitution example.}\label{fig:e-string-substitution}
	\end{subfigure}
	\caption{Conventional e-graphs vs.\ categorical e-graphs with binding.}
\end{figure*}

\subsection{Contributions}
Our contribution is the construction of a form of e-graphs that support binding, developed through a principled categorical approach.
Specifically, we extend a prior categorical semantics for e-graphs without binding by incorporating a monoidal closed structure.
Tiurin et al.~\cite{tiurin2025equivalencehypergraphsdporewriting} shows that e-graphs over signature $\Sigma$ can be represented as morphisms of a free Cartesian symmetric monoidal semilattice-enriched category over the same signature.
We argue that restricting the categorical domain to free \textbf{closed} symmetric monoidal semilattice-enriched categories gives rise to structures that naturally support binding and equivalence classes of morphisms.
That is, morphisms (string diagrams) of such category can both encode variable binding (using the closed structure) and equivalence (using the semilattice enrichment) between subterms (sub-diagrams).
We call these string diagrams, or more precisely, the combinatorial representation of them, \emph{closed e-hypergraphs}.
First, we define closed enriched $\Sigma$-terms which generate a theory with equivalences and bindings.
Then we define what closed e-hypergraphs are and provide the interpretation of closed enriched $\Sigma$-terms as cospans of e-hypergraphs (and by doing this we also formalise the corresponding string diagrams as concrete combinatorial objects).
Next we formulate the DPO rewriting for closed e-hypergraphs and show that this notion of rewriting is sound and complete with respect to rewriting of closed enriched $\Sigma$-terms modulo SMC equations.
We describe what benefits the closed e-hypergraph representation may bring for equality saturation in the context of $\lambda$-calculus with explicit substitution compared to conventional (slotted) e-graphs in~\autoref{sec:application}.

\section{Categorical semantics of E-Graphs with Bindings}%
\label{sec:categorical}

This section introduces preliminaries on symmetric monoidal semilattice-enriched categories generated by monoidal theories, and the string diagram formalism used to represent them.
We also show a specific case of closed monoidal semilattice-enriched categories and how the closed structure interacts with the enrichment.

We use $(\fatsemi\!)$ to denote the composition in diagrammatic order and use $\sym$ for the symmetry natural transformation for a (strict) tensor ($\otimes$).
We also omit subscripts on identities and natural transformations where it can be inferred.

\begin{definition}
	A monoidal signature $\Sigma = (\Sigma_{O}, \Sigma_{M}, t)$ is given by a set of objects (called types) $\Sigma_{O}$, a set of generators (operations) $\Sigma_{M}$ and a type assignment function $t\colon \Sigma_{M} \to \Sigma_{O}^{*} \times \Sigma_{O}^{*}$ assigning lists of input and output types to a given operation.
\end{definition}
\begin{definition}
	The set of $\Sigma$-terms generated by monoidal signature $\Sigma$ is given by the grammar:
	$\tau \Coloneqq \phi \;|\; \id_{I} \;|\; \id_{A} \;|\; \sym_{A,B} \;|\; \tau_{1} \fatsemi \tau_{2} \;|\; \tau_{1} \otimes \tau_{2}$,
	where $\phi \in \Sigma_{M}$ and $A,B \in \Sigma_{O}$.
\end{definition}
Identities and symmetries of complex types, e.g.\ $\id_{A \otimes B}$ and $\sym_{A \otimes B, C}$, are given by the composition and tensoring of $\id_{A}$ and $\sym_{A,B}$ as, e.g.\ $\id_{A \otimes B} = \id_{A} \otimes \id_{B}$ and $\sym_{A \otimes B, C} = (\id_{A} \otimes \sym_{B,C}) \fatsemi (\sym_{A,C} \otimes \id_{B})$.
\begin{definition}
	The free symmetric monoidal category $\catname{S}(\Sigma)$ is given by a set of $\Sigma$-terms quotiented by the axioms of a symmetric monoidal category.
	The objects are lists of types from $\Sigma_{O}$ with tensoring given by list concatenation.
\end{definition}
\begin{definition}
	A symmetric monoidal theory $\catname{SMT}(\Sigma, \mathcal{E})$ is given by a set of $\Sigma$-terms quotiented by a set of equations $\mathcal{E}$ consisting of pairs of appropriately typed $\Sigma$ terms $l = r$.
\end{definition}
\begin{definition}
	A categorical presentation of  $\catname{SMT}(\Sigma, \mathcal{E})$ is a free symmetric monoidal category $\catname{S}(\Sigma, \mathcal{E})$ given by a set of $\Sigma$ terms quotiented by the laws of a symmetric monoidal category and $\mathcal{E}$.
\end{definition}

We can similarly consider \emph{closed} terms that give rise to a free \emph{closed} symmetric monoidal category.

\begin{definition}%
	\label{def:closed}
	A (right) closed monoidal category is a monoidal category $\catname{C}$ satisfying that for
	every pair of objects $B,C$ there is an object $B \multimap C$ and a morphism $\textbf{ev}_{B,C} : (B \multimap C) \otimes B \to
		C$, and for every triple of objects $A,B,C$ there is an operation $\Lambda_{A,B,C} : \catname{C}(A \otimes B,C) \to
		\catname{C}(A,B \multimap{} C)$ such that for all $f : A \otimes B \to C$ and $g : Z \to A$ the following equations hold:
	\begin{gather*}
		f = \Lambda_{A,B,C}(f) \otimes \id_{B} \fatsemi \textbf{ev}_{B,C},
		\quad
		\id_{B \multimap{} C} = \Lambda_{B \multimap{} C,B,C}(\textbf{ev}_{B,C}),
		\\
		\Lambda_{Z,B,C}(g \otimes \id_{B} \fatsemi f) = g \fatsemi \Lambda_{A,B,C}(f).
	\end{gather*}
\end{definition}

This is equivalent to defining a (right) closed monoidal category as one admitting an adjunction $- \otimes B \dashv B \multimap{} -$ for every object $B$.
If the tensor is symmetric, then the monoidal category is \emph{closed}.

\begin{definition}
	Closed $\Sigma$-terms are defined inductively by first defining a family of objects $obj_{\Sigma_{O}}$.
	\begin{itemize}
		\item A designated $I$ is in $obj_{\Sigma_{O}}$ and for all $A \in \Sigma_{O}$, $A \in obj_{\Sigma_{O}}$.
		\item If $A \in obj_{\Sigma_{O}}$ and $B \in obj_{\Sigma_{O}}$, then $A \otimes B \in obj_{\Sigma_{O}}$ and $A \multimap B \in obj_{\Sigma_{O}}$.
	\end{itemize}
	A family of closed $\Sigma$ terms is then defined by:
	\begin{itemize}
		\item $\phi \colon A \to B \in \Sigma_{M}$ is a closed $\Sigma$-term, once we turned lists into tensors;
		\item if $A$ and $B$ are $obj_{\Sigma_{O}}$, then $\id_{A}$ and $\sym_{A,B}$ are closed $\Sigma$-terms;
		\item if $f \colon A \to B$ and $g : B \to C$ are closed $\Sigma$-terms, then so is $f \fatsemi g : A \to C$;
		\item if $f \colon A \to B$ and $g : C \to D$ are closed $\Sigma$-terms, then so is $f \otimes g \colon A \otimes C \to B \otimes D$;
		\item if $A$ and $B$ are $obj_{\Sigma_{O}}$, then $\textbf{ev}_{A,B} \colon (A \multimap B \otimes A) \to B$ is a closed $\Sigma$-term;
		\item if $h \colon X \otimes A \to B$ is a closed $\Sigma$-term, then $\Lambda_{X,A,B}(h) \colon X \to (A \multimap B)$ is a closed $\Sigma$-term.
	\end{itemize}
\end{definition}

\begin{definition}
	The free closed symmetric monoidal category $\catname{CS}(\Sigma)$ is given by a set of closed $\Sigma$-terms quotiented by the laws of a symmetric monoidal category and the equations from~\autoref{def:closed}.
	We can similarly consider a free closed symmetric monoidal category $\catname{CS}(\Sigma, \mathcal{E})$ induced by a $\catname{SMT}(\Sigma, \mathcal{E})$ over closed $\Sigma$-terms.
\end{definition}

To encode equivalences between morphisms (or $\Sigma$-terms), the Hom sets of the corresponding free categories can be endowed with the additional structure of a semilattice induced by an idempotent, commutative and associative operator \textit{join} $+$ as presented by Tiurin et al.~\cite{tiurin2025equivalencehypergraphsdporewriting}.
This is modelled by enrichment in semilattices, yielding $\catname{SLat}$-categories where $\catname{SLat}$ is the category of \textit{unbounded} semilattices and homomorphisms between them (see~\autoref{sec:appendix:slat}).

\begin{definition}[$\catname{SLat}$-category]%
	\label{def:slat-cat}
	An $\catname{SLat}$-category $\mathbb{C}$ is defined by the following data:
	\begin{itemize}
		\item a set of objects $\mathbb{C}$;
		\item for every pair of objects $A,B \in \mathbb{C}$ --- a Hom semilattice $\mathbb{C}(A,B) \in \catname{SLat}$;
		\item for every triple of objects $A,B,C \in \mathbb{C}$ --- a composition morphism
		      $\circ\colon \mathbb{C}(B,C) \otimes \mathbb{C}(A,B) \to \mathbb{C}(A,C)$ of $\catname{SLat}$;
		\item for every object $A \in \mathbb{C}$ --- a unit morphism
		      $u_{A}\colon I \to \mathbb{C}(A,A)$, where $I$ is the monoidal unit $\{*\}$ for $\catname{SLat}$\footnote{Semilattice homomorphisms $\{*\} \to P$ are in bijection with the elements of $P$ as a set, so this is an element of the underlying set of $\mathbb{C}(A,A)$.},
	\end{itemize}
   such that the coherence diagrams for associativity and unitality of composition commute.
\end{definition}
$\catname{SLat}$ is sufficiently close to $\catname{Set}$ that Hom semilattices can be thought of as the ordinary notion of Hom sets with additional semilattice structure; unit morphisms are as usual, and composition works similarly to composition in categories with an extra condition of respecting the join: as it is a semilattice homomorphism, we require that $(f + g)\fatsemi\; h = f\fatsemi\;h + g\fatsemi\;h$ and $f\fatsemi\;(g+h) = f\fatsemi\;g + f\fatsemi\;h$.

\begin{definition}[$\catname{SLat}$-functor]
	Let $\mathbb{C}$ and $\mathbb{D}$ be two $\catname{SLat}$-categories.
	An $\catname{SLat}$-functor $F\colon \mathbb{C} \to \mathbb{D}$ is defined by the following data:
	\begin{itemize}
		\item a mapping on objects $A$ of $\mathbb{C}$ to $F A$ of $\mathbb{D}$;
		\item for each pair of objects $A, B$ of $\mathbb{C}$, a morphism in $\catname{SLat}$ $F_{A,B}\colon \mathbb{C}(A,B) \to \mathbb{D}(FA,FB)$;
	\end{itemize}
	such that certain coherence diagrams, enforcing the preservation of composition and units, commute.
\end{definition}
As before, this is like an ordinary functor, except with an additional requirement that its action on morphisms respects the join: $F(f + g) = F(f) + F(g)$.

The collection of $\catname{SLat}$-categories and $\catname{SLat}$-functors form an $\catname{SLat}$-category $\catname{SLat}\text{--}\catname{Cat}$\footnote{More generally, if we were to define $\catname{SLat}$-natural transformations, we could view $\catname{SLat}\text{--}\catname{Cat}$ as a symmetric monoidal 2-category~\cite{Kelly2022BASICCO}.}, and for any pair of $\catname{SLat}$-categories $\mathbb{C}$ and $\mathbb{D}$, the product $\catname{SLat}$-category $\mathbb{C} \times \mathbb{D}$ is defined analogously to categories.

\begin{definition}[Symmetric monoidal $\catname{SLat}$-category]%
	\label{def:enriched-prop}
	A symmetric monoidal $\catname{SLat}$-category is given by a $\catname{SLat}$-category $\mathbb{C}$ that additionally has: a unit object $I \in \mathbb{C}$, an $\catname{SLat}$-functor tensor $\otimes\colon \mathbb{C} \times \mathbb{C} \to \mathbb{C}$, such that particular coherence diagrams analogous to those for symmetric monoidal categories commute.
\end{definition}
This is completely analogous to a symmetric monoidal category, with the fact that the tensor is an $\catname{SLat}$-functor ensuring compatibility with the join: $f \otimes (g+h) = f \otimes g + f \otimes h $ and $ (f+g) \otimes h = f \otimes h + g \otimes h$.
We take $\otimes$ to bind more tightly than $+$.

\begin{proposition}(A specialised case of Proposition 6.4.7~\cite{Borceux_1994} and Theorem 5.7.1~\cite{cruttwell2008normed})
	There is a 2-adjunction
	\(
	\mathcal{F} \dashv \mathcal{U}\colon \catname{SLat}\text{--}\catname{Cat} \to \catname{Cat}
	\)
	induced by the usual (symmetric monoidal) free-forgetful adjunction
	\(
	F \dashv U\colon \catname{SLat} \to \catname{Set}.
	\)
	Moreover, $\mathcal{F}$ and $\mathcal{U}$ preserve symmetric monoidal structure, sending symmetric monoidal categories to symmetric monoidal $\catname{SLat}$-categories and vice versa.
\end{proposition}
In particular, the 2-functor $\mathcal{F}$ turns every category $\mathbb{C}_0 \in \catname{Cat}$ freely into a semilattice-enriched category $\mathbb{C} \in \catname{SLat}\text{--}\catname{Cat}$ by making every Hom set of $\mathbb{C}_0$ into the free semilattice on that set along $F$.
Being a 2-functor, it preserves dualities (in the 2-categorical sense): it lifts an adjunction in between categories in $\catname{Cat}$ into an $\catname{SLat}$-adjunction between their free semilattice-enriched counterparts.
We take closure, in the sense of~\autoref{def:closed}, to be defined generally for $\catname{SLat}$-categories by an analogous $\catname{SLat}$-adjunction $- \otimes B \dashv B \multimap -$ for every object $B$.
Thus, applying the 2-functor $\mathcal{F}$ to $\catname{CS}(\Sigma)$ we get a closed symmetric monoidal $\catname{SLat}$-category $\catname{CS}(\Sigma)^{+}$ by this lifting.
We can similarly derive $\catname{SLat}$-enriched versions of  $\Lambda_{A,B,C}$ and $\textbf{ev}_{A,B}$ morphisms using the lifted adjunction as $\textbf{ev}_{A,B} = \varepsilon_{A}$ and $\Lambda_{A,B,C}(f) = \eta_A \fatsemi (B \multimap{} f)$.
The following distributivity law can then be derived, forming the basis of encoding equivalences inside abstractions:
\begin{equation}
	\begin{aligned}
		\Lambda_{A,B,C}(f + g)  &= \eta_A \fatsemi (B \multimap (f + g))                    
		                        = \eta_A \fatsemi (B \multimap f + B \multimap g)\\          
		                        &= \eta_A \fatsemi (B \multimap f) + \eta_A \fatsemi (B \multimap g) 
		                        = \Lambda_{A,B,C}(f) + \Lambda_{A,B,C}(g).
	\end{aligned}%
	\label{law:distributivity}
\end{equation}

Similar to $\catname{CS}(\Sigma)$, $\catname{CS}(\Sigma)^{+}$ can be freely generated from a set of closed $\Sigma^{+}$-terms as follows.
The set of closed $\Sigma^{+}$-terms is constructed just like the set of closed $\Sigma$-terms, with $+$ being an additional constructor for terms.
Terms of $\catname{CS}(\Sigma)^{+}$ are quotiented by the laws of closed symmetric monoidal $\catname{SLat}$-categories.
The same applies in the presence of equations for $\catname{CS}(\Sigma,\mathcal{E})^{+}$.

To aid reasoning, we introduce a new language of string diagrams for closed symmetric monoidal $\catname{SLat}$-categories, using a hierarchical \enquote{dashed box} structure to capture the join operation on morphisms (as was done by Tiurin et al.~\cite{tiurin2025equivalencehypergraphsdporewriting}) as well as another hierarchical \enquote{rounded box} for $\lambda$-abstraction.
\autoref{fig:egraph-strings} displays the generators of this language, which is the usual string-diagrammatic syntax~\cite{Selinger_2010} plus the syntax for $\lambda$-abstraction as established by Ghica and Zanasi~\cite{ghica2024stringdiagramslambdacalculifunctional} and syntax for semilattice enrichment as introduced by Tiurin et al.~\cite{tiurin2025equivalencehypergraphsdporewriting}; the first component denotes an empty diagram.

\begin{figure*}
	\centering
	\adjustbox{scale=0.65}{
		\input{./figures/egraph-strings-2.tikz}
	}
	\captionsetup{skip=0pt}
	\caption{String diagrams for closed symmetric monoidal $\catname{SLat}$-categories; dashes stand for arbitrary number of wires and appropriate labels.}
	\label{fig:egraph-strings}
\end{figure*}

\subsection{Categorical foundations for slotted e-graphs}

Slotted e-graphs~\cite{slotted-egraphs} extend the traditional signature employed by e-graphs by an explicit construct for binding ($\text{bind}\; \$x\; tc$) and variables (\emph{slots}).
	Their syntax is generated by the following grammar:
	\begin{alignat*}{2}
		\text{(functional symbols)} &        &                            & \; f,\;g, \ldots ;                                                                       \\
		\text{(slots)}              &        &                            & \; \$x,\; \$y, \ldots ;                                                                   \\
		\text{(term children)}      &        & tc               \Coloneqq & \; t \mid tc_{1},\; \ldots,\; tc_{n} \mid \text{bind}\; \$x\; tc \mid \$x ; \qquad n \geq 1 \\
		\text{(terms)}              & \qquad & t                \Coloneqq & \; f \mid f(tc) .
	\end{alignat*}
	The traditional congruence relation maintained by e-graphs for a set of equations $E$ is given below (plus the obvious rules for transitivity, symmetry and reflexivity of $\cong$):
	\begin{mathpar}
		\inferrule*[right=start]
		{a = b \in E}
		{E \vdash a \cong b}
		\quad
		\inferrule*[right=cong. (variadic)]
		{E \vdash tc_i \cong tc'_i,\; i = 1, \ldots, k}
		{E \vdash tc_1, \ldots, tc_k \cong tc'_1, \ldots, tc'_k}
		\quad 
		\inferrule*[right=cong. (f)]
		{E \vdash tc \cong tc'}
		{E \vdash f(tc_1) \cong f(tc')}
	\end{mathpar}
	Slotted e-graphs extend these with the following:
	\begin{equation}%
		\label{eq:cong-bind}
		\inferrule*[right=cong. (bind)]
		{E \vdash tc \cong tc'}
		{E \vdash \text{bind}\; \$x\; tc \cong \text{bind}\; \$x\; tc'},
	\end{equation}
	plus the rules for equivalence of terms up to renaming of slots.
	The earlier work of Tiurin et al.~\cite{tiurin2025equivalencehypergraphsdporewriting} that introduced free semilattice enrichment to describe e-graphs using morphism terms used the join operation $+$ of a semilattice to denote the equivalence of morphism subterms.
	For example, the e-graph (b) from~\autoref{fig:e-graph-example} can be written as a term as
	\[(a \otimes ((2\fatsemi\comonoid) \otimes (1\fatsemi id))\,\fatsemi\,\sym)\fatsemi((* \otimes \counit + (id \otimes \counit \otimes id)\fatsemi<\!\!<) \otimes id)\fatsemi/~,\] where $+$ encodes the equivalence of the subterms rooted at $*$ and $<\!\!<$.
	We can see that the first three rules (plus reflexivity, symmetry, and transitivity) above correspond to idempotence, commutativity and associativity of $+$ and the fact that $+$ respect the composition\footnote{For example, given a morphism $f\,\fatsemi\,(a + b)$, we know it is equal to $f\,\fatsemi\,a + f\,\fatsemi\,b$ as per the axioms of the category.} $(\fatsemi)$.
	Moreover, we can also justify~\autoref{eq:cong-bind} from a purely categorical perspective, by the free enrichment as per~\autoref{law:distributivity}.
	We do not require the additional rules regarding renaming as string diagrams are essentially a nameless representation of terms.
	There is a distinction here because some renaming rules are not consistent with string diagrams: slotted e-graphs try to represent, e.g., terms $(a + b)$ and $(c + d)$ as a single entity (up to renaming of free variables), unlike string diagrams (unless we do closure conversion---see~\autoref{fig:extended-egraphs}).



\section{String Diagram Rewrite Theory}%
\label{sec:combinatorial-semantics}

We recall the fundamental definitions and results of string diagram rewrite theory for symmetric monoidal categories, recapitulating the correspondence between string diagram rewriting and an appropriate notion of double pushout (DPO) rewriting of certain hypergraphs, as established by Bonchi et al.~\cite{bonchi_string_2022-1,bonchi_string_2022-2}.

We write $(-)^*$ for the free monoid monad over $\catname{Set}$, extended element-wise to relations; for $R \subseteq V \times W$ and $v \in V$ we write $R(v) \subseteq W$ in functional notation.
We elide coproduct associativity and unit isomorphisms, write $\iota_j$ for the $j^{\text{th}}$ injection into $X_1 \sqcup \cdots \sqcup X_n$, $[f,g]$ for co-pairing, and $A \sqcup_{f,g} B$ for the pushout of $A \xleftarrow{f} C \xrightarrow{g} B$; we call $A,B$ the \emph{feet} and $C$ the \emph{carrier} of the cospan.
Given an ordered sequence of objects $[A_1, \ldots, A_n]$, we write $w([A_1, \ldots, A_n])$ for the object obtained by tensoring them, $A_1 \otimes \cdots \otimes A_n$, with $w([\,]) = I$ the monoidal unit; e.g.\ $w([A,B]) = A \otimes B$.

\subsection{Hypergraphs}

Hypergraphs generalise directed graphs by allowing edges to have multiple (ordered) sources and targets.
Bonchi et al.~\cite{bonchi_string_2022-2} show that a certain class of hypergraphs \emph{with interfaces} are sound and complete for categorical presentations of arbitrary symmetric monoidal theories.
They develop the use of hypergraphs as combinatorial representations of string diagrams for symmetric monoidal categories, whereby morphisms (nodes) are represented as hyperedges, and objects (wires) as vertices.
String diagram rewriting is then formalised as DPO rewriting of said hypergraphs.
We briefly present this correspondence, but the interested reader should consult the work of Bonchi et al.~\cite{bonchi_string_2022-1,bonchi_string_2022-2}.

\begin{definition}[Category of hypergraphs]\label{def:hypergraph}
	A hypergraph $\mathcal{G}$ over a monoidal signature $\Sigma = (\Sigma_{O}, \Sigma_{M})$ is a tuple $(V,E,s,t,l)$, where $V$ is a finite set of vertices, $E$ is a finite set of hyperedges, $s\colon E \to V^{*}$ is a source function, $t\colon E \to V^{*}$ is a target function,  and $l = (l_{V}\colon V \to \Sigma_{O},l_{E}\colon E \to \Sigma_{M})$ is a pair of labelling functions that assigns each hyperedge and vertex a generator and a type from monoidal signature $\Sigma$ respectively.

	The labelling function must respect the typing of the generator:
	\[
		\forall e \in E. \, l(e) = c : A \to B \implies l_{V}^{*}(s(e)) = A \text{ and } l_{V}^{*}(t(e)) = B.
	\]
	We call a hypergraph discrete if its set of hyperedges is empty.

	A \emph{hypergraph homomorphism} $\phi\colon \mathcal{F} \to \mathcal{G}$ is given by a pair of functions $\phi_V\colon V_{\mathcal{F}} \to V_{\mathcal{G}}, \phi_E\colon E_{\mathcal{F}} \to E_{\mathcal{G}}$ which respect source, target, and labelling.

	$\catname{Hyp(\Sigma)}$ denotes the category of hypergraphs and hypergraph homomorphisms over $\Sigma$.
\end{definition}

String diagrams also have dangling wires exiting at their top/bottom boundaries, corresponding to special \emph{input} and \emph{output} vertices.



\begin{definition}[Hypergraphs with interfaces]%
	\label{def:cspd}
	A hypergraph with interfaces is a cospan $n \xrightarrow{\mathsf{in}} \mathcal{G} \xleftarrow{\mathsf{out}} m$ in $\catname{Hyp}(\Sigma)$ where $n$, $m$ are discrete hypergraphs denoting input and output interfaces.
	We write $\catname{HypI}(\Sigma)$ for the full subcategory of the quotient category of cospans\footnote{This is the 1-categorical version of the bicategory of cospans: for a bicategory, its quotient category is given by the same objects but by taking \emph{isomorphism classes} of morphisms instead.} of $\catname{Hyp}(\Sigma)$ whose objects are discrete hypergraphs and morphisms are isomorphism classes of cospans of hypergraphs\footnote{By abuse of notation, we ignore the distinction between a cospan and its isomorphism class --- this is validated by replacing \enquote{pushout} (which is only determined up to unique isomorphism) by \enquote{chosen pushouts}, which any concrete setting (e.g.\ an implementation of DPO rewriting) would satisfy.}.
\end{definition}

Additional constraints on hypergraphs are needed to put them in bijective correspondence with string diagrams in a symmetric monoidal category: \emph{monogamy}, which enforces that $\mathsf{in}$ and $\mathsf{out}$ are injective and constrain the in-degree and out-degree of vertices in $\mathcal{G}$, and \emph{acyclicity}, which states that there are no loops in the hypergraph.
Together, these ensure that a hypergraph with interfaces \enquote{looks like} a string diagram, and determines a subcategory $\catname{MHypI}(\Sigma) \hookrightarrow \catname{HypI}(\Sigma)$ --- the category of \emph{monogamous directed acyclic (MDA) hypergraphs with interfaces}.
An example and a non-example of an MDA cospan can be found in~\autoref{sec:appendix:hyp}.

There is a symmetric monoidal equivalence which witnesses that MDA hypergraphs with interfaces are precisely string diagrams.
\begin{theorem}[Corollary 26 of~\cite{bonchi_string_2022-2}]\label{thm:prop-equiv}
	$\catname{S}(\Sigma) \cong \MdaCospans$
\end{theorem}

The practical meaning of this is that two $\Sigma$-terms $f$ and $g$ are equal modulo symmetric monoidal category equations if and only if they are represented as the same hypergraph with interfaces.
\autoref{sec:appendix:interpretation} explains the functor underlying this equivalence explicitly.
Term rewriting and hence equality modulo $\mathcal{E}$ is then implemented as DPO rewriting over such hypergraphs with interfaces, which we recall.

\subsection{DPOI-Rewriting for Hypergraphs with Interfaces}
\begin{wrapfigure}{r}{0.3\textwidth}
	\centering
	\begin{adjustbox}{scale=0.8}
					\begin{tikzcd}
						{\mathcal{L}} & {i \sqcup j} & {\mathcal{R}} \\
						{\mathcal{G}} & {\mathcal{L}^{\bot}} & {\mathcal{H}} \\
						& {n \sqcup m}
						\arrow["f"', from=1-1, to=2-1]
						\arrow[from=1-2, to=1-1]
						\arrow[from=1-2, to=1-3]
						\arrow[from=1-2, to=2-2]
						\arrow[from=1-3, to=2-3]
						\arrow["\lrcorner"{anchor=center, pos=0.125, rotate=90}, draw=none, from=2-1, to=1-2]
						\arrow[from=2-2, to=2-1]
						\arrow[from=2-2, to=2-3]
						\arrow["\lrcorner"{anchor=center, pos=0.125, rotate=180}, draw=none, from=2-3, to=1-2]
						\arrow[from=3-2, to=2-1]
						\arrow[from=3-2, to=2-2]
						\arrow[from=3-2, to=2-3]
					\end{tikzcd}
	\end{adjustbox}
\end{wrapfigure}
The rewriting of hypergraphs with interfaces (and hence of string diagrams) is formalised as double-pushout rewriting with interfaces (DPOI)~\cite{bonchi_string_2022-1,bonchi_string_2022-2}, whose defining diagram is shown on the right.
A rewrite rule is a pair of discrete cospans with matching interfaces, $i \xrightarrow{} \mathcal{L} \xleftarrow{} j$ and $i \xrightarrow{} \mathcal{R} \xleftarrow{} j$, equivalently encoded as a single span $\mathcal{L} \xleftarrow{} i \sqcup j \xrightarrow{} \mathcal{R}$~\cite[Remark~3.14]{bonchi_string_2022-1}; we make this identification throughout.
Applying the rule to $\mathcal{G}$ (with interface $n \sqcup m$) proceeds in two pushout steps.
A \emph{match} $f\colon \mathcal{L} \to \mathcal{G}$ locates an occurrence of $\mathcal{L}$; the left pushout square computes the \emph{pushout complement} $\mathcal{L}^{\bot}$ --- what remains of $\mathcal{G}$ after excising the matched part, with both $n \sqcup m$ and $i \sqcup j$ contributing to the interface, the latter corresponding to newly created inputs and outputs of the \emph{hole}.
The right pushout square then glues $\mathcal{R}$ into this context, producing the rewritten graph $\mathcal{H}$, with the interface $n \sqcup m$ threaded through both steps.
\begin{definition}[DPOI rewriting]%
	\label{def:dpoi}
	Given a span of morphisms $\mathcal{G} \xleftarrow{} n \sqcup m \xrightarrow{} \mathcal{H}$ in $\catname{Hyp}(\Sigma)$, we say $\mathcal{G}$ rewrites to $\mathcal{H}$ with interface $n \sqcup m$ via rewrite rule $\mathcal{L} \xleftarrow{} i \sqcup j \xrightarrow{} \mathcal{R}$ if there exists an object $\mathcal{L}^{\bot}$ and morphisms completing the commutative diagram above such that the two marked squares are pushouts.
\end{definition}
For soundness and completeness with respect to $\Sigma$-term rewriting modulo SMC, matches are restricted to \emph{convex} matches and pushout complements to \emph{boundary complements} (see~\autoref{sec:appendix:hyp}).
Bonchi et al.~\cite{bonchi_string_2022-2} show that any two $\Sigma$-terms $f$ and $g$ are equal modulo SMC and $\mathcal{E}$ if and only if $\llbracket f \rrbracket \Rrightarrow{}_{\llbracket \mathcal{E} \rrbracket} \llbracket g \rrbracket$, where $\llbracket - \rrbracket$ is the functor from~\autoref{thm:prop-equiv} and $\llbracket \mathcal{E} \rrbracket = \{\langle \llbracket l \rrbracket, \llbracket r \rrbracket \rangle, \langle \llbracket r \rrbracket, \llbracket l \rrbracket \rangle \mid l=r \in \mathcal{E} \}$.
Next, we build a combinatorial representation for closed $\Sigma^{+}$-terms to obtain a notion of DPOI rewriting that also subsumes equality modulo the laws of closed symmetric monoidal $\catname{SLat}$-categories.

\section{Closed E-Hypergraphs}%
\label{sec:closed-e-hyp}
Tiurin et al.~\cite{tiurin2025equivalencehypergraphsdporewriting} introduced \emph{e-hypergraphs} --- a combinatorial representation of morphism terms for free symmetric monoidal $\catname{SLat}$-categories with a modified notion of DPOI rewriting to make it sound and complete for equalities between morphism terms.
The $\catname{SLat}$-enrichment structure was encoded using hierarchical hyperedges in a hypergraph following the approach of a well-established body of literature on hierarchical hypergraphs~\cite{plump:hierarchical-graphs, montanari:gs-lambda, palacz:hierarchical-transform, Gaducci:hierarchical-graphs, Ghica:hierarchical, fscd}.
The last work in loc.\ cit.\ introduced a combinatorial representation of closed $\Sigma$-terms (without enrichment).
Notably, the laws of a $\catname{SLat}$-enrichment used by Tiurin et al.~\cite{tiurin2025equivalencehypergraphsdporewriting} and the laws of~\autoref{def:closed}~\cite{fscd} were not absorbed by the combinatorial representation and instead were implemented on top as structural DPOI rewrites.

We follow an analogous approach by combining combinatorial representations of Alvarez-Picallo et al.~\cite{fscd} and Tiurin et al.~\cite{tiurin2025equivalencehypergraphsdporewriting}.
We next extend the definition of e-hypergraphs for symmetric monoidal $\catname{SLat}$-categories~\cite{tiurin2025equivalencehypergraphsdporewriting} to accommodate new structure brought by closure.

\begin{definition}
	Over a closed symmetric monoidal signature $\Sigma = (\Sigma_{O}, \Sigma_{M})$, we define a closed e-hypergraph $\mathcal{G}$ as a tuple 
	\[(V,E,s,t, l_{V}, l_{E}, \textcolor{e-color}{<}, \textcolor{closed-color}{<},\consistency), \text{ where}\]
	\begin{itemize}
		\item $V$ is a set of vertices.
		\item $E = \colorbox{yellow}{E} \cup \textcolor{e-color}{E} \cup \textcolor{closed-color}{E}$ is a set of hyperedges that is formed of three disjoint sets of hyperedges defined below.
		\item $s,t\colon E \to V^{*}$ are source and target functions.
		\item $l_{E}\colon E \to \Sigma_{M} \sqcup 1$, $l_{V}\colon V \to \Sigma_{O} \cup \text{obj}_{\multimap}$ are label functions, where $\text{obj}_{\multimap}$ consists of objects $A \multimap B$ for $A,B \in obj_{\Sigma_{O}}$.
		\item $\textcolor{e-color}{<} \subseteq \textcolor{e-color}{E} \times (V \sqcup E)$ and $\textcolor{closed-color}{<} \subseteq \textcolor{closed-color}{E} \times (V \sqcup E)$ with $\textcolor{e-color}{E}, \textcolor{closed-color}{E} \subset E$ are disjoint partial orders that are defined as a transitive closure of corresponding immediate predecessor relations that we denote as $\textcolor{e-color}{\;<^{\mu}\;}$ and $\textcolor{closed-color}{\;<^{\mu}\;}$.
		      For $x \in V \sqcup E$, the set of predecessors of $x$ is denoted as $[x) \coloneq \{x^\prime ~|~ \exists y_{1} = x^\prime, \ldots, y_{n} = x,  y_{i} \textcolor{e-color}{\;<^{\mu}\;} y_{i + 1} \text{ or } y_{i} \textcolor{closed-color}{\;<^{\mu}\;} y_{i + 1}\; \}$.
		      We call edges $e$ such that $l(e) = \bot$ hierarchical and edges $e$ and vertices $v$ such that $[e) = \varnothing$ and $[v) = \varnothing$ top-level.
		      Each of these relations should satisfy the following:
		      \begin{enumerate}
			      \item each set of predecessors consists exclusively of hierarchical edges;
			      \item each $x$ has at most one immediate predecessor;
			      \item edges $e$ such that there is no $x$ such that $e \textcolor{e-color}{\;<^{\mu}\;} x$ or $e \textcolor{closed-color}{\;<^{\mu}\;} x$ are labelled;
			      \item the relations are closed under connectivity, \emph{i.e.}, if $v \in s(e)$ then $e^\prime \textcolor{e-color}{\;<^{\mu}\;} e$ if and only if $e^\prime \textcolor{e-color}{\;<^{\mu}\;} v$, similarly for $v \in t(e)$ and $\textcolor{closed-color}{\;<^{\mu}\;}$.
		      \end{enumerate}
		\item $\consistency$ is a consistency relation which is given by the union of a family of equivalence relations $\consistency_p$ on each set $\{x \in V \sqcup E ~|~ p \textcolor{e-color}{<^\mu} x\}$ of elements which share the same parent where each relation is also closed under connectivity, i.e.\ if $v \in s(e)$ or $v \in t(e)$ such that $p \textcolor{e-color}{\;<^{\mu}\;}(v)$ and $p \textcolor{e-color}{\;<^{\mu}\;}(e)$ then $v \consistency_{p} e$.
		      We require that $\consistency_{p} \not = (E_{p} \sqcup V_{p}) \times (E_{p} \sqcup V_{p})$ where $V_{p} = \{ v ~ | ~ p \textcolor{e-color}{\;<^{\mu}\;} v\}$ and $E_{p} = \{ e ~ | ~ p \textcolor{e-color}{\;<^{\mu}\;} e\}$.
	\end{itemize}
	$\textcolor{e-color}{E}$ and $\textcolor{closed-color}{E}$ are then induced by the corresponding relations and consist of exclusively hierarchical edges, so $\colorbox{yellow}{E} = E \setminus \textcolor{e-color}{E} \cup \textcolor{closed-color}{E}$.
\end{definition}

Intuitively, edges from $\textcolor{e-color}{E}$ encode equivalence classes, i.e.\ e-classes, while edges from $\textcolor{closed-color}{E}$ encode $\lambda$-abstraction boxes.
Also note the typing on the $l_{V}$ function.
We work with freely generated closed $\catname{SLat}$-categories, and to avoid unnecessary complications we limit the labelling of vertices to such a type.
This prohibits the labelling of vertices by e.g.\ $A \otimes B$, but allows them to be labelled by $A \otimes B \multimap C \otimes D$.
$\textcolor{closed-color}{<}$ defines closed monoidal structure as defined by Alvarez-Picallo et al.~\cite{fscd} and $(\textcolor{e-color}{<},\consistency)$ defines $\catname{SLat}$ structure as per Tiurin et al.~\cite{tiurin2025equivalencehypergraphsdporewriting}.

Consider an example in the middle part of~\autoref{fig:e-cospan-example}.
The closed e-hypergraph there is given by $V = \{v_{1}, \ldots w_{10}\}$, $E = \{e_{1}, \ldots, e_{7}\}$, $\textcolor{e-color}{E} = \{e_{1}\}$, $\textcolor{closed-color}{E} = \{e_{2}\}$, $l_{E} = {e_{1} \mapsto \bot, e_{2} \mapsto \bot, e_{3} \mapsto g, \ldots, e_{7} \mapsto f}$; the elements of the involved relations are shown in~\autoref{fig:e-cospan-example-rels},
and, in particular, we have $[e_{5}) = \{e_2, e_1\}$.
$\consistency_{e_{1}}$ defines a partition on immediate successors of $e_{1}$ which we delimit by a dashed line.
We depict edges in $\textcolor{e-color}{E}$ with a dashed box, and edges in $\textcolor{closed-color}{E}$ with a round box.
\begin{figure}
	\begin{subfigure}{0.45\textwidth}
	\[
		\adjustbox{scale=0.6}{
			\begin{tikzpicture}
	\begin{pgfonlayer}{nodelayer}
		\node [style=none] (53) at (-15.5, -4) {};
		\node [style=none] (54) at (-0.5, -4) {};
		\node [style=none] (55) at (-0.5, 5.5) {};
		\node [style=none] (56) at (-15.5, 5.5) {};
		\node [style=hedge] (0) at (-10.75, 0) {$f$};
		\node [style=none] (2) at (-14, -2.25) {};
		\node [style=none] (3) at (-12.5, -2.25) {};
		\node [style=blue node, label={below:$v_{9}$}] (7) at (-10.75, -1.25) {};
		\node [style=blue node, label={above:$w_{9}$}] (8) at (-10.75, 1.25) {};
		\node [style=none] (9) at (-10.25, -2.25) {};
		\node [style=none] (10) at (-9.75, -1.75) {};
		\node [style=none] (11) at (-10.25, 2.25) {};
		\node [style=none] (12) at (-9.75, 1.75) {};
		\node [style=none] (13) at (-14.5, -2.25) {};
		\node [style=none] (14) at (-15, -1.75) {};
		\node [style=none] (15) at (-14.5, 2.25) {};
		\node [style=none] (16) at (-15, 1.75) {};
		\node [style=red node, label={below:$v_{4}$}] (17) at (-14, -3) {};
		\node [style=red node, label={below:$v_{5}$}] (18) at (-12.5, -3) {};
		\node [style=none] (22) at (-12.5, 2.25) {};
		\node [style=none] (23) at (-12.5, 2.75) {};
		\node [style=red node, label={below:$v_{6}$}] (26) at (-8, -3) {};
		\node [style=hedge] (27) at (-8, 0) {$g$};
		\node [style=none] (28) at (-12.5, 3.25) {};
		\node [style=none] (29) at (-11.5, 3.25) {};
		\node [style=none] (30) at (-12.25, 3.75) {};
		\node [style=none] (31) at (-11.75, 3.75) {};
		\node [style=none] (33) at (-12.25, 3.25) {};
		\node [style=none] (34) at (-11.75, 3) {};
		\node [style=none] (35) at (-11.75, 3.25) {};
		\node [style=red node, label={above:$w_{5}$}] (36) at (-12, 4.5) {};
		\node [style=hedge] (38) at (-4.25, 3.25) {$f$};
		\node [style=green node, label={left:$u_{3}$}] (39) at (-4.25, 1.75) {};
		\node [style=hedge] (45) at (-4.25, 0) {$g$};
		\node [style=green node, label={above:$w_{10}$}] (46) at (-4.25, 4.5) {};
		\node [style=none] (49) at (-3.25, 5.5) {};
		\node [style=green node, label={below:$v_{10}$}] (50) at (-4.25, -3) {};
		\node [style=none] (51) at (-6, -4) {};
		\node [style=none] (52) at (-6, 5.5) {};
		\node [style=none] (57) at (-7.25, -4) {};
		\node [style=none] (58) at (-6, -4) {};
		\node [style=none] (59) at (-4.75, -4) {};
		\node [style=node, label={below:$v_{2}$}] (60) at (-6, -4.75) {};
		\node [style=node, label={below:$v_{1}$}] (61) at (-7.25, -4.75) {};
		\node [style=node, label={below:$v_{3}$}] (62) at (-4.75, -4.75) {};
		\node [style=node, label={above:$w_{2}$}] (63) at (-10, 6) {};
		\node [style=none] (64) at (-12.25, 3.75) {};
		\node [style=red node, label={above:$w_{4}$}] (65) at (-13, 4.5) {};
		\node [style=none] (66) at (-12, 3.75) {};
		\node [style=none] (67) at (-11.75, 3.75) {};
		\node [style=red node, label={above:$w_{6}$}] (68) at (-11, 4.5) {};
		\node [style=none] (71) at (-9, 5.5) {};
		\node [style=node, label={above:$w_{3}$}] (73) at (-9, 6) {};
		\node [style=node, label={above:$w_{1}$}] (74) at (-11, 6) {};
		\node [style=green node, label={above:$v_{11}$}] (77) at (-2.75, 1) {};
		\node [style=green node, label={above:$v_{12}$}] (78) at (-1.25, 1) {};
		\node [style=blue node, label={above:$v_{7}$}] (79) at (-14, 0) {};
		\node [style=blue node, label={above:$v_{8}$}] (80) at (-12.5, 0) {};
		\node [style=node, label={below:$v_1$}] (82) at (-13.5, -6.25) {};
		\node [style=node, label={below:$v_2$}] (83) at (-12.5, -6.25) {};
		\node [style=node, label={below:$v_3$}] (84) at (-11.5, -6.25) {};
		\node [style=red node, label={below:$v_4$}] (85) at (-10.5, -6.25) {};
		\node [style=red node, label={below:$v_5$}] (86) at (-9.5, -6.25) {};
		\node [style=red node, label={below:$v_6$}] (87) at (-8.5, -6.25) {};
		\node [style=blue node, label={below:$v_7$}] (88) at (-7.5, -6.25) {};
		\node [style=blue node, label={below:$v_8$}] (89) at (-6.5, -6.25) {};
		\node [style=blue node, label={below:$v_9$}] (90) at (-5.5, -6.25) {};
		\node [style=green node, label={below:$v_{10}$}] (91) at (-4.5, -6.25) {};
		\node [style=green node, label={below:$v_{11}$}] (92) at (-3.5, -6.25) {};
		\node [style=green node, label={below:$v_{12}$}] (93) at (-2.5, -6.25) {};
		\node [style=node, label={below:$v_1$}] (106) at (-11, -8.5) {};
		\node [style=node, label={below:$v_2$}] (107) at (-10, -8.5) {};
		\node [style=node, label={below:$v_3$}] (108) at (-9, -8.5) {};
		\node [style=none] (109) at (-11.75, -9.3) {};
		\node [style=none] (110) at (-11.75, -8.25) {};
		\node [style=none] (111) at (-8.25, -8.25) {};
		\node [style=none] (112) at (-8.25, -9.3) {};
		\node [style=none] (113) at (-14.5, -7.05) {};
		\node [style=none] (114) at (-1.5, -7.05) {};
		\node [style=none] (115) at (-1.5, -5.9) {};
		\node [style=none] (116) at (-14.5, -5.9) {};
		\node [style=none] (117) at (-10, -5) {};
		\node [style=none] (118) at (-10, -5.75) {};
		\node [style=none] (119) at (-10, -8) {};
		\node [style=none] (120) at (-10, -7.5) {};
		\node [style=node, label={above:$w_1$}] (121) at (-13.75, 8) {};
		\node [style=node, label={above:$w_2$}] (122) at (-12.75, 8) {};
		\node [style=node, label={above:$w_3$}] (123) at (-11.75, 8) {};
		\node [style=red node, label={above:$w_4$}] (124) at (-10.75, 8) {};
		\node [style=red node, label={above:$w_5$}] (125) at (-9.75, 8) {};
		\node [style=red node, label={above:$w_6$}] (126) at (-8.75, 8) {};
		\node [style=blue node, label={above:$v_8$}] (127) at (-7.75, 8) {};
		\node [style=blue node, label={above:$v_7$}] (128) at (-6.75, 8) {};
		\node [style=blue node, label={above:$w_9$}] (129) at (-5.75, 8) {};
		\node [style=green node, label={above:$w_{10}$}] (130) at (-4.75, 8) {};
		\node [style=green node, label={above:$v_{12}$}] (131) at (-3.75, 8) {};
		\node [style=green node, label={above:$v_{11}$}] (132) at (-2.75, 8) {};
		\node [style=none] (133) at (-14.75, 7.7) {};
		\node [style=none] (134) at (-1.75, 7.7) {};
		\node [style=none] (135) at (-1.75, 9) {};
		\node [style=none] (136) at (-14.75, 9) {};
		\node [style=node, label={above:$w_1$}] (137) at (-11, 10.25) {};
		\node [style=node, label={above:$w_2$}] (138) at (-10, 10.25) {};
		\node [style=node, label={above:$w_3$}] (139) at (-9, 10.25) {};
		\node [style=none] (140) at (-11.75, 9.95) {};
		\node [style=none] (141) at (-11.75, 11.25) {};
		\node [style=none] (142) at (-8.25, 11.25) {};
		\node [style=none] (143) at (-8.25, 9.95) {};
		\node [style=none] (144) at (-0.5, 6) {$e_{1}$};
		\node [style=none] (145) at (-10, 2.75) {$e_{2}$};
		\node [style=none] (146) at (-9, -2.25) {$e_{3}$};
		\node [style=none] (147) at (-11, 3.75) {$e_{4}$};
		\node [style=none] (148) at (-11.75, -0.75) {$e_{5}$};
		\node [style=none] (149) at (-3, 0) {$e_{6}$};
		\node [style=none] (150) at (-3, 3.25) {$e_{7}$};
		\node [style=none] (151) at (-10, 7.5) {};
		\node [style=none] (152) at (-10, 6.95) {};
		\node [style=none] (153) at (-10, 9.75) {};
		\node [style=none] (154) at (-10, 9.25) {};
		\node [style=red node, label={left:$u_{1}$}] (155) at (-12.5, 2.75) {};
		\node [style=red node, label={right:$u_{2}$}] (156) at (-8, 1.75) {};
		\node [style=none] (157) at (-11, 5.5) {};
		\node [style=none] (158) at (-10, 5.5) {};
	\end{pgfonlayer}
	\begin{pgfonlayer}{edgelayer}
		\draw [style=dashed edge] (54.center)
			 to (53.center)
			 to (56.center)
			 to (55.center)
			 to cycle;
		\draw (0) to (8);
		\draw (16.center)
			 to [in=180, out=90, looseness=1.75] (15.center)
			 to (11.center)
			 to [in=90, out=0, looseness=1.75] (12.center)
			 to (10.center)
			 to [in=0, out=-90, looseness=1.75] (9.center)
			 to (13.center)
			 to [in=-90, out=180, looseness=1.75] (14.center)
			 to cycle;
		\draw (2.center) to (17);
		\draw (3.center) to (18);
		\draw (22.center) to (23.center);
		\draw (26) to (27);
		\draw [style={lambda_unit}] (29.center)
			 to [in=-45, out=90] (31.center)
			 to [in=45, out=135, looseness=0.75] (30.center)
			 to [in=105, out=-150, looseness=0.75] (28.center)
			 to cycle;
		\draw (34.center) to (35.center);
		\draw [in=90, out=-90] (39) to (45);
		\draw (38) to (46);
		\draw (45) to (50);
		\draw (39) to (38);
		\draw [style=dashed edge] (51.center) to (52.center);
		\draw (61) to (57.center);
		\draw (60) to (58.center);
		\draw (62) to (59.center);
		\draw [bend left=15] (64.center) to (65);
		\draw (66.center) to (36);
		\draw [bend right=15] (67.center) to (68);
		\draw (71.center) to (73);
		\draw [style=boundary frame] (112.center)
			 to (111.center)
			 to (110.center)
			 to (109.center)
			 to cycle;
		\draw [style=boundary frame] (114.center)
			 to (115.center)
			 to (116.center)
			 to (113.center)
			 to cycle;
		\draw [style=diredge] (118.center) to (117.center);
		\draw [style=diredge] (119.center) to (120.center);
		\draw [style=boundary frame] (134.center)
			 to (135.center)
			 to (136.center)
			 to (133.center)
			 to cycle;
		\draw [style=boundary frame] (143.center)
			 to (142.center)
			 to (141.center)
			 to (140.center)
			 to cycle;
		\draw (7) to (0);
		\draw [style=diredge] (151.center) to (152.center);
		\draw [style=diredge] (153.center) to (154.center);
		\draw (23.center) to (33.center);
		\draw [in=90, out=-90, looseness=0.50] (34.center) to (156);
		\draw (156) to (27);
		\draw (74) to (157.center);
		\draw (63) to (158.center);
	\end{pgfonlayer}
\end{tikzpicture}
		}
	\]
	\subcaption{Example of a cospan of closed e-hypergraphs.}
	\label{fig:e-cospan-example}
	\end{subfigure}
	\hfill
	\begin{subfigure}{0.45\textwidth}
	\[
	\begin{array}{cccc}
		v_{4} \textcolor{e-color}{\;<^{\mu}\;} e_{1}  & v_{7} \textcolor{closed-color}{\;<^{\mu}\;} e_{2} & v_{4} \consistency e_{2}   & v_{4} \not \consistency v_{10} \\
		e_{3} \textcolor{e-color}{\;<^{\mu}\;} e_{1}  & v_{8} \textcolor{closed-color}{\;<^{\mu}\;} e_{2} & e_{2} \consistency u_{1}   & e_{2} \not \consistency e_{6}  \\
		e_{2} \textcolor{e-color}{\;<^{\mu}\;} e_{1}  & e_{5} \textcolor{closed-color}{\;<^{\mu}\;} e_{2} & e_{2} \consistency e_{4}   & e_{4} \not \consistency e_{7}  \\
		\ldots                                        & \ldots                                            & \ldots                     & \ldots                         \\
		v_{12} \textcolor{e-color}{\;<^{\mu}\;} e_{1} & w_{9} \textcolor{closed-color}{\;<^{\mu}\;} e_{2} & v_{11} \consistency v_{12} & w_{6} \not \consistency w_{10}
	\end{array}
	\]
	\subcaption{Corresponding relations of the closed e-hypergraph in the middle.}
	\label{fig:e-cospan-example-rels}
\end{subfigure}
\caption{}
\end{figure}
\begin{definition}[Closed e-hypergraph homomorphism]
	A homomorphism $\phi\colon \mathcal{F}\to\mathcal{G}$ between closed e-hypergraphs $\mathcal{F},\mathcal{G}$ is a pair of functions $\phi_V\colon  V_{\mathcal{F}} \to V_{\mathcal{G}}, \phi_E\colon E_{\mathcal{F}} \to E_{\mathcal{G}}$ such that:
	\begin{enumerate}
		\item $\phi$ is a hypergraph homomorphism;
		\item $\phi_{E}(\textcolor{e-color}{E_{\mathcal{F}}}) \subseteq \textcolor{e-color}{E_{\mathcal{G}}}$ and $\phi_{E}(\textcolor{closed-color}{E_{\mathcal{F}}}) \subseteq \textcolor{closed-color}{E_{\mathcal{G}}}$;
		\item for $v \in V_{\mathcal{F}}$ and $e \in E_{\mathcal{F}}$
		      \begin{equation*}
			      \text{if } e \textcolor{e-color}{\;<_{\mathcal{F}}^{\mu}\;} v \text{ then } \phi(e) \textcolor{e-color}{\;<_{\mathcal{G}}^{\mu}\;} \phi(v),
			      \qquad
			      \text{if } e \textcolor{closed-color}{\;<_{\mathcal{F}}^{\mu}\;} v \text{ then } \phi(e) \textcolor{closed-color}{\;<_{\mathcal{G}}^{\mu}\;} \phi(v),
		      \end{equation*}
		      and for $e_1, e_2 \in E_{\mathcal{F}}$
		      \begin{equation*}
			      \text{if } e_1 \textcolor{e-color}{\;<_{\mathcal{F}}^{\mu}\;} e_2 \text{ then } \phi(e_1) \textcolor{e-color}{\;<_{\mathcal{G}}^{\mu}\;} \phi(e_2),
			      \qquad
			      \text{if } e_1 \textcolor{closed-color}{\;<_{\mathcal{F}}^{\mu}\;} e_2 \text{ then } \phi(e_1) \textcolor{closed-color}{\;<_{\mathcal{G}}^{\mu}\;} \phi(e_2);
		      \end{equation*}
		\item for all $x_1, x_2 \in \mathcal{F}$,  $x_1 ~\consistency_{\mathcal{F}}~ x_2$ implies $\phi(x_1) ~\consistency_{\mathcal{G}}~ \phi(x_2)$.
	\end{enumerate}
\end{definition}
The conditions on closed e-hypergraph homomorphisms require preserving immediate predecessors and the consistency relation.
\begin{definition}[Category of closed e-hypergraphs]
	The category of closed e-hypergraphs $\catname{CEHyp(\Sigma)}$ has closed e-hypergraphs as objects and closed e-hypergraph homomorphisms as morphisms.
\end{definition}

The category of closed e-hypergraphs has coproducts, given by the disjoint union of closed e-hypergraphs, and an initial object given by the empty closed e-hypergraph.
A concrete description of the pushout of two morphisms in this category is given in~\autoref{sec:appendix:pushout}.
Generally,  the pushout of two closed e-hypergraph homomorphisms need not exist,  but it does when certain technical conditions are satisfied; in particular, a pushout for the composition of two cospans with discrete feet always exists.

As in the case of $\catname{S}(\Sigma)$ and $\MdaCospans$, closed e-hypergraphs require interfaces in order to model string diagrams with dedicated inputs and outputs.
Because string diagrams for closed  symmetric monoidal $\catname{SLat}$-categories are nested, containing substring diagrams inside e-class boxes and $\lambda$-abstraction boxes, we need to account for the interfaces in these nested contexts as well.
Nested inputs and outputs of a string diagram do not participate in composition of string diagrams, hence they need to be distinguished from the ones which do.
We thus propose \emph{extended} cospans of discrete closed e-hypergraphs and use the notation $n \setminus m$ to denote the discrete closed e-hypergraph with $n \setminus m$ vertices, in particular
with the vertices of $m$ removed from $n$ when vertices of $m$ is a sub closed e-hypergraph of $n$.

\begin{definition}[Category of closed e-hypergraphs with interfaces]
	The category of closed e-hypergraphs with extended interfaces $\Ecospans$ has discrete \emph{ordered} closed e-hypergraphs as objects,  with Hom sets $\Ecospans(n,m)$ consisting of isomorphism classes (\autoref{def:cospan-isomorphism}) of extended cospans, defined as:
	\[
		n \xrightarrow{f_{ext}} n^\prime \xrightarrow{f_{int}} \mathcal{G} \xleftarrow{g_{int}} m^\prime \xleftarrow{g_{ext}} m,
	\]
	where $\mathcal{G}$ is a closed e-hypergraph,  and $n, n^\prime, m, m^\prime$ are discrete ordered closed e-hypergraphs,  $f_{ext},g_{ext}$ are monomorphisms in $\catname{CEHyp(\Sigma)}$,  and the image of $f_{ext} \fatsemi f_{int}$ and of $g_{ext} \fatsemi g_{int}$ consist exclusively of top-level vertices,  and the vertices in the strictly internal interface (defined below) are not top-level.
	Write $\mathcal G$ to denote the extended cospan when it is clear from context $\mathcal G$ is equipped with extended interfaces.
\end{definition}
We call $n$  \emph{external input interfaces}, $n^\prime$ \emph{internal input interfaces},
and $n^\prime \setminus f_{ext}(n)$  the \emph{strictly internal input interfaces}.
We do analogously for the \emph{output interfaces},  with respect to $m$,  $m^\prime$ and $m^\prime \setminus g_{ext}(m)$.
We conflate $f_{ext}$ with $f_{ext} \fatsemi f_{int}$ when it is clear from context,  and also conflate $n$ and $m$ with their images in $n^\prime$ and $m^\prime$,  and likewise $n, m, n^\prime, m^\prime$ with their images in $\mathcal{G}$.
Given an edge $e \in E_\mathcal{G}$ such that $l(e) = \bot$,  we call the \emph{inputs of $e$} the intersection of the strictly internal input interface of $\mathcal{G}$ with the immediate successors of $e$,  and analogously for the \emph{outputs of $e$}.
Graphically, we depict such extended cospans as in~\autoref{fig:e-cospan-example}: feet of the cospan are depicted within shaded regions.
All non-black vertices form strictly internal interfaces.
We follow the same orientation that we chose for string diagrams --- the input interface of the cospan is at the bottom and the output interface is at the top.

\begin{definition}[Isomorphic cospans]\label{def:cospan-isomorphism}
	Consider a relation
	\[
		R = \bigl\{ x R y \text{ if } f_{\text{int}}(x) \consistency f_{\text{int}}(y) \;\text{or}\; \exists z \;\text{s.t.}\; z \textcolor{closed-color}{\;<^{\mu}\;} f_{\text{int}}(x) \text{ and } z \textcolor{closed-color}{\;<^{\mu}\;} f_{\text{int}}(y)\bigr\}
	\]
	for $x, y \in n^\prime$, with $S$ its reflexive closure.
	The latter partitions $n^\prime$ into non-empty subsets $\{p_{j}\}^{k}_{j=1}$.
	We get an analogous partition for $m^\prime$. 
	In~\autoref{fig:e-cospan-example}, colours denote the partitions for input and output internal interfaces.
	Two extended cospans are isomorphic if there exist isomorphisms $\alpha$, $\beta$ and $\gamma$ making the following diagram commute and such that $\alpha$ and $\gamma$ preserve order within each $p_j$.
	\[
		\scalebox{0.8}{
			\begin{tikzpicture}[scale=0.7]
	\begin{pgfonlayer}{nodelayer}
		\node [style=none, label={left:$n$}] (0) at (-9.5, 1) {};
		\node [style=none] (1) at (-7.5, 2.5) {};
		\node [style=none] (3) at (-6.5, 2.5) {$n'$};
		\node [style=none] (4) at (-5.5, 2.5) {};
		\node [style=none] (5) at (-3.5, 2.5) {};
		\node [style=none] (6) at (-2.5, 2.5) {$\mathcal{G}$};
		\node [style=none] (7) at (-1.5, 2.5) {};
		\node [style=none] (8) at (0.5, 2.5) {};
		\node [style=none] (9) at (1.5, 2.5) {$m'$};
		\node [style=none] (10) at (4.5, 1.5) {};
		\node [style=none] (11) at (2.5, 2.5) {};
		\node [style=none] (12) at (-6.5, -0.5) {$n''$};
		\node [style=none] (13) at (-5.5, -0.5) {};
		\node [style=none] (14) at (-3.5, -0.5) {};
		\node [style=none] (15) at (-2.5, -0.5) {$\mathcal{G}'$};
		\node [style=none] (16) at (-1.5, -0.5) {};
		\node [style=none] (17) at (0.5, -0.5) {};
		\node [style=none] (18) at (1.5, -0.5) {$m''$};
		\node [style=none] (19) at (-7.5, -0.5) {};
		\node [style=none] (20) at (2.5, -0.5) {};
		\node [style=none] (21) at (-9.5, 1.5) {};
		\node [style=none] (22) at (-9.5, 0.5) {};
		\node [style=none] (23) at (4.5, 0.5) {};
		\node [style=none, label={right:$m$}] (24) at (4.5, 1) {};
		\node [style=none] (25) at (-6.5, 2) {};
		\node [style=none] (26) at (-6.5, 0) {};
		\node [style=none] (27) at (-2.5, 2) {};
		\node [style=none] (28) at (-2.5, 0) {};
		\node [style=none] (29) at (1.5, 2) {};
		\node [style=none] (30) at (1.5, 0) {};
		\node [style=none] (31) at (-6, 1) {$\alpha$};
		\node [style=none] (32) at (-2, 1) {$\beta$};
		\node [style=none] (33) at (2, 1) {$\gamma$};
	\end{pgfonlayer}
	\begin{pgfonlayer}{edgelayer}
		\draw [style=diredge] (4.center) to (5.center);
		\draw [style=diredge] (8.center) to (7.center);
		\draw [style=diredge, in=0, out=90] (10.center) to (11.center);
		\draw [style=diredge] (13.center) to (14.center);
		\draw [style=diredge] (17.center) to (16.center);
		\draw [style=diredge, in=-180, out=90] (21.center) to (1.center);
		\draw [style=diredge, in=0, out=-90] (23.center) to (20.center);
		\draw [style=diredge, in=180, out=-90] (22.center) to (19.center);
		\draw [style=diredge] (25.center) to (26.center);
		\draw [style=diredge] (27.center) to (28.center);
		\draw [style=diredge] (29.center) to (30.center);
	\end{pgfonlayer}
\end{tikzpicture}
		}
	\]
\end{definition}
Intuitively, such notion of isomorphism makes it possible to freely rearrange interfaces for disjoint components.
For example, in the input interface of cospan in~\autoref{fig:e-cospan-example}, we can swap
\begin{tikzpicture}[font=\small,scale=0.6]
	\begin{pgfonlayer}{nodelayer}
		\node [style=red node, label={above:$v_{4}$}] (0) at (-0.5, 3.5) {};
		\node [style=red node, label={above:$v_5$}] (1) at (0, 3.5) {};
		\node [style=red node, label={above:$v_6$}] (2) at (0.5, 3.5) {};
	\end{pgfonlayer}
\end{tikzpicture}
with
\begin{tikzpicture}[font=\small,scale=0.6]
	\begin{pgfonlayer}{nodelayer}
		\node [style=blue node, label={above:$v_{7}$}] (0) at (-0.5, 3.5) {};
		\node [style=blue node, label={above:$v_8$}] (1) at (0, 3.5) {};
		\node [style=blue node, label={above:$v_9$}] (2) at (0.5, 3.5) {};
	\end{pgfonlayer}
\end{tikzpicture}
but we can not swap
\begin{tikzpicture}[font=\small,scale=0.6]
	\begin{pgfonlayer}{nodelayer}
		\node [style=red node, label={above:$v_{4}$}] (0) at (-0.5, 3.5) {};
		\node [style=red node, label={above:$v_5$}] (1) at (0, 3.5) {};
	\end{pgfonlayer}
\end{tikzpicture}
with each other as it would result in a non-isomorphic cospan.

\begin{wrapfigure}{r}{0.35\textwidth}
	\centering
	\trimbox{0cm 0cm 0cm 0.75cm}{
		\adjustbox{scale=0.9,center}
		{\begin{tikzpicture}
	\begin{pgfonlayer}{nodelayer}
		\node [style=none] (2) at (-19.25, -2.75) {$n'$};
		\node [style=none] (4) at (-18.5, -4) {};
		\node [style=none] (5) at (-19.25, -3.25) {};
		\node [style=none] (6) at (-19.75, -2.75) {};
		\node [style=none] (7) at (-18, -4.5) {$\mathcal{F}$};
		\node [style=none] (8) at (-16.5, -3) {};
		\node [style=none] (9) at (-17.5, -4) {};
		\node [style=none] (10) at (-16, -2.75) {$m'$};
		\node [style=none] (11) at (-15.25, -1.75) {};
		\node [style=none] (12) at (-16, -2.5) {};
		\node [style=none] (13) at (-15, -1.5) {$m$};
		\node [style=none] (14) at (-14.75, -1.75) {};
		\node [style=none] (15) at (-14, -2.5) {};
		\node [style=none] (16) at (-13.75, -2.75) {$m''$};
		\node [style=none] (17) at (-13.5, -3) {};
		\node [style=none] (18) at (-12.5, -4) {};
		\node [style=none] (19) at (-12, -4.5) {$\mathcal{G}$};
		\node [style=none] (20) at (-15.25, -7.25) {};
		\node [style=none] (21) at (-14.75, -7.25) {};
		\node [style=none] (22) at (-17.75, -4.75) {};
		\node [style=none] (23) at (-12.5, -5) {};
		\node [style=none] (24) at (-15.25, -6.75) {};
		\node [style=none] (25) at (-15, -6.5) {};
		\node [style=none] (26) at (-14.75, -6.75) {};
		\node [style=none] (27) at (-21, -1.5) {$n$};
		\node [style=none] (28) at (-20.5, -2) {};
		\node [style=none] (29) at (-8.75, -1.5) {$k$};
		\node [style=none] (30) at (-10.5, -2.75) {$k'$};
		\node [style=none] (31) at (-11.25, -4) {};
		\node [style=none] (32) at (-10.5, -3.25) {};
		\node [style=none] (33) at (-9.25, -2) {};
		\node [style=none] (34) at (-10, -2.75) {};
		\node [style=none] (35) at (-16.25, -1.75) {$f'_{ext}$};
		\node [style=none] (36) at (-17.5, -3) {$f'_{int}$};
		\node [style=none] (37) at (-13.75, -1.75) {$g_{ext}$};
		\node [style=none] (38) at (-12.5, -3) {$g_{int}$};
		\node [style=none] (39) at (-15, -7.5) {$\mathcal{H}$};
		\node [style=none] (40) at (-20.75, -2.75) {$f_{ext}$};
		\node [style=none] (41) at (-19.25, -4) {$f_{int}$};
		\node [style=none] (42) at (-9, -2.75) {$g'_{ext}$};
		\node [style=none] (43) at (-10, -4) {$g'_{int}$};
		\node [style=none] (44) at (-17.25, -6.25) {$p_1$};
		\node [style=none] (45) at (-12.75, -6.25) {$p_2$};
	\end{pgfonlayer}
	\begin{pgfonlayer}{edgelayer}
		\draw [style=diredge] (5.center) to (4.center);
		\draw [style=diredge] (8.center) to (9.center);
		\draw [style=diredge] (11.center) to (12.center);
		\draw [style=diredge] (14.center) to (15.center);
		\draw [style=diredge] (17.center) to (18.center);
		\draw [style=diredge] (22.center) to (20.center);
		\draw [style=diredge] (23.center) to (21.center);
		\draw (24.center) to (25.center);
		\draw (25.center) to (26.center);
		\draw [style=diredge] (28.center) to (6.center);
		\draw [style=diredge] (32.center) to (31.center);
		\draw [style=diredge] (33.center) to (34.center);
	\end{pgfonlayer}
\end{tikzpicture}}
	}
\end{wrapfigure}
Composition of two morphisms
$
	n \xrightarrow{f_{ext}} n^\prime \xrightarrow{f_{int}} \mathcal{F} \xleftarrow{f^\prime_{int}} m^\prime \xleftarrow{f^\prime_{ext}} m
$ and $
	m \xrightarrow{g_{ext}} m^\prime\!^\prime \xrightarrow{g_{int}} \mathcal{G} \xleftarrow{g^\prime_{int}} k^\prime \xleftarrow{g^\prime_{ext}} k
$ is computed in two stages.
First, $\mathcal{H}$ is computed as the result of the pushout square shown on the right;
then,  the result of composition is defined as follows:
\[
	n \xrightarrow{f_{ext} \fatsemi \iota_1} n^\prime \sqcup (m^{\prime\prime} \setminus m) \xrightarrow{h_{1}} \mathcal{H} \xleftarrow{h_2} k^\prime \sqcup (m^\prime \setminus m) \xleftarrow{g^\prime_{ext} \fatsemi \iota_1} k
\]
where $h_i$ are defined as follows,  using $|$ to denote the restriction of a function on a discrete closed e-hypergraph.
\begin{align*}
	h_1 = [ f_{int} \fatsemi p_1, ~(g_{int} \fatsemi p_2)|_{m^{\prime\prime} \setminus m} ] \; h_2 = [ g^\prime_{int} \fatsemi p_2, ~(f^\prime_{int} \fatsemi p_1)|_{m^\prime \setminus m} ]
\end{align*}
The identity of composition is given by the obvious extended cospan of identities.
$\Ecospans$ inherits a symmetric monoidal structure from the coproduct (and initial object) structure of $\catname{CEHyp({\Sigma}})$,  analogously to~\autoref{def:cspd}.

For the same reasons as in the case of $\catname{S}(\Sigma)$ we restrict cospans to MDA cospans.
We also need to enforce typing constraints on hierarchical edges so that they represent well-typed closed $\Sigma^{+}$-terms.
\begin{definition}
	We call a cospan
	$
		n \xrightarrow{f_{\text{ext}}} n^\prime \xrightarrow{f_{\text{int}}} \mathcal{G} \xleftarrow{g^\prime_{\text{int}}} m^\prime \xleftarrow{g_{\text{ext}}} m
	$
	monogamous directed acyclic if
	\begin{itemize}
		\item underlying hypergraph of $\mathcal{G}$ is directed acyclic;
		\item in- and out- degrees of every vertex is at most 1;
		\item $f_{\text{int}}$ and $g_{\text{int}}$ are monomorphisms;
		\item vertices with in-degree (respectively, out-degree) of 0 are precisely the image of $f_{\text{int}}$ (respectively, $g_{\text{int}}$).
	\end{itemize}
	MDA cospans and discrete closed e-hypergraphs form a category $\MdaEcospans$.
\end{definition}

\begin{definition}
	We call a monogamous cospan
	\[
		n \xrightarrow{f_{\text{ext}}} n^\prime \xrightarrow{f_{\text{int}}} \mathcal{G} \xleftarrow{g^\prime_{\text{int}}} m^\prime \xleftarrow{g_{\text{ext}}} m
	\]
	well-typed if all hierarchical edges of $\mathcal{G}$ are well-typed:
	\begin{itemize}
		\item for each $e \in \textcolor{e-color}{E}$ consider sets $I$ and $O$ of its input and output vertices partitioned according to $\consistency_{e}$;
		      $e$ is well-typed if for each element $S$ of the partition of $I$ $w(l_{V}^{*}(S)) = w(l^{*}_{V}(s(e)))$ and similarly for $O$ and $t(e)$;
		\item for each $e \in \textcolor{closed-color}{E}$ consider sets $I$ and $O$ of its input and output vertices;
		      $e$ is well-typed if there exists an object $B \in \text{obj}_{\Sigma_{O}}$ such that $w(l^{*}_{V}(s(e))) \otimes B = w(l^{*}_{V}(I))$ and $B \multimap w(l^{*}_{V}(O)) = w(l^{*}_{V}(t(e)))$.
	\end{itemize}

	We denote the category of well typed MDA cospans as $\WellTypedMdaEcospans$.
\end{definition}

Again, considering~\autoref{fig:e-cospan-example}, to make edge $e_{1}$ well-typed, it is necessary that $l_{V}^{*}([v_{4},v_{5},v_{6}]) = l_{V}^{*}([v_{10},v_{11},v_{12}]) = l_{V}^{*}([v_1,v_2,v_3])$ and $l_{V}^{*}([w_{4},w_{5},w_{6}]) = l_{V}^{*}([w_{10},v_{12},v_{11}]) = l_{V}^{*}([w_1,w_2,w_3])$ (note the change of ordering in the output of $e_{1}$).
To make edge $e_{2}$ well-typed, it is necessary that $l_{V}^{*}([v_4,v_5,v_9]) = l_{V}^{*}([v_7,v_8,v_9])$ and $l_{V}([u_1]) = l_{V}([v_9]) \multimap l_{V}^{*}([v_8,v_7,w_9])$.

\section{DPOI-Rewriting for Closed E-Hypergraphs}

Because extended cospans have a more general notion of interface,  including \emph{internal} vertices,  DPOI rewriting as presented in~\autoref{sec:combinatorial-semantics} needs some adjustments, which we define in this section.
The construction in this section closely follows the one introduced by Tiurin et al.~\cite{tiurin2025equivalencehypergraphsdporewriting} with some adjustments to account for $\textcolor{closed-color}{<}$.
We do not expect internal interfaces to be preserved during rewriting: e.g.\ when the semilattice equations are considered as rewrites.
Thus,  we wish a rewrite rule to be a pair of \emph{extended} cospans of closed e-hypergraphs with matching \emph{external} (but not necessarily internal) interfaces:
\[
	n \xrightarrow{} n^\prime \xrightarrow{} \mathcal{L} \xleftarrow{} m^\prime \xleftarrow{} m,
	\qquad
	n \xrightarrow{} n^{\prime\prime} \xrightarrow{} \mathcal{R} \xleftarrow{} m^{\prime\prime} \xleftarrow{} m.
\]
Analogously to~\autoref{sec:combinatorial-semantics}, observing that the cospan
$
	0 \xrightarrow{} 0 \xrightarrow{} \mathcal{L} \xleftarrow{} n^{\prime} \sqcup m^{\prime} \xleftarrow{} n \sqcup m$ expresses the same data as $
	n \xrightarrow{} n^\prime \xrightarrow{} \mathcal{L} \xleftarrow{} m^\prime \xleftarrow{} m
$
allows us to encode rewrite rules as extended cospans of the form
$
	\mathcal{L} \xleftarrow{} n^\prime \sqcup m^\prime \xleftarrow{} n \sqcup m \xrightarrow{} n^{\prime\prime} \sqcup m^{\prime\prime} \xrightarrow{} \mathcal{R},
$
which fits into the DPO formalism. We implicitly make this identification from hereon.

Before introducing \emph{extended} DPOI rewriting,  the definition of boundary complement must guarantee that rewriting yields a monogamous directed acyclic closed e-hypergraph.
First,  we introduce the  ancillary notion of a \emph{down-closed} graph:
\begin{definition}[Convex down-closed subgraph]
	We call a sub closed e-hypergraph $\mathcal G $ of $\mathcal H$ down-closed if for all $e \in E_{\mathcal{G}}$,   all children of $e$ are also in $\mathcal{G}$.
	A sub e-hypergraph $\mathcal{G}$ is convex, if for all $v,w \in V_{\mathcal{G}}$,  any directed path from $v$ to $w$ in $\mathcal{H}$ consists entirely of vertices and edges in $\mathcal{G}$.
\end{definition}

\begin{definition}[Extended boundary complement]
	\label{def:boundary_new}
	For MDA cospans
	$
		i \xrightarrow{} i^\prime \xrightarrow{} \mathcal{L} \xleftarrow{} j^\prime \xleftarrow{} j
	$ and $
		n \xrightarrow{} n^\prime \xrightarrow{} \mathcal{G} \xleftarrow{} k^\prime \xleftarrow{} k
	$ and mono $m : \mathcal{L} \to \mathcal{G}$ in $\catname{CEHyp(\Sigma)}$, a pushout complement $i \sqcup j \to \mathcal{L}^{\bot} \to \mathcal{G}$ in the square of
	\[
		\scalebox{0.75}{

			\begin{tikzpicture}[scale=1]
	\begin{pgfonlayer}{nodelayer}
		\node [style=none] (0) at (-4.25, 2) {};
		\node [style=none] (1) at (3, 2) {$i \sqcup j$};
		\node [style=none] (2) at (2, 2) {};
		\node [style=none] (3) at (-4.75, 2) {$\mathcal{L}$};
		\node [style=none] (5) at (-0.5, 2) {$i^\prime \sqcup j^\prime$};
		\node [style=none] (9) at (-4.75, 1.5) {};
		\node [style=none] (11) at (-4.75, -0.5) {};
		\node [style=none] (12) at (-4.75, -1) {$\mathcal{G}$};
		\node [style=none] (13) at (3, -1) {$\mathcal{L}^{\bot}$};
		\node [style=none] (14) at (3, -0.5) {};
		\node [style=none] (15) at (3, 1.5) {};
		\node [style=none] (16) at (2.5, -1) {};
		\node [style=none] (17) at (-4.25, -1) {};
		\node [style=none] (23) at (3.5, -1) {};
		\node [style=none] (31) at (-4.25, -1.5) {};
		\node [style=none] (32) at (-2.75, -2) {};
		\node [style=none] (33) at (-1.75, -2.25) {$n^\prime \sqcup k^\prime$};
		\node [style=none] (34) at (3, -3.75) {$n \sqcup k$};
		\node [style=none] (37) at (3, -3.5) {};
		\node [style=none] (39) at (3, -1.5) {};
		\node [style=none] (40) at (2.25, -3.5) {};
		\node [style=none] (41) at (-1, -2.5) {};
		\node [style=none] (44) at (-1.5, 2) {};
		\node [style=none] (45) at (0.5, 2) {};
		\node [style=none] (48) at (-4.25, 0) {};
		\node [style=none] (49) at (-3.75, 0) {};
		\node [style=none] (50) at (-3.75, -0.5) {};
		\node [style=none] (54) at (-5.25, 0.5) {$m$};
		\node [style=none] (55) at (4, 0.5) {$[c_1, c_2]$};
		\node [style=none] (56) at (4, -2.5) {$[d_1,d_2]$};
		\node [style=none] (57) at (-4.25, -2.25) {$[g_{int},h_{int}]$};
		\node [style=none] (58) at (0, -3.5) {$[g_{ext},h_{ext}]$};
		\node [style=none] (59) at (-2.75, 2.5) {$[f^\prime_{int},g^\prime_{int}]$};
		\node [style=none] (60) at (1.25, 2.5) {$[f^\prime_{ext},g^\prime_{ext}]$};
		\node [style=none] (61) at (-0.75, -0.5) {$l$};
	\end{pgfonlayer}
	\begin{pgfonlayer}{edgelayer}
		\draw [style=diredge] (9.center) to (11.center);
		\draw [style=diredge] (15.center) to (14.center);
		\draw [style=diredge] (16.center) to (17.center);
		\draw [style=diredge] (32.center) to (31.center);
		\draw [style=diredge] (37.center) to (39.center);
		\draw [style=diredge] (40.center) to (41.center);
		\draw [style=diredge] (2.center) to (45.center);
		\draw [style=diredge] (44.center) to (0.center);
		\draw (48.center) to (49.center);
		\draw (49.center) to (50.center);
	\end{pgfonlayer}
\end{tikzpicture}
		}
	\]
	is a \emph{boundary complement} if:
	\begin{enumerate}
		\item $m(\mathcal L)$ is a convex down-closed closed e-hypergraph;
		\item $[c_1,c_2]$ is a monomorphism;
		\item for all $v,w$ in $\mathcal{G}$ in the image of $i \sqcup j$,  $v$ and $w$ share the same sequence of predecessors, and either $v,w$ are top-level, or $v \consistency w$, or there exists $e \in \textcolor{closed-color}{E}$ s.t. $v,w \textcolor{closed-color}{\;<^{\mu}\;} e$;
		\item for all $v,w$ in $\mathcal{L}^{\bot}$ in the image of $i \sqcup j$,  $v$ and $w$ share the same sequence of predecessors, and either $v,w$ are top-level or $v \consistency w$, or there exists $e \in \textcolor{closed-color}{E}$ s.t. $v,w \textcolor{closed-color}{\;<^{\mu}\;} e$;
		\item there exist $d_1 : n \to \mathcal{L}^\bot$ and $d_2 : k \to \mathcal{L}^\bot$ making the above triangle commute; and
		\item if the image of $\mathcal{L}$'s external interfaces under $m$ consists exclusively of top-level vertices of $\mathcal{G}$ then there exists a \emph{well-typed} extended MDA cospan
		      \[
			      \hspace{-2em}n \sqcup j \xrightarrow{f_1 \sqcup id_{j}} n^\prime \setminus (i^\prime \setminus i) \sqcup j \xrightarrow{[g_1,c_2]} \mathcal{L}^{\bot} \xleftarrow{[g_2,c_1]} k^\prime \setminus (j^\prime \setminus j) \sqcup i \xleftarrow{f_2 \sqcup id_{i}} k \sqcup i
		      \]
		\item if the image of $\mathcal{L}$'s external interfaces under $m$ consists exclusively of not top-level vertices of $\mathcal{G}$ then there exists a \emph{not necessarily well-typed} extended MDA cospan
		      \[
			      n \xrightarrow{f_1} n^\prime \setminus (i^\prime \setminus i) \sqcup j \xrightarrow{[g_1,c_2]} \mathcal{L}^{\bot} \xleftarrow{[g_2,c_1]} k^\prime \setminus (j^\prime \setminus j) \sqcup i \xleftarrow{f_2} k
		      \]
	\end{enumerate}
	where $f_i$ and $g_i$ are defined as follows.
	The strictly internal interface $i^\prime \setminus i$ of $\mathcal L$ is mapped to the internal interface $n^\prime$ of $\mathcal G$,  since $\mathcal L$ is a down-closed subgraph of $\mathcal G$,  inducing an identification of $i^\prime \setminus i$ in  $n^\prime$.
	Then map $f_1$ is given by $g_{ext}$ with its codomain restricted to $n^\prime \setminus (i^\prime \setminus i)$.\footnote{Noting the image of $g_{ext}$ indeed lies within $n^\prime \setminus (i^\prime \setminus i)$}  The map $g_1$ is derived from the restriction of $g_{int}$ to type $n^\prime \setminus (i^\prime \setminus i) \to \mathcal G$ by further observing that $\mathcal L^\bot$ can be identified within $\mathcal G$ --- and in particular has the internal interface $n^\prime$ of $\mathcal G$ minus $i^\prime \setminus i$.  The maps $f_2$ and $g_2$ are defined similarly.
\end{definition}
Intuitively, the above definition means that when the occurrence of $\mathcal{L}$ within $\mathcal{G}$ is top-level then $\mathcal{L}$'s external interfaces become the part of $\mathcal{L}^{\bot}$'s external interfaces and when the occurrence is nested, those interfaces become the part of $\mathcal{L}^{\bot}$'s strict internal interfaces.
In both cases when removing the occurrence of $\mathcal{L}$ from $\mathcal{G}$, $\mathcal{L}$'s strict internal interfaces are removed from strict internal interfaces of $\mathcal{G}$ to form internal interfaces of the corresponding MDA cospans above.

The conditions (3) and (4) require that top-level connected components in $\mathcal{L}$ do not end up inside edges of different types within $\mathcal{G}$ and $\mathcal{L}^{\bot}$ under $m$ and $[c_1,c_2]$.
Conditions (2) and (7) are the original boundary complement conditions (forgetting the extra interfaces) introduced by Bonchi et al.~\cite{bonchi_string_2022-2}.
Intuitively these conditions ensure that inputs (respectively, outputs) of $\mathcal{R}$ are glued to the outputs (respectively, inputs) of $\mathcal{L}^{\bot}$.
Convexity and down-closed-ness conditions ensures that the subgraph corresponds to a proper subterm modulo SMC.
\begin{proposition}%
	\label{prop:boundary_unique}
	The boundary complement of~\autoref{def:boundary_new} when it exists is unique.
\end{proposition}

We are now ready to define \emph{convex extended DPOI (EDPOI) rewriting} for $\catname{CEHyp({\Sigma})}$.
It is analogous to~\autoref{def:convex_dpo},  except we must construct internal interfaces explicitly.
More precisely, when removing the occurrence of $\mathcal{L}$ from $\mathcal{G}$, \emph{i.e.}, by computing the pushout complement, the internal interfaces of the resulting closed e-hypergraph should be modified to \emph{exclude} the vertices corresponding to $\mathcal{L}$'s strictly internal interfaces (since the vertices they map to have been removed).
Then, when gluing $\mathcal{R}$ into the hole,  the internal interfaces of the resulting closed e-hypergraph should be modified to \emph{include} the strictly internal interfaces of $\mathcal{R}$ (since new internal interfaces for them to map to have been added).
\begin{definition}[Convex EDPOI rewriting]\label{def:dpoi-e}
	Given an extended span of morphisms
	$
		\mathcal{G} \xleftarrow{} n^\prime \sqcup k^\prime \xleftarrow{} n \sqcup k \xrightarrow{} n^{\prime\prime} \sqcup k^{\prime\prime} \xrightarrow{} \mathcal{H}
	$
	in $\catname{CEHyp}(\Sigma)$, we say $\mathcal{G}$ \emph{rewrites (convexly) to} $\mathcal{H}$ (\emph{with external interface} $n \sqcup k$ and \emph{taking internal interface} $n^\prime \sqcup k^\prime$ \emph{to} $n^{\prime\prime} \sqcup k^{\prime\prime}$) --- denoted by $\mathcal{G} \Rrightarrow \mathcal H$  --- \emph{via a rewrite rule}
	$
		\mathcal{L} \xleftarrow{} i^\prime \sqcup j^\prime \xleftarrow{} i \sqcup j \xrightarrow{} i^{\prime\prime} \sqcup j^{\prime\prime} \xrightarrow{} \mathcal{R}
	$
	if there exists an object $\mathcal{C}$ and morphisms which complete the following commutative diagram
	\[
		\scalebox{0.75}{
			\begin{tikzpicture}[scale=1]
	\begin{pgfonlayer}{nodelayer}
		\node [style=none] (0) at (-3, 4) {};
		\node [style=none] (1) at (3, 4) {$i \sqcup j$};
		\node [style=none] (2) at (2, 4) {};
		\node [style=none] (3) at (-3.5, 4) {$\mathcal{L}$};
		\node [style=none] (5) at (-1, 4) {$i' \sqcup j'$};
		\node [style=none] (9) at (-3.5, 3.5) {};
		\node [style=none] (11) at (-3.5, 2.5) {};
		\node [style=none] (12) at (-3.5, 2) {$\mathcal{G}$};
		\node [style=none] (13) at (3, 2) {$\mathcal{C}$};
		\node [style=none] (14) at (3, 2.5) {};
		\node [style=none] (15) at (3, 3.5) {};
		\node [style=none] (16) at (2.5, 2) {};
		\node [style=none] (17) at (-3, 2) {};
		\node [style=none] (18) at (4, 4) {};
		\node [style=none] (19) at (9, 4) {};
		\node [style=none] (20) at (9.5, 4) {$\mathcal{R}$};
		\node [style=none] (21) at (9.5, 2) {$\mathcal{H}$};
		\node [style=none] (22) at (9, 2) {};
		\node [style=none] (23) at (3.5, 2) {};
		\node [style=none] (24) at (9.5, 3.5) {};
		\node [style=none] (25) at (9.5, 2.5) {};
		\node [style=none] (26) at (7, 4) {$i'' \sqcup j''$};
		\node [style=none] (31) at (-3.25, 1.75) {};
		\node [style=none] (32) at (-1.5, 1.25) {};
		\node [style=none] (33) at (-0.75, 1) {$n^\prime \sqcup k^\prime$};
		\node [style=none] (34) at (3, 0) {$n \sqcup k$};
		\node [style=none] (35) at (7, 1) {$n^{\prime\prime} \sqcup k^{\prime\prime}$};
		\node [style=none] (36) at (9.25, 1.75) {};
		\node [style=none] (37) at (3, 0.5) {};
		\node [style=none] (38) at (8, 1.35) {};
		\node [style=none] (39) at (3, 1.5) {};
		\node [style=none] (40) at (2.25, 0.25) {};
		\node [style=none] (41) at (0, 0.85) {};
		\node [style=none] (42) at (3.75, 0.25) {};
		\node [style=none] (43) at (6, 0.85) {};
		\node [style=none] (44) at (-2, 4) {};
		\node [style=none] (45) at (0, 4) {};
		\node [style=none] (46) at (8, 4) {};
		\node [style=none] (47) at (6, 4) {};
		\node [style=none] (48) at (-3, 3) {};
		\node [style=none] (49) at (-2.5, 3) {};
		\node [style=none] (50) at (-2.5, 2.5) {};
		\node [style=none] (51) at (9, 3) {};
		\node [style=none] (52) at (8.5, 3) {};
		\node [style=none] (53) at (8.5, 2.5) {};
		\node [style=none] (54) at (-4, 3) {$m$};
		\node [style=none] (55) at (9, 0.85) {$[f_1,f_2]$};
	\end{pgfonlayer}
	\begin{pgfonlayer}{edgelayer}
		\draw [style=diredge] (9.center) to (11.center);
		\draw [style=diredge] (15.center) to (14.center);
		\draw [style=diredge] (16.center) to (17.center);
		\draw [style=diredge] (23.center) to (22.center);
		\draw [style=diredge] (24.center) to (25.center);
		\draw [style=diredge] (32.center) to (31.center);
		\draw [style=diredge] (38.center) to (36.center);
		\draw [style=diredge] (37.center) to (39.center);
		\draw [style=diredge] (40.center) to (41.center);
		\draw [style=diredge] (42.center) to (43.center);
		\draw [style=diredge] (2.center) to (45.center);
		\draw [style=diredge] (44.center) to (0.center);
		\draw [style=diredge] (18.center) to (47.center);
		\draw [style=diredge] (46.center) to (19.center);
		\draw (48.center) to (49.center);
		\draw (49.center) to (50.center);
		\draw (51.center) to (52.center);
		\draw (52.center) to (53.center);
	\end{pgfonlayer}
\end{tikzpicture}
		}
	\]
	such that the following conditions hold:
	\begin{enumerate}
		\label{dpoi-e:assumptions}
		\item $i \sqcup j \to \mathcal{C} \to \mathcal{G}$ is a boundary complement;
		\item the internal interfaces of $\mathcal H$ are such that:
		      $
			      n^{\prime\prime} = n^\prime \setminus (i^\prime \setminus i) \sqcup (i^{\prime\prime} \setminus i) $ and $ k^{\prime\prime} = k^\prime \setminus (j^\prime \setminus j) \sqcup  (j^{\prime\prime} \setminus j)
		      $
		\item the map $f_1 = [g_1,h_1]\colon n^{\prime\prime} \to \mathcal H$ in the diagram above consists of $g_1$ as defined in~\autoref{def:boundary_new} of boundary complement,  and $h_1$ which is the restriction of the composite $i^{\prime\prime} \to \mathcal R \to \mathcal H$ to $i^{\prime\prime} \setminus i$,  and similarly for the map $f_2\colon k^{\prime\prime} \to \mathcal H$.
	\end{enumerate}
	Given  a set $\mathfrak{R}$ of EDPOI rewrite rules and closed e-hypergraphs with extended interfaces $\mathcal G$ and $\mathcal H$,  we write $\mathcal G \Rrightarrow_\mathfrak{R} \mathcal H$ if there exists a EDPOI rewrite in $\mathfrak R$ such that via it $\mathcal G$ rewrites to $\mathcal H$,  and we write $\mathcal G \Rrightarrow^*_\mathfrak{R} \mathcal H$ for its reflexive, transitive and symmetric closure.
\end{definition}

An example of EDPOI rewrite can be seen in~\autoref{sec:appendix:dpoi}.
The enhanced definition of extended DPOI ensures that the combinatorial representation of morphisms in closed symmetric monoidal $\catname{SLat}$-categories is sound and complete for rewriting of closed $\Sigma^{+}$-terms.

\section{Soundness and completeness}

In this section we explore soundness and completeness of EDPOI rewriting with respect to closed $\Sigma^{+}$-terms.
First we recall the rewriting of $\Sigma$-terms modulo the laws of symmetric monoidal categories.

\begin{definition}[Closed $\Sigma$-term rewrite rule]
	A $\Sigma$-term rewrite rules is a pair $\langle l, r \rangle$ where $l$ and $r$ are $\Sigma$ terms with matching inputs and outputs (as lists of input and output types).
	When $l$ and $r$ are closed $\Sigma$-terms, we call such a rewrite rule closed.
\end{definition}

In practice such rewrite rules are generated by the equations of the corresponding theory.

\begin{definition}%
	\label{def:rewrite}
	For $\Sigma$-terms, $f$ \emph{rewrites} to $g$ modulo the laws of symmetric monoidal categories via a rewrite rule $\langle l, r \rangle$, writing $\leadsto_{\langle l, r \rangle}$, if they are representable as
	$
		f = c_{1} \fatsemi (id_{k} \otimes l) \fatsemi c_{2} $ and $ g = c_{1}\fatsemi(id_{k} \otimes r)\fatsemi c_{2}
	$.
\end{definition}
This notion generalises to rewriting of closed $\Sigma$-terms.

\begin{definition}
	For closed-$\Sigma$-terms $f,g$, and a closed rewrite rule $\langle l, r \rangle$, $f \leadsto_{\langle l, r \rangle} g$ modulo the laws of symmetric monoidal categories if they are representable as either in~\autoref{def:rewrite}, or as
	$
		f = c_{1}\fatsemi(id_{k} \otimes d)\fatsemi c_{2} $ and $ g = c_{1}\fatsemi(id_{k} \otimes e)\fatsemi c_{2}
	$
	and $d = \Lambda(f')$, $e = \Lambda(g')$ and $f' \leadsto_{\langle l, r \rangle} g'$.
\end{definition}
\begin{lemma}
	\label{lemma:normal_form}
	Every closed $\Sigma^{+}$-term $f$ can be represented as
	$
		f_{1} + \ldots + f_{n}
	$ modulo the laws of symmetric monoidal $\catname{SLat}$-categories and distributivity law (\autoref{law:distributivity}) where none of $f_{i}$ contain the join operator.
\end{lemma}
This leads to the following, for rewriting for closed $\Sigma^{+}$-terms.
\begin{definition}
	For closed $\Sigma^{+}$-term $f,g$, and a closed $\Sigma$-term rewrite rule $\langle l, r \rangle$ $f \leadsto_{\langle l, r \rangle} g$ modulo the laws of symmetric monoidal $\catname{SLat}$-categories and distributivity law (\autoref{law:distributivity}) if they are representable as
	$
		f = f_{1} + \ldots + f_{i} + \ldots + f_{n} $ and $ g = f_{1} + \ldots + f_{j} + \ldots + f_{n}
	$
	such that there is a permutation $\sigma$ such that there exist indices $i,j$ and $f_{\sigma(i)} \leadsto f_{j}$.
\end{definition}
Next, to formulate the correspondence between rewriting systems we first define interpretation $\llbracket - \rrbracket$ of closed $\Sigma^{+}$ terms in $\WellTypedMdaEcospans$.
It follows by first defining the interpretation of generators by extending the interpretation of plain $\Sigma$-terms as morphisms in $\MdaCospans$ with the interpretation of $\llbracket \text{ev}_{A,B}\rrbracket$ and $\llbracket \Lambda_{A,B,C} \rrbracket$ by following the work of Alvarez-Picallo et al.~\cite{fscd} (\autoref{sec:appendix:interpretation}).

Then we can make $\WellTypedMdaEcospans$ into $\catname{SLat}$-category by defining a join of two cospans and introducing EDPOI rewrite rules to satisfy the laws of symmetric monoidal $\catname{SLat}$-categories --- one EDPOI rewrite rule for each axiom instance --- as was done by Tiurin et al.~\cite{tiurin2025equivalencehypergraphsdporewriting}.
Under this interpretation, the cospan from~\autoref{fig:e-cospan-example} corresponds to the following closed $\Sigma^{+}$-term $(g\fatsemi f) + (((\Lambda((id \otimes id\fatsemi\sym) \otimes f)) \otimes g)\fatsemi\textbf{ev})$.
Finally, we add a structural rewrite schema rule corresponding to the distributivity law (\autoref{law:distributivity}).

We collect all such rewrite schemas into a set $\mathcal{S}$ and then determine a symmetric monoidal $\catname{SLat}$-category $\WellTypedMdaEcospans / \mathcal{S}$ by this quotienting.
Note that it is not a \emph{closed} symmetric monoidal $\catname{SLat}$-category.
To make it so, we would need to also quotient by the DPO rewrites that arise from the equations in~\autoref{def:closed}.
We do not perform such quotienting as these equations bear operational meaning, i.e.\ they model $\beta$-reduction when interpreting $\lambda$-calculus.
\begin{definition}[Quotient by rewrites]
	Given a set of EDPOI rewrite rules $\mathcal{E}$,  we denote by $\WellTypedMdaEcospans/\mathcal{E}$ the $\catname{SLat}$-category $\WellTypedMdaEcospans$ quotiented by the relation $f \sim g \quad \text{if} \quad f \Rrightarrow^{*}_{\mathcal{E}} g$.
\end{definition}
This explicit quotienting allows us to only consider cospans of the form $f_1 + \ldots + f_{n}$, similar to~\autoref{lemma:normal_form}, such that carrier of $f_{i}$ does not contain any hierarchical edges from $\textcolor{e-color}{E}$.
\begin{proposition}[Proposition 26 and Proposition 27~\cite{fscd}]
	\label{prop:fscd}
	For two closed $\Sigma$-terms $f$ and $g$ and a closed rewrite rule $\langle l, r \rangle$,  $f \leadsto_{\langle l, r \rangle} g$ if and only if $\llbracket f \rrbracket \Rrightarrow_{\langle \llbracket l \rrbracket, \llbracket r \rrbracket \rangle} \llbracket g \rrbracket$, where the latter is defined in a sub-category of $\WellTypedMdaEcospans$ where carriers only contain edges from $\colorbox{yellow}{E} \cup \textcolor{closed-color}{E}$.
\end{proposition}
\begin{proposition}%
	\label{prop:quotient-structural}
	For two closed $\Sigma^{+}$-terms $f$ and $g$ and a closed rewrite rule $\langle l, r \rangle$,  $f \leadsto_{\langle l, r \rangle} g$ if and only if $\llbracket f \rrbracket \Rrightarrow_{\langle \llbracket l \rrbracket, \llbracket r \rrbracket \rangle} \llbracket g \rrbracket$ in $\WellTypedMdaEcospans / \mathcal{S}$.
\end{proposition}
\begin{proof}
	By definition, we have
	$f = f_{1} + \ldots + f_{n} $ and $  g = f_{1} + \ldots + f_{n}$,
	such that $f_{i} \leadsto f_{j}$ for some $i, j$.
	This yields two cospans of the form $\llbracket f_{1} \rrbracket + \ldots + \llbracket f_{n} \rrbracket$ for $\llbracket f \rrbracket$ and $\llbracket g \rrbracket$; because every $f_{i}$ is a $\Sigma$-term, and carriers in every $\llbracket f_{i} \rrbracket$ have no edges from $\textcolor{e-color}{E}$, we can apply~\autoref{prop:fscd}.
\end{proof}
Practically, this quotienting means that we perform DPO rewrites modulo the equations of $\catname{SLat}$-SMC and modulo the distributivity law (\autoref{law:distributivity}).
We can either rewrite the normal forms as mentioned above that contain all possible redexes, or we can genuinely search for redexes modulo these laws by traversing the e-hypergraph through both types of boxes when searching for a redex, which makes this representation practically viable even though it does not absorb all the possible equations.

The proposition above allows us to reason about theories of closed $\Sigma^{+}$-terms $\catname{CS}(\Sigma,\mathcal{E})^{+}$.
Given a set of equations $\mathcal{E}$ of closed $\Sigma$-terms, we construct $\llbracket \mathcal{E} \rrbracket = \{\langle \llbracket l \rrbracket, \llbracket r \rrbracket \rangle, \langle \llbracket r \rrbracket, \llbracket l \rrbracket \rangle \;|\; l = r \;\in\; \mathcal{E}\}$ and then we can perform DPO rewriting in $\WellTypedMdaEcospans / \mathcal{S}$ which formalises, \textit{e.g.}, the example in~\autoref{fig:e-string-substitution}.

\section{Application --- $\lambda^\approx_{\text{lsub}}$-calculus}
\label{sec:application}

In this section, we show the utility of our theory by demonstrating how it can accommodate the \emph{linear substitution calculus} $\lambda_{\text{lsub}}$ of Accattoli et al.~\cite[§ 2]{accattoli2014nonstandard}.
We summarise the linear substitution calculus as the $\lambda$-calculus with an additional term former representing explicit substitution; i.e.\ terms are generated by the grammar
\[
	t \coloneqq x \mid t t \mid \lambda x. t \mid t [x/t],
\]
where $x$ is a variable drawn from some countably infinite set, $t u$ is an application, and $\lambda x. t$ is an abstraction, as in standard $\lambda$-calculus.
The term $t [x / u]$ is an explicit substitution of the variable $x$ for the term $u$ in the term $t$.

The notions of context $C[-]$, bound and free variable are as usual, with the addition that $\textsf{fv} (t [x/u]) \coloneq \textsf{fv} (t) \setminus \{ x \} \cup \textsf{fv} (u)$.
Write $C \llbracket t \rrbracket$ for the term given by substitution of the hole in the context $C[-]$ for the term $t$, under the additional restriction that the free variables of $t$ are not captured (via abstraction or explicit substitution) by $C$.
In addition to the $\beta$- and $\eta$-reduction rules of standard $\lambda$-calculus, there are reduction rules to account for this explicit substitution which \enquote{perform} the substitution or commute it with $\beta$-reduction:
\begin{align*}
	(\lambda x. t) L u              \quad & \to_{d\beta} \quad t[x/u] L,                                                \\
	C \llbracket x \rrbracket [x/u] \quad & \to_{ls} \quad C \llbracket u \rrbracket [x/u],                             \\
	t [x/u]                         \quad & \to_{gc} \quad t,                               & x \notin \textsf{fv} (t),
\end{align*}
where $L$ is a finite list of substitutions of the form $[x/u] [y/v] \ldots$.

In programming terms, this explicit substitution can be used to model \textbf{let}-bindings and sharing.
For instance, the ML-style term $\textbf{let} \; x = 1 \; \textbf{in} \; \text{plus} (x, x)$ is formally represented as the term $\text{plus} (x, x) [x / 1]$ in this calculus\footnote{Technically speaking, in this case we mean the extension of this calculus with flat natural numbers equipped with a binary operation $\text{plus}$.}.
In this case, the variable $x$ is valued in the constant $1$, but one could imagine instead that $x$ is to be substituted for the result of some expensive computation which we do not desire to repeat twice unnecessarily --- in this sense, the sharing of both operands of the operation $\text{plus}$ is captured by the explicit substitution.

\emph{Graphical equivalence} $\sim$ of $\lambda_{\text{lsub}}$-terms is augmentation of $\alpha$-equi\-va\-lence of terms with the additional stipulations that
\begin{align}
	t [x/u] [y/v] \quad        & \sim \quad t [y/v] [x/u],        & x \notin \textsf{fv} (v) \land y \notin \textsf{fv} (u), \nonumber \\
	\label{eq:subst-lambda}
	(\lambda y. t) [x/u] \quad & \sim \quad \lambda y. (t [x/u]), & y \notin \textsf{fv} (u), \tag{$\ast$}                             \\
	(t v) [x/u] \quad          & \sim \quad t [x/u] v,            & x \notin \textsf{fv} (v), \nonumber
\end{align}
closed under contexts, transitivity, symmetry, and reflexivity.
Put simply, it is the stipulation that \textbf{let}-bindings may be freely reordered without changing the meaning of the term\footnote{Recall that our setting is \emph{pure}, in the sense that no sequencing of effects is meaningful.}.
The $\lambda^\sim_{\text{lsub}}$-calculus is given by these graphical equivalence classes of terms with the reduction relation defined by that of the $\lambda_{\text{lsub}}$-calculus modulo graphical equivalence.

This calculus is explicitly designed to capture the representation of terms as proof nets: for two terms $t, u$ in the $\lambda_{\text{lsub}}$-calculus, $t \sim u$ if and only if both map to the same proof net~\cite{accattoliJumpingBoxGraphical2011}.
\begin{proposition}{\cite[§ 5.1.1]{ghica2024stringdiagramslambdacalculifunctional}}
	Let \emph{string diagram equivalence} $\approx$ of $\lambda_{\text{lsub}}$-terms be the extension of graph equivalence $\sim$ by taking the symmetric closure by $\to_{ls}$ and $\to_{gc}$, thus identifying terms under elimination of explicit substitution.
	It defines the corresponding $\lambda^\approx_{\text{lsub}}$-calculus.
	We then have that $t \approx u$ if and only if both map to the same string diagram.
\end{proposition}
This allows us to exploit the string diagram setting of e-hypergraphs to absorb many\footnote{The reason this is not \emph{all} reductions is because e-hypergraphs quotient away only the \emph{structural} isomorphisms of symmetric monoidal closed $\catname{SLat}$-categories (e.g.\ associators, unitors, etc.) which generate the \enquote{isotopy of string diagrams}; even in the graphical equivalence relation $\sim$, while~\autoref{eq:subst-lambda} is validated in every symmetric monoidal closed $\catname{SLat}$-category, it is not so by isotopy, but rather the adjunction $ - \otimes A \dashv A \multimap -$. In other words, terms floating freely through the boundaries of rounded boxes is validated by our semantics, but produces concretely different e-hypergraphs --- this is the purpose of rectification by quotienting along $\mathcal{S}$ in~\autoref{prop:quotient-structural}, and manifests as additional DPO rewrites to be imposed.} of the reductions arising due to the equivalence $\approx$ into our representation rather than explicitly providing them as equations attached to a closed symmetric monoidal theory presenting the $\lambda^\approx_{\text{lsub}}$-calculus.

We will define a monoidal signature whose closed symmetric monoidal theory is the $\lambda^\approx_{\text{lsub}}$-calculus extended by tupling of terms and flat natural numbers equipped with basic arithmetic operations.
Let $\Sigma \coloneq (\{ \mathbf{N} \}, \{ \comonoid, \counit \} \cup \mathbb{N} \cup \{ \text{plus}, \text{mult} \}, t)$ where
\begin{align*}
	t\colon \{ \comonoid, \counit \} \cup \mathbb{N} \cup \{ \text{plus}, \text{mult} \} & \to {\{ \mathbf{N} \}}^* \times {\{ \mathbf{N} \}}^*                                                    \\
	t(n)                                                                                 & = ((), (\mathbf{N})),                                & n \in \mathbb{N},                                \\
	t(o)                                                                                 & = ((\mathbf{N}, \mathbf{N}), (\mathbf{N})),          & o \in \{ \comonoid, \text{plus}, \text{mult} \}, \\
	t(\counit)                                                                           & = ((\mathbf{N}), ()).
\end{align*}
The operations that the signature $\Sigma$ avails are copy and delete, denoted by $\comonoid$ and $\counit$ respectively, constants for each natural number, and the binary arithmetic operations $+$ and $\times$.
We interpret the $\Sigma$-terms $\llbracket \tau \rrbracket$ as follows: operations as themselves, $\id_{I}$ and $\id_{\mathbf{N}}$ as the identity function $\lambda x. x$ on the empty tuple $()$ and $\mathbf{N}$ respectively, $\sym_{\mathbf{N}, \mathbf{N}}$ as the swap function $\lambda (x, y). (y, x)$, composition $\tau_1 \fatsemi \tau_2$ as the composition function $\lambda x. \llbracket \tau_2 \rrbracket \left(\llbracket \tau_1 \rrbracket x\right)$, and $\tau_1 \otimes \tau_2$ as tupling $\left(\llbracket \tau_1 \rrbracket, \llbracket \tau_2 \rrbracket\right)$.

Note that, in contrast to conventional presentations of the $\lambda$-calculus as a language of (slotted~\cite[Example 1]{slotted-egraphs}) e-graphs~\cite[Figure 10]{EggPaper}, here there is no encoding of the term formers for application, abstraction, and explicit substitution.
Instead, these constructs are handled ambiently by our semantic setting: any closed symmetric monoidal category (and in particular, $\catname{CS}(\Sigma, \mathcal{E})$ for any set of equations $\mathcal{E}$) admits an internal language in the multiplicative fragment of intuitionistic linear logic (a particular type of linear $\lambda$-calculus)~\cite[§ 1.7.3]{abramskyIntroductionCategoriesCategorical2010}.
In other words, with respect to~\autoref{fig:egraph-strings}, this allows us to \emph{natively} represent abstraction by the rounded boxes of ($\Lambda(f)$), and application by the $\textsf{ev}$ operation with $\beta$- and $\eta$-conversion arising via equality of morphisms in the category (that is, without having to provide them as data in the form of equations in $\mathcal{E}$).
For explicit substitution, we note that the linearity is not desirable; this is why the copy-delete comonoid $(\comonoid, \counit)$ is present: by Fox's theorem~\cite{foxCoalgebrasCartesianCategories1976}, a Cartesian category is precisely a monoidal one admitting such a \emph{natural} copy-delete structure.
Thus, we must endow $\mathcal{E}$ with equations stating that $(\comonoid, \counit)$ forms a comonoid, and moreover that it is natural: $\comonoid$ behaves uniformly as a copy map, and $\counit$ behaves as deletion.

\begin{figure}
	\centering
	\begin{subfigure}{0.3\linewidth}
		\begin{adjustbox}{scale=0.8}
			\begin{lstlisting}[language=ML]
let x = 1 in
let y = 2 in
plus(x, y)
                \end{lstlisting}
		\end{adjustbox}
		\hrule
		\begin{adjustbox}{scale=0.8}
			\begin{lstlisting}[language=ML]
let y = 2 in
let x = 1 in
plus(x, y)
                \end{lstlisting}
		\end{adjustbox}
		\hrule
		\begin{adjustbox}{scale=0.8}
			\begin{lstlisting}[language=ML]
plus(1, 2)
                \end{lstlisting}
		\end{adjustbox}
		\caption{Three programs for $1 + 2$}%
		\label{fig:plus-1-2-programs}
	\end{subfigure}
	\hspace{1em}
	\begin{subfigure}{0.2\linewidth}
		\[
			\begin{adjustbox}{scale=0.8}$
					\begin{aligned}
						        & \text{plus}(x, y)[y/2][x/1] \\
						\approx &                             \\
						        & \text{plus}(x, y)[x/1][y/2] \\
						\approx &                             \\
						        & \text{plus}(1, 2)
					\end{aligned}
				$\end{adjustbox}
		\]
		\caption{As terms in $\lambda_{\text{lsub}}$}%
		\label{fig:plus-1-2-terms}
	\end{subfigure}
	\hspace{2em}
	\begin{subfigure}{0.3\linewidth}
		\centering
		\scalebox{0.8}{
			\begin{tikzpicture}[scale=2]
	\begin{pgfonlayer}{nodelayer}
		\node [style=round box] (1) at (-0.5, -1) {$1$};
		\node [style=none] (1o) at (-0.5, -0.75) {};
		\node [style=round box] (2) at (0.5, -1) {$2$};
		\node [style=none] (2o) at (0.5, -0.75) {};
		\node [style=round box] (+) at (0, -0) {$+$};
		\node [style=none] (+l) at (-0.1, -0.25) {};
		\node [style=none] (+r) at (0.1, -0.25) {};
		\node [style=none] (out) at (0, 0.75) {};
	\end{pgfonlayer}
	\begin{pgfonlayer}{edgelayer}
		\draw [out=90, in=-90, looseness=1.75] (1o.center) to (+l.center);
		\draw [out=90, in=-90, looseness=1.75] (2o.center) to (+r.center);
		\draw (+.center) to (out.center);
	\end{pgfonlayer}
\end{tikzpicture}
		}
		\caption{String-diagrammatic representation}%
		\label{fig:plus-1-2-string-diagram}
	\end{subfigure}
	\caption{$\approx$-equivalent terms}%
	\label{fig:plus-1-2}
\end{figure}

\autoref{fig:plus-1-2-programs} shows some programs which are semantically equivalent but syntactically distinct.
As alluded to in~\autoref{sec:vs-e-graphs-with-binding}, the traditional (slotted) e-graph approach to modelling the equivalence of these terms (\autoref{fig:plus-1-2-terms}) relies on explicitly encoding substitution operations as e-graph nodes, and their behaviour (the string diagram equivalence relation $\approx$) as rewrite rules, as in~\autoref{fig:e-graph-substitution} --- not only is this more complicated, but it also leads to a blow-up in the state space required to perform equality saturation in any extension of the $\lambda^\approx_{\text{lsub}}$-calculus, because we are really rewriting $\lambda_{\text{lsub}}$-terms and computational resources are spent rewriting between equivalent terms (an empirical analysis is provided in~\autoref{sec:e-graph-explicit-subst}).
In contrast, in our framework, $\lambda^\approx_{\text{lsub}}$-terms are directly modelled as string diagrams (\autoref{fig:plus-1-2-string-diagram}\footnote{Note that dashed boxes are elided for convenience here because no rewriting has occurred.}).
We do however note that the slotted approach is more flexible, in the sense that it may model more theories with binders (e.g.\ $\pi$-calculus), whereas our approach is essentially specialised to $\lambda$-calculus, relying heavily on the metatheoretic internal language of the ambient setting.

As a second example, we may wish to incorporate the optimisation ${\text{plus} (x, x) \leadsto \text{mult} (x, 2)}$ in our toy theory.
This corresponds to endowing $\mathcal{E}$ with the equation of string diagrams:
\[
	\scalebox{0.7}{
		\begin{tikzpicture}[scale=2]
	\begin{pgfonlayer}{nodelayer}
		\node [style=none] (in) at (0, -1.25) {};
		\node [style=vertex] (copy) at (0, -0.75) {};
		\node [style=round box] (+) at (0, 0) {$+$};
		\node [style=none] (+l) at (-0.1, -0.25) {};
		\node [style=none] (+r) at (0.1, -0.25) {};
		\node [style=none] (out) at (0, 0.75) {};
	\end{pgfonlayer}
	\begin{pgfonlayer}{edgelayer}
		\draw [out=120, in=-90] (copy.center) to (+l.center);
		\draw [out=60, in=-90] (copy.center) to (+r.center);
		\draw (+.center) to (out.center);
		\draw (in.center) to (copy.center);
	\end{pgfonlayer}
\end{tikzpicture}
	}
	\quad
	=
	\quad
	\scalebox{0.7}{
		\begin{tikzpicture}[scale=2]
	\begin{pgfonlayer}{nodelayer}
		\node [style=none] (1) at (-0.5, -1.75) {};
		\node [style=none] (1o) at (-0.5, -0.75) {};
		\node [style=round box] (2) at (0.5, -1) {$2$};
		\node [style=none] (2o) at (0.5, -0.75) {};
		\node [style=round box] (+) at (0, -0) {$\times$};
		\node [style=none] (+l) at (-0.1, -0.25) {};
		\node [style=none] (+r) at (0.1, -0.25) {};
		\node [style=none] (out) at (0, 0.75) {};
	\end{pgfonlayer}
	\begin{pgfonlayer}{edgelayer}
		\draw (1.center) to (1o.center);
		\draw [out=90, in=-90, looseness=1.75] (1o.center) to (+l.center);
		\draw [out=90, in=-90, looseness=1.75] (2o.center) to (+r.center);
		\draw (+.center) to (out.center);
	\end{pgfonlayer}
\end{tikzpicture}
	}
	\;
	.
\]
Now, in our semantic model, this rewrite can be non-destructively applied to model one step of e-graph saturation:
\[\adjustbox{scale=0.5}{
		\begin{tikzpicture}
	\begin{pgfonlayer}{nodelayer}
		\node [style=round box] (0) at (-2.5, 1) {$+$};
		\node [style=vertex set] (1) at (-2.5, 3) {};
		\node [style=vertex] (2) at (-2.5, -1.5) {};
		\node [style=vertex set] (3) at (-2.5, -3) {};
		\node [style=none] (4) at (-4, 3) {};
		\node [style=none] (5) at (-1, 3) {};
		\node [style=none] (6) at (-4, -3) {};
		\node [style=none] (7) at (-1, -3) {};
		\node [style=none] (8) at (-2.5, 4) {};
		\node [style=none] (9) at (-2.5, -4) {};
	\end{pgfonlayer}
	\begin{pgfonlayer}{edgelayer}
		\draw (1) to (0);
		\draw [in=120, out=-120] (0) to (2);
		\draw [in=60, out=-60] (0) to (2);
		\draw (2) to (3);
		\draw [style=black dash] (5.center)
			 to (7.center)
			 to (6.center)
			 to (4.center)
			 to cycle;
		\draw (8.center) to (1);
		\draw (3) to (9.center);
	\end{pgfonlayer}
\end{tikzpicture}
	}
	\quad
	\equiv
	\quad
	\adjustbox{scale=0.5}{
		\begin{tikzpicture}
	\begin{pgfonlayer}{nodelayer}
		\node [style=round box] (0) at (-2.5, 1) {$+$};
		\node [style=none] (1) at (-2.5, 3) {};
		\node [style=vertex] (2) at (-2.5, -1.5) {};
		\node [style=none] (3) at (-2.5, -3) {};
		\node [style=none] (4) at (-4, 3) {};
		\node [style=vertex set] (5) at (-1, 3) {};
		\node [style=none] (6) at (-4, -3) {};
		\node [style=none] (7) at (-1, -3) {};
		\node [style=round box] (10) at (0.5, 1) {$+$};
		\node [style=none] (11) at (0.5, 3) {};
		\node [style=vertex] (12) at (0.5, -1.5) {};
		\node [style=none] (13) at (0.5, -3) {};
		\node [style=none] (20) at (2, 3) {};
		\node [style=none] (21) at (2, -3) {};
		\node [style=vertex set] (22) at (-1, -3) {};
		\node [style=none] (23) at (-1, -4) {};
		\node [style=none] (24) at (-1, 4) {};
	\end{pgfonlayer}
	\begin{pgfonlayer}{edgelayer}
		\draw (1.center) to (0);
		\draw [in=120, out=-120] (0) to (2);
		\draw [in=60, out=-60] (0) to (2);
		\draw (2) to (3.center);
		\draw [style=black dash] (5.center)
			 to (7.center)
			 to (6.center)
			 to (4.center)
			 to cycle;
		\draw (11.center) to (10);
		\draw [in=120, out=-120] (10) to (12);
		\draw [in=60, out=-60] (10) to (12);
		\draw (12) to (13.center);
		\draw [style=black dash] (5) to (20.center);
		\draw [style=black dash] (20.center) to (21.center);
		\draw [style=black dash] (21.center) to (7.center);
		\draw (22) to (23.center);
		\draw (24.center) to (5);
	\end{pgfonlayer}
\end{tikzpicture}
	}
	\quad
	\to
	\quad
	\adjustbox{scale=0.5}{
		\begin{tikzpicture}
	\begin{pgfonlayer}{nodelayer}
		\node [style=round box] (0) at (-2.5, 1) {$+$};
		\node [style=none] (1) at (-2.5, 3) {};
		\node [style=vertex] (2) at (-2.5, -1.5) {};
		\node [style=none] (3) at (-2.5, -3) {};
		\node [style=none] (4) at (-4, 3) {};
		\node [style=vertex set] (5) at (-1, 3) {};
		\node [style=none] (6) at (-4, -3) {};
		\node [style=none] (7) at (-1, -3) {};
		\node [style=round box] (10) at (0.5, 0.5) {$\times$};
		\node [style=none] (11) at (0.5, 3) {};
		\node [style=none] (13) at (-0.25, -3) {};
		\node [style=none] (20) at (2.75, 3) {};
		\node [style=none] (21) at (2.75, -3) {};
		\node [style=vertex set] (22) at (-1, -3) {};
		\node [style=none] (23) at (-1, -4) {};
		\node [style=none] (24) at (-1, 4) {};
		\node [style=round box] (25) at (1.5, -1.5) {$2$};
	\end{pgfonlayer}
	\begin{pgfonlayer}{edgelayer}
		\draw (1.center) to (0);
		\draw [in=120, out=-120] (0) to (2);
		\draw [in=60, out=-60] (0) to (2);
		\draw (2) to (3.center);
		\draw [style=black dash] (5.center)
			 to (7.center)
			 to (6.center)
			 to (4.center)
			 to cycle;
		\draw (11.center) to (10);
		\draw [style=black dash] (5) to (20.center);
		\draw [style=black dash] (20.center) to (21.center);
		\draw [style=black dash] (21.center) to (7.center);
		\draw (22) to (23.center);
		\draw (24.center) to (5);
		\draw [in=90, out=-30] (10) to (25);
		\draw [in=90, out=-120] (10) to (13.center);
	\end{pgfonlayer}
\end{tikzpicture}
	}
	\;
	.
\]
Here, $\equiv$ is a definitional equality arising from the idempotence in the Hom semilattice, and $\to$ is an application of DPO rewriting using the previous equation along the right branch.
In isolation, this isn't particularly interesting, but the important observation is that this rewriting all happens \emph{locally} in the context of a bigger diagram.

From a practical perspective, the encoding of string diagrams as (e-)hypergraphs provides a normal form for string diagram equivalence classes of terms.
As a data structure, this provides an efficient representation upon which to build an implementation of this theory as a fully fledged \enquote{e-graphs for string diagram rewriting}, using DPO rewriting as a primitive to implement the standard e-graph operations.
We leave this as future work, noting its applications to not only $\lambda^\approx_{\text{lsub}}$-calculus, but also other monoidal theories where there are two levels of rewriting:
\begin{enumerate}
	\item \enquote{bureaucratic} reassociation and reparenthesisation, absorbed into isotopy of string diagrams (here, a fragment of string diagram equivalence of $\lambda_{\text{lsub}}$-terms);
	\item algebraic identities arising from the semantics of the theory (e.g. ${\text{plus} (x, x) \leadsto \text{mult} (x, 2)}$ for arithmetic).
\end{enumerate}
One such theory is the ZX calculus~\cite{coeckeInteractingQuantumObservables2011}.


\bibliographystyle{eptcs}
\bibliography{bibliography}

\appendix

\section{$\catname{SLat}$}
\label{sec:appendix:slat}

In this section we will define the category of semilattices that we use as a base for enrichment throughout the paper.

\begin{definition}[Semilattice]
	A \textit{semilattice} $(S,\join)$ is a set $S$ equipped with a binary operation $\join : S \times S \to S$ satisfying the following properties:
    \begin{align*}
        a \join (b \join c) &\equiv (a \join b) \join c \quad \text{(associativity)}\\
        a \join b &\equiv b \join a \quad \text{(commutativity)}\\
        a \join a &\equiv a \quad \text{(idempotence)}
    \end{align*}
\end{definition}

Note that we do not require the existence of a unit for $+$, adopting this requirement would mean that every hom-object must be non-empty.
Semilattices that satisfy this extra requirement are sometimes called \textit{bounded}, i.e., they are idempotent commutative monoids.

\begin{definition}[Semilattice homomorphism]
	A homomorphism between two semilattices $S_{1}$ and $S_{2}$ is a map $\phi : S_1 \to S_2$ such that $\phi(s \join_{S_{1}} s') = \phi(s) \join_{S_{2}} \phi(s')$ for all $s,s' \in S_{1}$.
\end{definition}

\begin{definition}[Category of semilattices]
	Semilattices with their respective homomorphisms form a category that we denote $\catname{SLat}$.
\end{definition}

\begin{definition}\label{def:slat-tensor}
We define a tensor product of semilattices $S_1$ and $S_2$ as follows.
$S_1 \otimes S_2$ is the free semilattice on the set $S_1 \times S_2$ --- finite non-empty formal joins of pairs $(s_1,s_2)$ with $s_{1} \in S_{1}$, $s_2 \in S_{2}$ --- quotiented by the following bilinearity relations
\begin{itemize}
    \item $(s_1 \join_{S_{1}} s_1', s_2) \equiv (s_1,s_2) \join_{S_{1} \otimes S_{2}} (s'_1, s_2)$
    \item $(s_1, s_2 \join_{S_{2}} s'_2) \equiv (s_1, s_2) \join_{S_{1} \otimes S_{2}} (s_1,s'_2)$
\end{itemize}
Writing $\eta_{S_1,S_2} : S_1 \times S_2 \to S_1 \otimes S_2$, $(s_1,s_2) \mapsto [(s_1,s_2)]$, for the inclusion of pairs, the tensor product is characterised by the universal property that every \emph{bimorphism} out of $S_1 \times S_2$---a function $b : S_1 \times S_2 \to T$ that is a homomorphism in each variable separately---factors as $b = \bar{b} \circ \eta_{S_1,S_2}$ for a unique homomorphism $\bar{b} : S_1 \otimes S_2 \to T$.
On morphisms, given $f : S_1 \to T_1$ and $g : S_2 \to T_2$, the composite $\eta_{T_1,T_2} \circ (f \times g)$ is a bimorphism, and we define $f \otimes g : S_1 \otimes S_2 \to T_1 \otimes T_2$ to be the unique homomorphism with
\[
(f \otimes g) \circ \eta_{S_1,S_2} = \eta_{T_1,T_2} \circ (f \times g).
\]
The unit for this tensor product is $\I = \{*\}$ --- a one-element semilattice.
\end{definition}

\begin{remark}\label{rem:slat-tensor-functorial}
The action of $\otimes$ on morphisms is functorial: $\id \otimes \id = \id$ and
\[
(f \otimes g) \circ (h \otimes k) = (f \circ h) \otimes (g \circ k).
\]
Both equations follow from the uniqueness in the universal property, since each side restricts along $\eta$ to the same bimorphism; for the second,
\[
\begin{aligned}
(f \otimes g) \circ (h \otimes k) \circ \eta
&= (f \otimes g) \circ \eta \circ (h \times k) \\
&= \eta \circ (f \times g) \circ (h \times k) \\
&= \eta \circ \big((f \circ h) \times (g \circ k)\big),
\end{aligned}
\]
which is exactly the bimorphism defining $(f \circ h) \otimes (g \circ k)$. By the same reasoning the unitors $S \otimes \I \cong S$ (given by $(s,*) \mapsto s$), the associators and the symmetry are the unique homomorphisms induced by the evident bimorphisms, and the coherence axioms hold because the corresponding diagrams already commute at the level of bimorphisms, where they reduce to coherence for the cartesian product in $\catname{Set}$. This equips $\catname{SLat}$ with a symmetric monoidal structure.
\end{remark}

\begin{remark}
    The above definition of a tensor product should not be confused with a tensor given by $\times$, which is constructed in the same vein as the tensor product above with the additional quotienting
    \[
        (s_1,s_2) \join_{S_1 \times S_2} (s_1',s_2') \equiv (s_1 \join_{S_1} s_1', s_2 \join_{S_2} s_2')~.
    \]
    The reason for using the former construction is that it has better theoretical properties, in particular, the free semilattice functor is strong monoidal (\cref{cor:slat-free-monoidal}).
\end{remark}

\begin{definition}\label{def:slat-internal-hom}
For semilattices $S_1$ and $S_2$, the \emph{internal hom} $[S_1, S_2]$ is the set $\catname{SLat}(S_1, S_2)$ of homomorphisms $S_1 \to S_2$, made into a semilattice by the pointwise join
\[
(f \join g)(x) \defeq f(x) \join g(x).
\]
This is well-defined, as $f \join g$ is again a homomorphism:
\begin{align*}
    (f \join g)(x \join y) &= f(x \join y) \join g(x \join y)\\
    &= f(x) \join f(y) \join g(x) \join g(y)\\
    &= f(x) \join g(x) \join f(y) \join g(y)\\
    &= (f \join g)(x) \join (f \join g)(y).
\end{align*}
\end{definition}
  \begin{proposition}\label{prop:slat-smcc}
    $\catname{SLat}$ is a closed symmetric monoidal category.
  \end{proposition}
  \begin{proof}
    The symmetric monoidal structure is given by the tensor product of \cref{def:slat-tensor}, with functoriality and coherence established in \cref{rem:slat-tensor-functorial}.
    For closedness, take the internal hom $[S_2, T]$ of \cref{def:slat-internal-hom}; the universal property of the tensor gives bijections, natural in $S_1$ and $T$,
    \[
    \catname{SLat}(S_1 \otimes S_2, T) \;\cong\; \{\text{bimorphisms } S_1 \times S_2 \to T\} \;\cong\; \catname{SLat}(S_1, [S_2, T]),
    \]
    where the first is the universal property of $\otimes$ and the second is currying: a bimorphism $b : S_1 \times S_2 \to T$ corresponds to the homomorphism $s_1 \mapsto b(s_1, -)$ into $[S_2, T]$.
    Hence $- \otimes S_2 \dashv [S_2, -]$, so $\catname{SLat}$ is closed.
  \end{proof}

\begin{remark}
The above is a special case of the result by Keigher~\cite{keigher_symmetric_nodate} that says that a category of Eilenberg-Moore algebras induced by a commutative monad on $\catname{Set}$ is canonically symmetric monoidal closed.
\end{remark}
\begin{definition}
	The forgetful functor $U : \catname{SLat} \to \catname{Set}$ is given by $\catname{SLat}(\I_{\catname{SLat}}, -) : \catname{SLat} \to \catname{Set}$.
\end{definition}

Intuitively, the above functor returns the underlying set of a given semilattice $S$ as each morphism from $\{*\} \to S$ picks out an element of $S$.

\begin{proposition}[Special case of Proposition 6.4.6~\cite{Borceux_1994}]\label{prop:slat-free}
	The forgetful functor $U : \catname{SLat} \to \catname{Set}$ has a left adjoint free functor $F : \catname{Set} \to \catname{SLat}$.
\end{proposition}
\begin{proof}
	The functor $F$ is defined by letting $F(A) = \coprod_{A} \I_{\catname{SLat}}$, the coproduct of $A$-many copies of the unit semilattice.
	The coproduct $S_{1} \coprod S_{2}$ of two semilattices has as elements the formal joins $s_{1} \join s_{2} \join \ldots \join s_{n}$ with each $s_{i}$ drawn from $S_{1}$ or $S_{2}$, quotiented by the semilattice axioms together with the relations of $S_{1}$ and $S_{2}$.
	The adjunction is then given by the following natural isomorphisms
	\begin{align*}
		\catname{SLat}(\coprod_{A}(\I_{\catname{SLat}}), B) & \cong \prod_{A}(\catname{SLat}(\I_{\catname{SLat}}, B))       \\
		                                                   & \cong \catname{Set}(A,\catname{SLat}(\I_{\catname{SLat}}, B)) \\
		                                                   & \cong \catname{Set}(A, U(B))
	\end{align*}
	The first isomorphism holds because the contravariant hom-functor sends colimits in its first argument to limits, so the coproduct becomes a product.
	The second is the isomorphism $\prod_{A} X \cong \catname{Set}(A, X)$. By the Yoneda lemma it suffices to give a bijection natural in $Y$:
	\[
	\catname{Set}\big(Y, \textstyle\prod_A X\big) \cong \prod_A \catname{Set}(Y, X) \cong \catname{Set}\big(\textstyle\coprod_A Y, X\big) \cong \catname{Set}(Y \times A, X) \cong \catname{Set}\big(Y, \catname{Set}(A,X)\big),
	\]
	using, in order, the universal property of the product, that of the coproduct (with $\coprod_A Y \cong Y \times A$), and currying.
    The last bijection is the definition of $U$.
	As all three bijections are natural in $B$, their composite exhibits $F \dashv U$.
\end{proof}

\begin{corollary}\label{cor:slat-free-monoidal}
	The free functor $F : \catname{Set} \to \catname{SLat}$ is strong monoidal as a functor $(\catname{Set}, \times) \to (\catname{SLat}, \otimes)$.
\end{corollary}
\begin{proof}
	We have $F(\I_{\catname{Set}}) \cong \I_{\catname{SLat}}$, and
	\begin{align*}
		F(A) \otimes F(B) & \cong (\coprod_{A} \I_{\catname{SLat}}) \otimes (\coprod_{B} \I_{\catname{SLat}}) \\
		                  & \cong \coprod_{A} (\I_{\catname{SLat}} \otimes \coprod_{B} \I_{\catname{SLat}})   \\
		                  & \cong \coprod_{A} (\coprod_{B} (\I_{\catname{SLat}} \otimes \I_{\catname{SLat}})) \\
		                  & \cong \coprod_{A} (\coprod_{B} \I_{\catname{SLat}})                              \\
		                  & \cong \coprod_{A \times B} \I_{\catname{SLat}}                                   \\
		                  & \cong F(A \times B),
	\end{align*}
	naturally in $A$ and $B$, where we used that $- \otimes X$ and $X \otimes -$ preserve colimits, being left adjoints by the symmetry and monoidal closedness of $\catname{SLat}$ (\cref{prop:slat-smcc}).
\end{proof}

While a monoidal base suffices to define enriched categories, we assume that the base is symmetric monoidal closed, so that the standard category-theoretic machinery (such as the tensoring of enriched categories) is available in the enriched setting.

\section{String diagram rewrite theory}
\label{sec:appendix:hyp}

\subsection{Hypergraphs}

\begin{proposition}[\cite{bonchi_string_2022-1}]
$\catname{Hyp}(\Sigma)$ is equivalently the slice category $\catname{Hyp} \downarrow G_{\Sigma}$, where $\catname{Hyp}$ is the category of (unlabelled) hypergraphs --- equivalently, the category of finite presheaves over the index category $\mathbf{I}$ whose objects are the pairs $(k, l) \in \mathbb{N} \times \mathbb{N}$ together with a distinguished object $\star$, with $k + l$ morphisms $\{ i\colon \star \to (k, l) \mid 0 \leq i < k + l \}$ --- and $G_{\Sigma}$ is the hypergraph with a single vertex $v$ and, for each generator $g : n \to m$ of $\Sigma$, one hyperedge $e_g$ in which $v$ occurs $n$ times in $s(e_g)$ and $m$ times in $t(e_g)$.
\end{proposition}
Intuitively, $\mathbf{I}$ presents a hypergraph by its vertices (the object $\star$) and its hyperedges sorted by arity (an object $(k,l)$ is the sort of hyperedges with $k$ sources and $l$ targets), the $k+l$ maps picking out each source and target vertex.
The hyperedges of $G_{\Sigma}$ are then the available labels, and a morphism $H \to G_{\Sigma}$ in $\catname{Hyp}$ labels each hyperedge of $H$ by the generator it is sent to, automatically matching arities.
As a presheaf category, $\catname{Hyp}$ has all finite colimits --- the initial object is the empty hypergraph and coproducts are disjoint unions --- and these are created by the forgetful functor out of the slice, so $\catname{Hyp}(\Sigma)$ inherits them.

The single vertex of $G_{\Sigma}$ reflects that $\Sigma$ here carries a single object type.
For a many-sorted signature with a set of types (colours) $C$, one works instead with the category $\catname{Hyp}_{C}$ of $C$-\emph{coloured} hypergraphs --- finite presheaves over the analogous index category $\mathbf{I}_{C}$ in which the single vertex object $\star$ is replaced by one object $\star_{c}$ per colour $c \in C$, and the arity objects record the colour of each source and target --- and takes $G_{\Sigma}$ to have one vertex per colour, attaching each generator's sources and targets to the vertices of the corresponding colours.
The resulting slice $\catname{Hyp}_{C} \downarrow G_{\Sigma} =: \catname{Hyp}_{C,\Sigma}$ then enforces the typing discipline of~\autoref{def:hypergraph} for free, since a labelling morphism must respect sources and targets and hence colours.

\begin{proposition}
	$\HypI{\Sigma}$ is a symmetric monoidal category.
\end{proposition}
\begin{proof}
Composition of two cospans $n \xrightarrow{f} \mathcal{F} \xleftarrow{g} m \fatsemi m \xrightarrow{f^\prime} \mathcal{G} \xleftarrow{g^\prime} k$ is given by the pushout
$
	n \xrightarrow{f} \mathcal{F} \sqcup_{g,f^\prime} \mathcal{G} \xleftarrow{g^\prime} k
$,
computed by identifying vertices in $\mathcal{F} \sqcup \mathcal{G}$ that share pre-image in $m$.
The latter is a pushout in $\catname{Hyp}(\Sigma)$ and always exists since $\catname{Hyp}(\Sigma)$ has all finite colimits.
	\begin{enumerate}
	\item the tensor product given by coproduct in $\catname{Hyp}(\Sigma)$:
	      \begin{align*}
		      n \xrightarrow{f}               & \mathcal{F} \xleftarrow{g} m                                                                                                                                       \\
		                                      & \otimes \qquad \quad \Coloneqq n \sqcup n^\prime \xrightarrow{f \sqcup f^\prime} \mathcal{F} \sqcup \mathcal{G} \xleftarrow{g \sqcup g^\prime} m \sqcup m^\prime ; \\
		      n^\prime \xrightarrow{f^\prime} & \mathcal{G} \xleftarrow{g^\prime} m^\prime
	      \end{align*}
	\item the monoidal unit is given by the initial object of $\catname{Hyp}(\Sigma)$;
	\item symmetry is inherited from coproduct in $\catname{Hyp}(\Sigma)$~\cite{MonoidalCoproduct}.
\end{enumerate}
\end{proof}

\begin{definition}[Category of MDA Hypergraphs with Interfaces~\cite{bonchi_string_2022-2}]%
	\label{def:monogamy_hyp}
	We call a cospan $n \xrightarrow{f} \mathcal{G} \xleftarrow{g} m$ in $\HypI{\Sigma}$ monogamous if:
	\begin{enumerate}
		\item $\mathcal{G}$ is directed acyclic;
		\item $f$ and $g$ are monomorphisms;
		\item the in-degree and out-degree of every vertex is at most~$1$;
		\item vertices with in-degree $0$ are precisely the image of $f$;
		\item vertices with out-degree $0$ are precisely the image of $g$.
	\end{enumerate}
	We denote by $\MdaCospans$ the subcategory of $\HypI{\Sigma}$ of monogamous directed acyclic (MDA) cospans.
\end{definition}

\autoref{fig:hypergraph-with-interfaces} gives an example of an MDA cospan, which is such because it \enquote{looks like} the string diagram \comonoid.
Analogously, we can see that~\autoref{fig:hypergraph} (equipped with the interfaces choosing $\{ v_1 \}$ and $\{ v_4, v_5, v_6, v_7 \}$) fails to form an MDA cospan --- it is not monogamous for many reasons: $v_1$ has out-degree 0, $v_1, v_5, v_7$ have in-degree 2, and $v_3$ has in-degree 0; additionally, it fails be be directed acyclic due to the path along $e_4$ from $v_7$ to itself\footnote{Recall that loops are not part of the syntax of string diagrams for symmetric monoidal categories in the absence of a trace --- a relaxation of monogamy to allow for this is explored by Ghica et al.~\cite{ghica_rewriting_2023}.}.

\begin{figure}
	\begin{subfigure}[c]{0.3\linewidth}
		\centering


        \[
        \adjustbox{scale=0.8}{
        \begin{tikzpicture}
	\begin{pgfonlayer}{nodelayer}
		\node [style=hedge] (0) at (-0.5, 7.5) {$e_3$};
		\node [style=hedge] (1) at (2.5, 7.5) {$e_4$};
		\node [style=hedge] (2) at (-0.5, 3.5) {$e_1$};
		\node [style=hedge] (3) at (2.5, 4) {$e_2$};
		\node [style=node, label={above:$v_4$}] (4) at (-1.5, 9.5) {};
		\node [style=node, label={above:$v_5$}] (5) at (0.5, 9.5) {};
		\node [style=node, label={above:$v_6$}] (6) at (2.5, 9.5) {};
		\node [style=node, label={above:$v_7$}] (7) at (4.5, 9.5) {};
		\node [style=node, label={left:$v_2$}] (8) at (-0.5, 5.5) {};
		\node [style=node, label={right:$v_3$}] (9) at (2.5, 5.5) {};
		\node [style=node, label={below:$v_1$}] (10) at (2.5, 2.5) {};
		\node [style=none] (11) at (2.5, 7) {};
	\end{pgfonlayer}
	\begin{pgfonlayer}{edgelayer}
		\draw [in=120, out=-120] (8) to (2);
		\draw [in=-60, out=60] (2) to (8);
		\draw [in=120, out=-90] (4) to (0);
		\draw [in=60, out=-90] (5) to (0);
		\draw (1) to (5);
		\draw (6) to (1);
		\draw (7) to (1);
		\draw [in=-90, out=-30, looseness=0.75] (11.center) to (7);
		\draw (0) to (8);
		\draw (11.center) to (9);
	\end{pgfonlayer}
\end{tikzpicture}
        }
        \]
		\caption{Hypergraph with 7 vertices and 4 hyperedges, all directed bottom-to-top; source/target vertices are ordered left-to-right.}%
		\label{fig:hypergraph}
	\end{subfigure}
	\hfill
	\begin{subfigure}[c]{0.3\linewidth}
		\centering



        \[
        \adjustbox{scale=0.8}{
        \begin{tikzpicture}
	\begin{pgfonlayer}{nodelayer}
		\node [style=none] (19) at (-3, -2.5) {};
		\node [style=none] (20) at (2, -2.5) {};
		\node [style=none] (21) at (2, 4) {};
		\node [style=none] (22) at (-3, 4) {};
		\node [style=node, label={below:$v_1$}] (4) at (-0.5, -5.75) {};
		\node [style=node, label={above:$v_2$}] (7) at (-1.5, 7) {};
		\node [style=node, label={above:$v_3$}] (8) at (0.5, 7) {};
		\node [style=none] (9) at (-2.25, 6.5) {};
		\node [style=none] (10) at (-2.25, 8) {};
		\node [style=none] (11) at (1.5, 8) {};
		\node [style=none] (12) at (1.5, 6.5) {};
		\node [style=none] (13) at (-1.5, -6.75) {};
		\node [style=none] (14) at (-1.5, -5.25) {};
		\node [style=none] (15) at (0.5, -5.25) {};
		\node [style=none] (16) at (0.5, -6.75) {};
		\node [style=none] (17) at (-0.5, -4.75) {};
		\node [style=none] (18) at (-0.5, -3.25) {};
		\node [style=hedge] (0) at (-0.5, 0.5) {$e_1$};
		\node [style=node, label={above:$v_2$}] (1) at (-1.5, 2.5) {};
		\node [style=node, label={above:$v_3$}] (2) at (0.5, 2.5) {};
		\node [style=node, label={below:$v_1$}] (3) at (-0.5, -1.5) {};
		\node [style=none] (23) at (-0.5, 6) {};
		\node [style=none] (24) at (-0.5, 4.5) {};
		\node [style=none] (25) at (0.25, 5.25) {$g$};
		\node [style=none] (26) at (0.25, -4) {$f$};
	\end{pgfonlayer}
	\begin{pgfonlayer}{edgelayer}
		\draw [style=graph frame] (21.center)
			 to (22.center)
			 to (19.center)
			 to (20.center)
			 to cycle;
		\draw [in=-90, out=120] (0) to (1);
		\draw [in=60, out=-90] (2) to (0);
		\draw (0) to (3);
		\draw [style=boundary frame] (12.center)
			 to (9.center)
			 to (10.center)
			 to (11.center)
			 to cycle;
		\draw [style=boundary frame] (13.center)
			 to (14.center)
			 to (15.center)
			 to (16.center)
			 to cycle;
		\draw [style=diredge] (23.center) to (24.center);
		\draw [style=diredge] (17.center) to (18.center);
	\end{pgfonlayer}
\end{tikzpicture}
        }
        \]
		\caption{Hypergraph with interfaces}%
		\label{fig:hypergraph-with-interfaces}
	\end{subfigure}
	\caption{Example hypergraphs}%
	\label{fig:hypergraph-examples}
\end{figure}

\subsection{DPOI-Rewriting for Hypergraphs with Interfaces}

\begin{definition}[Boundary complement]%
	\label{def:boundary_original}
	For MDA cospans $i \xrightarrow{a_1} \mathcal{L} \xleftarrow{a_2} j$ and $n \xrightarrow{b_1} \mathcal{G} \xleftarrow{b_2} m$ and monomorphism $f\colon \mathcal L \to \mathcal G$, a pushout complement $i \sqcup j \to \mathcal{L}^\bot \to \mathcal{G}$ as depicted in the square below
	\[
		\adjustbox{scale=0.8}{
			\begin{tikzpicture}
	\begin{pgfonlayer}{nodelayer}
		\node [style=none] (0) at (-3.5, 4.5) {$\mathcal{L}$};
		\node [style=none] (1) at (-2.5, 4.5) {};
		\node [style=none] (2) at (-0.5, 4.5) {};
		\node [style=none] (3) at (0.5, 4.5) {$i \sqcup j$};
		\node [style=none] (7) at (-3.5, 3.5) {};
		\node [style=none] (8) at (-3.5, 1.5) {};
		\node [style=none] (9) at (-3.5, 0.5) {$\mathcal{G}$};
		\node [style=none] (10) at (-0.5, 0.5) {};
		\node [style=none] (11) at (-2.5, 0.5) {};
		\node [style=none] (12) at (0.5, 0.5) {$\mathcal{L}^{\bot}$};
		\node [style=none] (16) at (0.5, 3.5) {};
		\node [style=none] (17) at (0.5, 1.5) {};
		\node [style=none] (20) at (-4.25, 2.5) {$f$};
		\node [style=none] (21) at (-3, 1.5) {};
		\node [style=none] (22) at (-2.5, 1.5) {};
		\node [style=none] (23) at (-2.5, 1) {};
		\node [style=none] (27) at (0.5, -2.5) {$n \sqcup m$};
		\node [style=none] (28) at (-3, 0) {};
		\node [style=none] (29) at (0, -2) {};
		\node [style=none] (32) at (0.5, 0) {};
		\node [style=none] (33) at (0.5, -2) {};
		\node [style=none] (34) at (-1.5, 5.25) {$[a_1,a_2]$};
		\node [style=none] (35) at (-2.25, -1.5) {$[b_1,b_2]$};
		\node [style=none] (36) at (1.75, -1) {$[d_1,d_2]$};
		\node [style=none] (37) at (1.75, 2.5) {$[c_1,c_2]$};
	\end{pgfonlayer}
	\begin{pgfonlayer}{edgelayer}
		\draw [style=diredge] (2.center) to (1.center);
		\draw [style=diredge] (7.center) to (8.center);
		\draw [style=diredge] (10.center) to (11.center);
		\draw [style=diredge] (16.center) to (17.center);
		\draw (21.center) to (22.center);
		\draw (22.center) to (23.center);
		\draw [style=diredge] (29.center) to (28.center);
		\draw [style=diredge] (33.center) to (32.center);
	\end{pgfonlayer}
\end{tikzpicture}
		}\]
	is called a boundary complement if $[c_1, c_2]$ is a monomorphism and there exist $d_1\colon n \to \mathcal{L}^{\bot}$ and $d_2\colon m \to \mathcal{L}^{\bot}$ making the above triangle commute and such that
	$
		n \sqcup j \xrightarrow{[d_1,c_2]} \mathcal{L}^{\bot} \xleftarrow{[d_2,c_1]} m \sqcup i
	$
	is an MDA cospan.
\end{definition}

Intuitively, requiring boundary complements ensures that inputs are glued exclusively to outputs and vice-versa in the two pushout squares.

\begin{example}[Example 29~\cite{bonchi_string_2022-2}]
	Consider the following signature $\Sigma = \{\alpha_1 : 0 \to 1, \alpha_2 : 1 \to 0, \alpha_3 : 1 \to 1\}$ and the rewrite rule $id \to \alpha_3$ for the example term $\alpha_1 \fatsemi \alpha_2$.
	Without the boundary complement condition, we get, for example, the following two possible DPOI squares.
	\[
	\adjustbox{scale=0.5}{
	\begin{tikzpicture}
	\begin{pgfonlayer}{nodelayer}
		\node [style=none] (30) at (-4.75, -3.5) {};
		\node [style=none] (31) at (-1, -3.5) {};
		\node [style=none] (32) at (-1, -8.75) {};
		\node [style=none] (33) at (-4.75, -8.75) {};
		\node [style=none] (26) at (6, 5) {};
		\node [style=none] (27) at (9, 5) {};
		\node [style=none] (28) at (9, 0) {};
		\node [style=none] (29) at (6, 0) {};
		\node [style=none] (38) at (6, -2.5) {};
		\node [style=none] (39) at (9, -2.5) {};
		\node [style=none] (40) at (9, -10.5) {};
		\node [style=none] (41) at (6, -10.5) {};
		\node [style=none] (34) at (1, -3) {};
		\node [style=none] (35) at (4, -3) {};
		\node [style=none] (36) at (4, -9.75) {};
		\node [style=none] (37) at (1, -9.75) {};
		\node [style=none] (22) at (-4, 5) {};
		\node [style=none] (23) at (-1, 5) {};
		\node [style=none] (24) at (-1, 0) {};
		\node [style=none] (25) at (-4, 0) {};
		\node [style=node, label={above:$v_1,w_1$}] (0) at (-2.5, 2.5) {};
		\node [style=hedge] (1) at (-3.25, -8) {$\alpha_1$};
		\node [style=hedge] (2) at (-3.25, -4.25) {$\alpha_2$};
		\node [style=node, label={right:$v_1,w_1$}] (3) at (-3.25, -6) {};
		\node [style=node, label={above:$w_1$}] (4) at (2.5, 3.25) {};
		\node [style=node, label={below:$v_1$}] (5) at (2.5, 2) {};
		\node [style=node, label={below:$v_1$}] (6) at (7.5, 0.75) {};
		\node [style=node, label={above:$w_1$}] (7) at (7.5, 4) {};
		\node [style=hedge] (8) at (7.5, 2.25) {$\alpha_3$};
		\node [style=hedge] (9) at (2.5, -3.75) {$\alpha_2$};
		\node [style=hedge] (10) at (2.5, -8.75) {$\alpha_1$};
		\node [style=node, label={below:$w_1$}] (11) at (2.5, -5) {};
		\node [style=node, label={above:$v_1$}] (12) at (2.5, -7.5) {};
		\node [style=hedge] (13) at (7.5, -3.25) {$\alpha_2$};
		\node [style=hedge] (14) at (7.5, -9.75) {$\alpha_1$};
		\node [style=node, label={right:$w_1$}] (15) at (7.5, -4.5) {};
		\node [style=node, label={right:$v_1$}] (16) at (7.5, -8.5) {};
		\node [style=hedge] (17) at (7.5, -6.25) {$\alpha_3$};
		\node [style=none] (18) at (1.75, 1) {};
		\node [style=none] (19) at (3.25, 1) {};
		\node [style=none] (20) at (1.75, 4.25) {};
		\node [style=none] (21) at (3.25, 4.25) {};
		\node [style=none] (42) at (1, 2.5) {};
		\node [style=none] (43) at (0, 2.5) {};
		\node [style=none] (44) at (4, 2.5) {};
		\node [style=none] (45) at (5, 2.5) {};
		\node [style=none] (46) at (-2.5, -1) {};
		\node [style=none] (47) at (-2.5, -2) {};
		\node [style=none] (48) at (2.5, -1) {};
		\node [style=none] (49) at (2.5, -2) {};
		\node [style=none] (50) at (7.5, -1) {};
		\node [style=none] (51) at (7.5, -2) {};
		\node [style=none] (52) at (4.5, -6) {};
		\node [style=none] (53) at (5.5, -6) {};
		\node [style=none] (54) at (0.25, -6) {};
		\node [style=none] (55) at (-0.5, -6) {};
		\node [style=none] (56) at (2, -12) {};
		\node [style=none] (57) at (3, -12) {};
		\node [style=none] (58) at (2, -13) {};
		\node [style=none] (59) at (3, -13) {};
		\node [style=none] (60) at (2.5, -11.5) {};
		\node [style=none] (61) at (2.5, -10.75) {};
		\node [style=none] (62) at (3.75, -12.5) {};
		\node [style=none] (63) at (5, -12) {};
		\node [style=none] (64) at (1, -12.5) {};
		\node [style=none] (65) at (-1, -11) {};
	\end{pgfonlayer}
	\begin{pgfonlayer}{edgelayer}
		\draw [style=graph frame] (33.center)
			 to (32.center)
			 to (31.center)
			 to (30.center)
			 to cycle;
		\draw [style=graph frame] (29.center)
			 to (28.center)
			 to (27.center)
			 to (26.center)
			 to cycle;
		\draw [style=graph frame] (41.center)
			 to (40.center)
			 to (39.center)
			 to (38.center)
			 to cycle;
		\draw [style=graph frame] (37.center)
			 to (36.center)
			 to (35.center)
			 to (34.center)
			 to cycle;
		\draw [style=graph frame] (25.center)
			 to (24.center)
			 to (23.center)
			 to (22.center)
			 to cycle;
		\draw (1) to (3);
		\draw (3) to (2);
		\draw (7) to (8);
		\draw (8) to (6);
		\draw (9) to (11);
		\draw (12) to (10);
		\draw (13) to (15);
		\draw (16) to (14);
		\draw (15) to (17);
		\draw (17) to (16);
		\draw [style=boundary frame] (21.center)
			 to (19.center)
			 to (18.center)
			 to (20.center)
			 to cycle;
		\draw [style=diredge] (42.center) to (43.center);
		\draw [style=diredge] (44.center) to (45.center);
		\draw [style=diredge] (46.center) to (47.center);
		\draw [style=diredge] (48.center) to (49.center);
		\draw [style=diredge] (50.center) to (51.center);
		\draw [style=diredge] (52.center) to (53.center);
		\draw [style=diredge] (54.center) to (55.center);
		\draw [style=boundary frame] (57.center)
			 to (59.center)
			 to (58.center)
			 to (56.center)
			 to cycle;
		\draw [style=diredge] (60.center) to (61.center);
		\draw [style=diredge] (62.center) to (63.center);
		\draw [style=diredge] (64.center) to (65.center);
	\end{pgfonlayer}
\end{tikzpicture}
	}
	\qquad
	\adjustbox{scale=0.5}{
	\begin{tikzpicture}
	\begin{pgfonlayer}{nodelayer}
		\node [style=none] (30) at (-4.75, -3.5) {};
		\node [style=none] (31) at (-1, -3.5) {};
		\node [style=none] (32) at (-1, -8.75) {};
		\node [style=none] (33) at (-4.75, -8.75) {};
		\node [style=none] (26) at (6, 5) {};
		\node [style=none] (27) at (9, 5) {};
		\node [style=none] (28) at (9, 0) {};
		\node [style=none] (29) at (6, 0) {};
		\node [style=none] (38) at (6, -2.5) {};
		\node [style=none] (39) at (9, -2.5) {};
		\node [style=none] (40) at (9, -10.5) {};
		\node [style=none] (41) at (6, -10.5) {};
		\node [style=none] (34) at (1, -3) {};
		\node [style=none] (35) at (4, -3) {};
		\node [style=none] (36) at (4, -9.75) {};
		\node [style=none] (37) at (1, -9.75) {};
		\node [style=none] (22) at (-4, 5) {};
		\node [style=none] (23) at (-1, 5) {};
		\node [style=none] (24) at (-1, 0) {};
		\node [style=none] (25) at (-4, 0) {};
		\node [style=node, label={above:$v_1,w_1$}] (0) at (-2.5, 2.5) {};
		\node [style=hedge] (1) at (-3.25, -8) {$\alpha_1$};
		\node [style=hedge] (2) at (-3.25, -4.25) {$\alpha_2$};
		\node [style=node, label={right:$v_1,w_1$}] (3) at (-3.25, -6) {};
		\node [style=node, label={above:$w_1$}] (4) at (2.5, 3.25) {};
		\node [style=node, label={below:$v_1$}] (5) at (2.5, 2) {};
		\node [style=node, label={below:$v_1$}] (6) at (7.5, 0.75) {};
		\node [style=node, label={above:$w_1$}] (7) at (7.5, 4) {};
		\node [style=hedge] (8) at (7.5, 2.25) {$\alpha_3$};
		\node [style=hedge] (9) at (2.5, -3.75) {$\alpha_2$};
		\node [style=hedge] (10) at (2.5, -8.75) {$\alpha_1$};
		\node [style=node, label={below:$v_1$}] (11) at (2.5, -5) {};
		\node [style=node, label={above:$w_1$}] (12) at (2.5, -7.5) {};
		\node [style=hedge] (13) at (7.5, -3.25) {$\alpha_2$};
		\node [style=hedge] (14) at (7.5, -9.75) {$\alpha_1$};
		\node [style=node, label={right:$v_1$}] (15) at (7.5, -4.5) {};
		\node [style=node, label={right:$w_1$}] (16) at (7.5, -8.5) {};
		\node [style=hedge] (17) at (7.5, -6.25) {$\alpha_3$};
		\node [style=none] (18) at (1.75, 1) {};
		\node [style=none] (19) at (3.25, 1) {};
		\node [style=none] (20) at (1.75, 4.25) {};
		\node [style=none] (21) at (3.25, 4.25) {};
		\node [style=none] (42) at (1, 2.5) {};
		\node [style=none] (43) at (0, 2.5) {};
		\node [style=none] (44) at (4, 2.5) {};
		\node [style=none] (45) at (5, 2.5) {};
		\node [style=none] (46) at (-2.5, -1) {};
		\node [style=none] (47) at (-2.5, -2) {};
		\node [style=none] (48) at (2.5, -1) {};
		\node [style=none] (49) at (2.5, -2) {};
		\node [style=none] (50) at (7.5, -1) {};
		\node [style=none] (51) at (7.5, -2) {};
		\node [style=none] (52) at (4.5, -6) {};
		\node [style=none] (53) at (5.5, -6) {};
		\node [style=none] (54) at (0.25, -6) {};
		\node [style=none] (55) at (-0.5, -6) {};
		\node [style=none] (56) at (2, -12) {};
		\node [style=none] (57) at (3, -12) {};
		\node [style=none] (58) at (2, -13) {};
		\node [style=none] (59) at (3, -13) {};
		\node [style=none] (60) at (2.5, -11.5) {};
		\node [style=none] (61) at (2.5, -10.75) {};
		\node [style=none] (62) at (3.75, -12.5) {};
		\node [style=none] (63) at (5, -12) {};
		\node [style=none] (64) at (1, -12.5) {};
		\node [style=none] (65) at (-1, -11) {};
		\node [style=none] (66) at (7.5, -5.5) {};
		\node [style=none] (67) at (7.5, -7) {};
		\node [style=none] (68) at (7, -5) {};
		\node [style=none] (69) at (6.5, -5.5) {};
		\node [style=none] (70) at (6.5, -7) {};
		\node [style=none] (71) at (8, -7.5) {};
		\node [style=none] (72) at (8.5, -7) {};
		\node [style=none] (73) at (8.5, -5.5) {};
	\end{pgfonlayer}
	\begin{pgfonlayer}{edgelayer}
		\draw [style=graph frame] (33.center)
			 to (32.center)
			 to (31.center)
			 to (30.center)
			 to cycle;
		\draw [style=graph frame] (29.center)
			 to (28.center)
			 to (27.center)
			 to (26.center)
			 to cycle;
		\draw [style=graph frame] (41.center)
			 to (40.center)
			 to (39.center)
			 to (38.center)
			 to cycle;
		\draw [style=graph frame] (37.center)
			 to (36.center)
			 to (35.center)
			 to (34.center)
			 to cycle;
		\draw [style=graph frame] (25.center)
			 to (24.center)
			 to (23.center)
			 to (22.center)
			 to cycle;
		\draw (1) to (3);
		\draw (3) to (2);
		\draw (7) to (8);
		\draw (8) to (6);
		\draw (9) to (11);
		\draw (12) to (10);
		\draw (13) to (15);
		\draw (16) to (14);
		\draw [style=boundary frame] (21.center)
			 to (19.center)
			 to (18.center)
			 to (20.center)
			 to cycle;
		\draw [style=diredge] (42.center) to (43.center);
		\draw [style=diredge] (44.center) to (45.center);
		\draw [style=diredge] (46.center) to (47.center);
		\draw [style=diredge] (48.center) to (49.center);
		\draw [style=diredge] (50.center) to (51.center);
		\draw [style=diredge] (52.center) to (53.center);
		\draw [style=diredge] (54.center) to (55.center);
		\draw [style=boundary frame] (57.center)
			 to (59.center)
			 to (58.center)
			 to (56.center)
			 to cycle;
		\draw [style=diredge] (60.center) to (61.center);
		\draw [style=diredge] (62.center) to (63.center);
		\draw [style=diredge] (64.center) to (65.center);
		\draw (17) to (66.center);
		\draw (17) to (67.center);
		\draw [bend right] (66.center) to (68.center);
		\draw [bend right=45] (68.center) to (69.center);
		\draw (69.center) to (70.center);
		\draw [in=90, out=-90] (70.center) to (16);
		\draw [bend right=45] (67.center) to (71.center);
		\draw [bend right=45] (71.center) to (72.center);
		\draw (73.center) to (72.center);
		\draw [in=90, out=-90] (15) to (73.center);
	\end{pgfonlayer}
\end{tikzpicture}
	}
	\]
	The second DPOI square does not satisfy the boundary complement condition as the input interface is mapped to the input vertex (and similarly for the output interface) and therefore, once pushout square is complete, the resuling cospan does not represent a valid SMC morphism as it requires wires to be able to go backwards, i.e. it requires compact closed structure.
	If we were to consider non-monic maps into $\mathcal{L}^{\bot}$ we would have more DPOI squares non of which would result in a valid morphism (due to in- or out- degrees of vertices being greater than 1).
\end{example}

\begin{definition}[Convex match]
	A subgraph $\mathcal{H}$ of a hypergraph $\mathcal{G}$ is a convex subgraph if for all $v_i, v_j \in V_{\mathcal{H}}$ every path from $v_i$ to $v_j$ is also in $\mathcal{H}$.
	A match $f\colon \mathcal{L} \to \mathcal{G}$ is a convex match if it is a monomorphism and its image is a convex subgraph of $\mathcal{G}$.
\end{definition}

\begin{definition}[Convex DPOI rewriting]%
	\label{def:convex_dpo}
	Let $\mathfrak{R}$ be a set of DPOI rewrite rules.
	Then, given $\mathcal{G} \xleftarrow{} n \sqcup m$ and $\mathcal{H} \xleftarrow{} n \sqcup m$ in $\catname{Hyp(\Sigma)}$, $\mathcal{G}$ rewrites convexly into $\mathcal{H}$ with interface $n \sqcup m$, and we write $(\mathcal G \xleftarrow{} n \sqcup m ) \Rrightarrow_{\mathfrak{R}}  (\mathcal H \xleftarrow{} n \sqcup m )$, if
	\begin{enumerate}
		\item there exists a rule $\mathcal{L} \xleftarrow{} i \sqcup j \xrightarrow{} \mathcal{R}$ in $\mathfrak{R}$ and object $\mathcal{L}^{\bot}$ and cospan arrows $i \sqcup j \xrightarrow{} \mathcal{L}^{\bot} \xleftarrow{} n \sqcup m$ such that the DPOI diagram of~\autoref{def:dpoi} commutes and its marked squares are pushouts;
		\item $f\colon L \to \mathcal G$ is a convex match;
		\item $i \sqcup j \to \mathcal{L}^{\bot} \to \mathcal G$ is a boundary complement in the leftmost pushout.
	\end{enumerate}
\end{definition}

\begin{example}[Example 32~\cite{bonchi_string_2022-2}]
	Consider the following signature $\Sigma \{e_1 : 1 \to 2, e_2 : 2 \to 1, e_3 : 1 \to 1, e_4 : 1 \to 1\}$ and the following rewrite rule
	\[
	id \otimes e_{1} \fatsemi \sym \otimes id \fatsemi id \otimes e_{2} \quad \to \quad \sym \fatsemi e_4 \otimes e_4
	\]
	represented as the following span of hypergraphs
	\[
	\adjustbox{scale=0.65}{
	\begin{tikzpicture}
	\begin{pgfonlayer}{nodelayer}
		\node [style=none] (37) at (3.5, 7.5) {};
		\node [style=none] (38) at (8.5, 7.5) {};
		\node [style=none] (39) at (8.5, -2) {};
		\node [style=none] (40) at (3.5, -2) {};
		\node [style=none] (18) at (-10, 7.5) {};
		\node [style=none] (19) at (-5, 7.5) {};
		\node [style=none] (20) at (-5, -2) {};
		\node [style=none] (21) at (-10, -2) {};
		\node [style=node, label={below:$v_2$}] (0) at (-6.75, -0.5) {};
		\node [style=hedge] (1) at (-6.75, 1.5) {$e_1$};
		\node [style=node] (2) at (-6.5, 3) {};
		\node [style=hedge] (3) at (-6.75, 4.5) {$e_2$};
		\node [style=none] (4) at (-6.5, 4) {};
		\node [style=none] (5) at (-7, 4) {};
		\node [style=node, label={below:$v_1$}] (6) at (-8.25, -0.5) {};
		\node [style=none] (7) at (-6.5, 2) {};
		\node [style=none] (8) at (-7, 2) {};
		\node [style=node, label={above:$w_2$}] (9) at (-6.75, 6.5) {};
		\node [style=node, label={above:$w_1$}] (10) at (-8.25, 6.5) {};
		\node [style=none] (11) at (-7.5, 3) {};
		\node [style=none] (12) at (-8.25, 2) {};
		\node [style=none] (13) at (-8.25, 4) {};
		\node [style=node, label={above:$w_2$}] (14) at (-0.5, 4) {};
		\node [style=node, label={above:$w_1$}] (15) at (-0.5, 2.75) {};
		\node [style=node, label={below:$v_2$}] (16) at (-0.5, 1.5) {};
		\node [style=node, label={below:$v_1$}] (17) at (-0.5, 0.5) {};
		\node [style=none] (22) at (-1.5, -0.75) {};
		\node [style=none] (23) at (0.5, -0.75) {};
		\node [style=none] (24) at (0.5, 5) {};
		\node [style=none] (25) at (-1.5, 5) {};
		\node [style=none] (26) at (-2.5, 2.5) {};
		\node [style=none] (27) at (-3.5, 2.5) {};
		\node [style=node, label={below:$v_1$}] (28) at (5, 0.5) {};
		\node [style=node, label={below:$v_2$}] (29) at (7, 0.5) {};
		\node [style=hedge] (30) at (7, 4.5) {$e_4$};
		\node [style=hedge] (31) at (5, 4.5) {$e_4$};
		\node [style=none] (32) at (6, 2.25) {};
		\node [style=none] (33) at (7, 4) {};
		\node [style=none] (34) at (5, 4) {};
		\node [style=node, label={above:$w_1$}] (35) at (5, 6) {};
		\node [style=node, label={above:$w_2$}] (36) at (7, 6) {};
		\node [style=none] (41) at (1, 2.5) {};
		\node [style=none] (42) at (2, 2.5) {};
	\end{pgfonlayer}
	\begin{pgfonlayer}{edgelayer}
		\draw [style=graph frame] (37.center)
			 to (38.center)
			 to (39.center)
			 to (40.center)
			 to cycle;
		\draw [style=graph frame] (18.center)
			 to (19.center)
			 to (20.center)
			 to (21.center)
			 to cycle;
		\draw (1) to (0);
		\draw (4.center) to (2);
		\draw (2) to (7.center);
		\draw (9) to (3);
		\draw [bend left] (5.center) to (11.center);
		\draw [bend right] (8.center) to (11.center);
		\draw (12.center) to (6);
		\draw [bend left] (12.center) to (11.center);
		\draw (10) to (13.center);
		\draw [bend right] (13.center) to (11.center);
		\draw [style=boundary frame] (22.center)
			 to (23.center)
			 to (24.center)
			 to (25.center)
			 to cycle;
		\draw [style=diredge] (26.center) to (27.center);
		\draw [bend left, looseness=0.75] (28) to (32.center);
		\draw [bend left, looseness=0.75] (32.center) to (29);
		\draw [bend right] (34.center) to (32.center);
		\draw [bend left] (33.center) to (32.center);
		\draw (35) to (31);
		\draw (36) to (30);
		\draw [style=diredge] (41.center) to (42.center);
	\end{pgfonlayer}
\end{tikzpicture}~.
	}
	\]
	Now consider the cospan of hypergraphs corresponding to the following $\Sigma$-term $e_1 \fatsemi e_3\otimes id \fatsemi e_2$.
	\[
	\adjustbox{scale=0.65}{
	\begin{tikzpicture}
	\begin{pgfonlayer}{nodelayer}
		\node [style=none] (28) at (-6.75, -2.5) {};
		\node [style=none] (29) at (-0.75, -2.5) {};
		\node [style=none] (30) at (-6.75, 10) {};
		\node [style=none] (31) at (-0.75, 10) {};
		\node [style=node, label={below:$v_1$}] (0) at (-3.5, -1.5) {};
		\node [style=hedge] (1) at (-3.5, 0) {$e_1$};
		\node [style=none] (2) at (-3.75, 0.5) {};
		\node [style=none] (3) at (-3.25, 0.5) {};
		\node [style=hedge] (4) at (-4.5, 4) {$e_3$};
		\node [style=hedge] (5) at (-3.5, 7.25) {$e_2$};
		\node [style=node] (6) at (-4.5, 5.5) {};
		\node [style=node] (7) at (-4.5, 2.5) {};
		\node [style=none] (8) at (-3.75, 6.75) {};
		\node [style=none] (9) at (-3.25, 6.75) {};
		\node [style=none] (10) at (-3.75, 1.25) {};
		\node [style=none] (11) at (-3.75, 0.5) {};
		\node [style=none] (12) at (-3.25, 1.25) {};
		\node [style=none] (13) at (-3.5, 7.75) {};
		\node [style=node, label={above:$w_1$}] (14) at (-3.5, 8.75) {};
		\node [style=node] (15) at (-3.25, 4) {};
		\node [style=node, label={below:$v_1$}] (16) at (-3.5, -4.5) {};
		\node [style=node, label={above:$w_1$}] (17) at (-3.5, 11.75) {};
		\node [style=none] (18) at (-4.5, 11.5) {};
		\node [style=none] (19) at (-4.5, 12.75) {};
		\node [style=none] (20) at (-2.5, 12.75) {};
		\node [style=none] (21) at (-2.5, 11.5) {};
		\node [style=none] (22) at (-4.5, -5.5) {};
		\node [style=none] (23) at (-4.5, -4.25) {};
		\node [style=none] (24) at (-2.5, -4.25) {};
		\node [style=none] (25) at (-2.5, -5.5) {};
		\node [style=none] (26) at (-3.5, -3.75) {};
		\node [style=none] (27) at (-3.5, -3) {};
		\node [style=none] (32) at (-3.5, 11.25) {};
		\node [style=none] (33) at (-3.5, 10.5) {};
	\end{pgfonlayer}
	\begin{pgfonlayer}{edgelayer}
		\draw [style=graph frame] (30.center)
			 to (28.center)
			 to (29.center)
			 to (31.center)
			 to cycle;
		\draw (10.center) to (11.center);
		\draw (12.center) to (3.center);
		\draw [in=270, out=90, looseness=0.75] (10.center) to (7);
		\draw [in=-90, out=90] (6) to (8.center);
		\draw (14) to (13.center);
		\draw (9.center) to (15);
		\draw (15) to (12.center);
		\draw (4) to (7);
		\draw (6) to (4);
		\draw (1) to (0);
		\draw [style=boundary frame] (19.center)
			 to (20.center)
			 to (21.center)
			 to (18.center)
			 to cycle;
		\draw [style=boundary frame] (23.center)
			 to (24.center)
			 to (25.center)
			 to (22.center)
			 to cycle;
		\draw [style=diredge] (26.center) to (27.center);
		\draw [style=diredge] (32.center) to (33.center);
	\end{pgfonlayer}
\end{tikzpicture}~.
	}
	\]
	Applying the rule to the cospan above results in the DPOI square from~\autoref{fig:non_convex_dpoi_example}.
	While the DPOI square is valid and the pushout complement satisfies the boundary complement condition, the match (highlighted in red) is not convex as there exists a path from the input of $e_3$ to its output that does not belong to the match itself.
	The consequence of this is that this rewrite step does not correspond to a term rewriting modulo SMC --- there is no way to identify such a redex without at least traced monoidal structure.
	That is, there is no decomposition of $e_1 \fatsemi e_3\otimes id \fatsemi e_2$ as $c_1 \fatsemi id_{n}\otimes(id \otimes e_1 \fatsemi \sym\otimes id \fatsemi id \otimes e_2) \fatsemi c_2$.
	There is, however, a match modulo traced symmetric monoidal structure --- $\text{tr}(e_3\otimes id \fatsemi \mathcolorbox{yellow}{e_1 \otimes id \fatsemi \sym \otimes id \fatsemi  id \otimes e_2})$.
	Requiring convex matches prevents such situations from arising.
\end{example}

\begin{figure}
	\[
	\adjustbox{scale=0.65}{
	\begin{tikzpicture}
	\begin{pgfonlayer}{nodelayer}
		\node [style=none] (28) at (-14.75, -10.25) {};
		\node [style=none] (29) at (-8.75, -10.25) {};
		\node [style=none] (30) at (-14.75, 2.25) {};
		\node [style=none] (31) at (-8.75, 2.25) {};
		\node [style=red node, label={below:$u_1$}] (0) at (-11.5, -9.25) {};
		\node [style=hedge] (1) at (-11.5, -7.75) {$e_1$};
		\node [style=none] (2) at (-11.75, -7.25) {};
		\node [style=none] (3) at (-11.25, -7.25) {};
		\node [style=hedge] (4) at (-12.5, -3.75) {$e_3$};
		\node [style=hedge] (5) at (-11.5, -0.5) {$e_2$};
		\node [style=red node] (6) at (-12.5, -2.25) {};
		\node [style=red node] (7) at (-12.5, -5.25) {};
		\node [style=none] (8) at (-11.75, -1) {};
		\node [style=none] (9) at (-11.25, -1) {};
		\node [style=none] (10) at (-11.75, -6.5) {};
		\node [style=none] (11) at (-11.75, -7.25) {};
		\node [style=none] (12) at (-11.25, -6.5) {};
		\node [style=none] (13) at (-11.5, 0) {};
		\node [style=red node, label={above:$q_1$}] (14) at (-11.5, 1) {};
		\node [style=red node] (15) at (-11.25, -3.75) {};
		\node [style=node, label={below:$u_1$}] (16) at (-1, -14.25) {};
		\node [style=node, label={above:$q_1$}] (17) at (-1, -13.5) {};
		\node [style=none] (18) at (-1.75, -15.25) {};
		\node [style=none] (19) at (-1.75, -12.5) {};
		\node [style=none] (20) at (-0.25, -12.5) {};
		\node [style=none] (21) at (-0.25, -15.25) {};
		\node [style=none] (34) at (-14.25, 17.75) {};
		\node [style=none] (35) at (-9.25, 17.75) {};
		\node [style=none] (36) at (-9.25, 8.25) {};
		\node [style=none] (37) at (-14.25, 8.25) {};
		\node [style=node, label={below:$v_2$}] (38) at (-11, 9.75) {};
		\node [style=hedge] (39) at (-11, 11.75) {$e_1$};
		\node [style=node] (40) at (-10.75, 13.25) {};
		\node [style=hedge] (41) at (-11, 14.75) {$e_2$};
		\node [style=none] (42) at (-10.75, 14.25) {};
		\node [style=none] (43) at (-11.25, 14.25) {};
		\node [style=node, label={below:$v_1$}] (44) at (-12.5, 9.75) {};
		\node [style=none] (45) at (-10.75, 12.25) {};
		\node [style=none] (46) at (-11.25, 12.25) {};
		\node [style=node, label={above:$w_2$}] (47) at (-11, 16.75) {};
		\node [style=node, label={above:$w_1$}] (48) at (-12.5, 16.75) {};
		\node [style=none] (49) at (-11.75, 13.25) {};
		\node [style=none] (50) at (-12.5, 12.25) {};
		\node [style=none] (51) at (-12.5, 14.25) {};
		\node [style=none] (52) at (7.5, 18) {};
		\node [style=none] (53) at (12.5, 18) {};
		\node [style=none] (54) at (12.5, 8.5) {};
		\node [style=none] (55) at (7.5, 8.5) {};
		\node [style=node, label={below:$v_1$}] (56) at (9, 11) {};
		\node [style=node, label={below:$v_2$}] (57) at (11, 11) {};
		\node [style=hedge] (58) at (11, 15) {$e_4$};
		\node [style=hedge] (59) at (9, 15) {$e_4$};
		\node [style=none] (60) at (10, 12.75) {};
		\node [style=none] (61) at (11, 14.5) {};
		\node [style=none] (62) at (9, 14.5) {};
		\node [style=node, label={above:$w_1$}] (63) at (9, 16.5) {};
		\node [style=node, label={above:$w_2$}] (64) at (11, 16.5) {};
		\node [style=node, label={above:$w_2$}] (65) at (-1, 14.5) {};
		\node [style=node, label={above:$w_1$}] (66) at (-1, 13.25) {};
		\node [style=node, label={below:$v_2$}] (67) at (-1, 12) {};
		\node [style=node, label={below:$v_1$}] (68) at (-1, 11) {};
		\node [style=none] (69) at (-2, 9.75) {};
		\node [style=none] (70) at (0, 9.75) {};
		\node [style=none] (71) at (0, 15.5) {};
		\node [style=none] (72) at (-2, 15.5) {};
		\node [style=none] (73) at (-4, 13) {};
		\node [style=none] (74) at (-6.75, 13) {};
		\node [style=none] (75) at (2.25, 13) {};
		\node [style=none] (76) at (5, 13) {};
		\node [style=none] (77) at (-12, 6) {};
		\node [style=none] (78) at (-12, 4) {};
		\node [style=none] (79) at (-10.75, -9.5) {$v_2$};
		\node [style=none] (80) at (-13.25, -2.25) {$v_1$};
		\node [style=none] (81) at (-13.25, -5.25) {$w_1$};
		\node [style=none] (82) at (-4, -10.5) {};
		\node [style=none] (83) at (2, -10.5) {};
		\node [style=none] (84) at (-4, 2) {};
		\node [style=none] (85) at (2, 2) {};
		\node [style=node, label={below:$u_1$}] (86) at (-0.75, -9.5) {};
		\node [style=hedge] (90) at (-1.75, -4) {$e_3$};
		\node [style=node] (92) at (-1.75, -2.5) {};
		\node [style=node] (93) at (-1.75, -5.5) {};
		\node [style=node, label={above:$q_1$}] (100) at (-0.75, 0.75) {};
		\node [style=none] (102) at (0, -9.75) {$v_2$};
		\node [style=none] (103) at (-2.5, -2.5) {$v_1$};
		\node [style=none] (104) at (-2.5, -5.5) {$w_1$};
		\node [style=none] (105) at (-5.5, -4) {};
		\node [style=none] (106) at (-7.5, -4) {};
		\node [style=none] (107) at (-1, -12) {};
		\node [style=none] (108) at (-1, -11) {};
		\node [style=none] (109) at (-3, -14) {};
		\node [style=none] (110) at (-10, -11) {};
		\node [style=none] (113) at (7, -10.5) {};
		\node [style=none] (114) at (13, -10.5) {};
		\node [style=none] (115) at (7, 2) {};
		\node [style=none] (116) at (13, 2) {};
		\node [style=node, label={below:$u_1$}] (117) at (10.25, -9.5) {};
		\node [style=hedge] (118) at (10.25, -4) {$e_3$};
		\node [style=node] (120) at (10.25, -5.5) {};
		\node [style=node, label={above:$q_1$}] (121) at (10.25, 0.75) {};
		\node [style=none] (122) at (11, -9.75) {$v_2$};
		\node [style=none] (123) at (9.25, -2.5) {$v_1$};
		\node [style=none] (124) at (9.25, -5.5) {$w_1$};
		\node [style=node] (125) at (10.25, -2.5) {};
		\node [style=hedge] (127) at (10.25, -0.5) {$e_4$};
		\node [style=hedge] (128) at (10.25, -7) {$e_4$};
		\node [style=node] (132) at (10.25, -5.5) {};
		\node [style=none] (134) at (-10.75, 1) {$w_2$};
		\node [style=none] (135) at (11, 0.75) {$w_2$};
		\node [style=none] (136) at (10, 6) {};
		\node [style=none] (137) at (10, 4) {};
		\node [style=none] (138) at (-1.25, 6) {};
		\node [style=none] (139) at (-1.25, 4) {};
		\node [style=none] (140) at (1, -14) {};
		\node [style=none] (141) at (8, -11) {};
		\node [style=none] (142) at (3.5, -4) {};
		\node [style=none] (143) at (5.5, -4) {};
		\node [style=none] (144) at (0, 0.75) {$w_2$};
	\end{pgfonlayer}
	\begin{pgfonlayer}{edgelayer}
		\draw [style=graph frame] (30.center)
			 to (28.center)
			 to (29.center)
			 to (31.center)
			 to cycle;
		\draw [style=component] (10.center) to (11.center);
		\draw (12.center) to (3.center);
		\draw [style=component, in=270, out=90, looseness=0.75] (10.center) to (7);
		\draw [style=component, in=-90, out=90] (6) to (8.center);
		\draw [style=component] (14) to (13.center);
		\draw [style=component] (9.center) to (15);
		\draw [style=component] (15) to (12.center);
		\draw (4) to (7);
		\draw (6) to (4);
		\draw [style=component] (1) to (0);
		\draw [style=boundary frame] (19.center)
			 to (20.center)
			 to (21.center)
			 to (18.center)
			 to cycle;
		\draw [style=graph frame] (34.center)
			 to (35.center)
			 to (36.center)
			 to (37.center)
			 to cycle;
		\draw (39) to (38);
		\draw (42.center) to (40);
		\draw (40) to (45.center);
		\draw (47) to (41);
		\draw [bend left] (43.center) to (49.center);
		\draw [bend right] (46.center) to (49.center);
		\draw (50.center) to (44);
		\draw [bend left] (50.center) to (49.center);
		\draw (48) to (51.center);
		\draw [bend right] (51.center) to (49.center);
		\draw [style=graph frame] (52.center)
			 to (53.center)
			 to (54.center)
			 to (55.center)
			 to cycle;
		\draw [bend left, looseness=0.75] (56) to (60.center);
		\draw [bend left, looseness=0.75] (60.center) to (57);
		\draw [bend right] (62.center) to (60.center);
		\draw [bend left] (61.center) to (60.center);
		\draw (63) to (59);
		\draw (64) to (58);
		\draw [style=boundary frame] (69.center)
			 to (70.center)
			 to (71.center)
			 to (72.center)
			 to cycle;
		\draw [style=diredge] (73.center) to (74.center);
		\draw [style=diredge] (75.center) to (76.center);
		\draw [style=diredge] (77.center) to (78.center);
		\draw [style=graph frame] (84.center)
			 to (82.center)
			 to (83.center)
			 to (85.center)
			 to cycle;
		\draw (90) to (93);
		\draw (92) to (90);
		\draw [style=diredge] (105.center) to (106.center);
		\draw [style=diredge] (107.center) to (108.center);
		\draw [style=diredge] (109.center) to (110.center);
		\draw [style=graph frame] (115.center)
			 to (113.center)
			 to (114.center)
			 to (116.center)
			 to cycle;
		\draw (118) to (120);
		\draw (132) to (128);
		\draw (117) to (128);
		\draw (121) to (127);
		\draw (127) to (125);
		\draw (125) to (118);
		\draw [style=diredge] (136.center) to (137.center);
		\draw [style=diredge] (138.center) to (139.center);
		\draw [style=diredge] (140.center) to (141.center);
		\draw [style=diredge] (142.center) to (143.center);
	\end{pgfonlayer}
\end{tikzpicture}
	}
	\]
	\caption{DPOI diagram for a non-convex match}%
	\label{fig:non_convex_dpoi_example}
\end{figure}

\section{Pushout in $\catname{EHyp}(\Sigma)$}
\label{sec:appendix:pushout}

In constructing the pushout the functional versions of e-hypergraph relations will be useful.
\begin{remark}
	$\textcolor{e-color}{<}$, $\textcolor{closed-color}{<}$ and $\consistency$ can be considered as (partial) functions defined on $V_{\mathcal{F}} \sqcup  E_{\mathcal{F}}$, \textit{i.e.}, on the coproduct of vertices and edges.
	To make things well-typed, we will use corresponding coproduct injections $\iota_{V} : {V_{\mathcal{F}}} \to V_{\mathcal{F}} \sqcup  E_{\mathcal{F}}$ and $\iota_{E} : {E_{\mathcal{F}}} \to V_{\mathcal{F}} \sqcup  E_{\mathcal{F}}$ when passing either a vertex or an edge into these functions.
	For example, the immediate predecessor of a vertex $x$ can be written functionally as $\textcolor{e-color}{<_{\mathcal{F}}^{\mu}}(\iota_{V_{\mathcal{F}}}(x))$.
\end{remark}

\begin{remark}
	Using functional notation we can reformulate the notion of a homomorphism between two e-hypergraphs in the following way.
\end{remark}
\begin{definition}
	\label{def:e-homo-2}
	A \emph{homomorphism} $\phi: \mathcal{F} \to \mathcal{G}$ of e-hypergraphs $\mathcal{F},\mathcal{G}$ is a pair of functions $\phi_V : V_{\mathcal{F}} \to V_{\mathcal{G}}, \phi_E : E_{\mathcal{F}} \to E_{\mathcal{G}}$ such that

	\begin{enumerate}
		\item $\phi$ is hypergraph homomorphism.

		\item When $x$ is not a top-level vertex such that $\textcolor{e-color}{<}(\iota_{V}(x))$ is defined,
		      \[
			      \phi_{E}(\textcolor{e-color}{<_{\mathcal{F}}^{\mu}}(\iota_{V_{\mathcal{F}}}(x))) = \textcolor{e-color}{<_{\mathcal{G}}^{\mu}}(\phi_{V} \fatsemi \iota_{V_{\mathcal{G}}}(x))
		      \]
		      and
		      \[
			      \phi_{E}(\textcolor{e-color}{<_{\mathcal{F}}^{\mu}}(\iota_{E_{\mathcal{F}}}(x))) = \textcolor{e-color}{<_{\mathcal{G}}^{\mu}}(\phi_{E} \fatsemi \iota_{E_{\mathcal{G}}}(x))
		      \] when $x$ is a not top-level edge.
		\item When $x$ is not a top-level vertex such that $\textcolor{closed-color}{<}(\iota_{V}(x))$ is defined,
		      \[
			      \phi_{E}(\textcolor{closed-color}{<_{\mathcal{F}}^{\mu}}(\iota_{V_{\mathcal{F}}}(x))) = \textcolor{closed-color}{<_{\mathcal{G}}^{\mu}}(\phi_{V} \fatsemi \iota_{V_{\mathcal{G}}}(x))
		      \]
		      and
		      \[
			      \phi_{E}(\textcolor{closed-color}{<_{\mathcal{F}}^{\mu}}(\iota_{E_{\mathcal{F}}}(x))) = \textcolor{closed-color}{<_{\mathcal{G}}^{\mu}}(\phi_{E} \fatsemi \iota_{E_{\mathcal{G}}}(x))
		      \] when $x$ is a not top-level edge.
		\item
		      When $x \in E_{\mathcal{F}}$
		      \[
			      [\phi_{V} \fatsemi \iota_{V_{\mathcal{G}}}, \phi_{E} \fatsemi \iota_{E_{\mathcal{G}}} ]^{*}(\consistency_{\mathcal{F}}(\iota_{E_{\mathcal{F}}}(x)))
			      \subseteq
			      \consistency_{\mathcal{G}}(\phi_{E} \fatsemi \iota_{E_{\mathcal{G}}}(x))
		      \]
		      where $\phi_{V} \fatsemi \iota_{V_{\mathcal{G}}} : V_{\mathcal{F}} \to V_{\mathcal{G}} \sqcup  E_{\mathcal{G}}$, and similarly for $\phi_{E} \fatsemi \iota_{E_\mathcal{G}}$ so that $[\phi_{V} \fatsemi \iota_{V_{\mathcal{G}}}, \phi_{E} \fatsemi \iota_{E_{\mathcal{G}}}] : V_{\mathcal{F}} \sqcup  E_{\mathcal{F}} \to  V_{\mathcal{G}} \sqcup  E_{\mathcal{G}}$.
		\item When $x \in V_{\mathcal{F}}$
		      \[
			      [\phi_{V} \fatsemi \iota_{V_{\mathcal{G}}}, \phi_{E} \fatsemi \iota_{E_{\mathcal{G}}}]^{*}(\consistency_{\mathcal{F}}(\iota_{V_{\mathcal{F}}}(x)))
			      \subseteq
			      \consistency_{\mathcal{G}}(\phi_{V} \fatsemi \iota_{V_{\mathcal{G}}}(x)).
		      \]
	\end{enumerate}
\end{definition}

\begin{theorem}[Existence of pushouts in $\catname{EHyp}(\Sigma)$]
	\label{th:existence_of_pushouts}
	Consider the following span in $\catname{EHyp}(\Sigma)$
	\[\begin{tikzcd}
			Z && X \\
			\\
			Y
			\arrow["f", from=1-1, to=1-3]
			\arrow["g"', from=1-1, to=3-1]
		\end{tikzcd}\]
	such that
	\begin{enumerate}
		\label{pushout:assumptions}
		\item $Z$ is a \textit{discrete} e-hypergraph.
		\item \label{assumption:equal_predecessors} $[f_{V}(v_i)) = [f_{V}(v_j))$ and $[g_{V}(v_i)) = [g_{V}(v_j))$ for all $v_{i},v_{j}$ in $V_{Z}$.
		\item \label{assumption:non_ambiguous_predecessors} If $[f_{V}(v)) \not = \varnothing$ then $[g_{V}(v)) = \varnothing$ and if $[g_{V}(v)) \not = \varnothing$ then $[f_{V}(v)) = \varnothing$.
		\item $\consistency(f_{V}(v_i)) = \consistency(f_{V}(v_j))$ and $\consistency(g_{V}(v_i)) = \consistency(g_{V}(v_j))$ for all $v_i,v_j$ in $V_{Z}$.
	\end{enumerate}
	then the pushout $X \sqcup _{f,g} Y$ exists.
\end{theorem}
\begin{proof}

	We next explicitly construct a pushout.
	Consider the diagram below.
	\[
		\adjustbox{scale=1.25}{
			\begin{tikzcd}
				Z \arrow[r, "f"] \arrow[d, "g"']                                   & X \arrow[d, "\sfrac{\iota_1}{\sim R}"] \arrow[rdd, "j_1", bend left] &   \\
				Y \arrow[r, "\sfrac{\iota_2}{\sim R}"] \arrow[rrd, "j_2"', bend right] & \sfrac{X\sqcup Y}{\sim R} \arrow[rd, "u"]                              &   \\
				&                                                                  & Q
			\end{tikzcd}}
	\]
	Then, the pushout of e-hypergraphs $X$ and $Y$ is computed in two steps.
	First, a coproduct of $X\sqcup Y$ is computed which is
	\[
		X \sqcup  Y = \{V_{X} \sqcup  V_{Y}, E_{X} \sqcup  E_{Y}, s_{X\sqcup Y}, t_{X\sqcup Y}, \consistency_{X\sqcup Y}, \textcolor{e-color}{<_{X\sqcup Y}}, \textcolor{closed-color}{<_{X\sqcup Y}}\}
	\]
	where $s_{X\sqcup Y} : E_{X} \sqcup  E_{Y} \to (V_{X} \sqcup  V_{Y})^{*}$ which can be defined as a copairing $[s'_{X}, s'_{Y}] : E_{X} \sqcup  E_{Y} \to (V_{X} \sqcup  V_{Y})^{*}$ and $s'_{X} : E_{X} \to (V_{X} \sqcup  V_{Y})^{*}$, $s'_{Y} : E_{Y} \to (V_{X} \sqcup  V_{Y})^{*}$ defined as $s'_{X} = s_{X} \fatsemi \iota_{1,V}^{*}$ and $s'_{Y} = s_{Y} \fatsemi \iota_{2,V}^{*}$ where $\iota_1,\iota_2$ are corresponding coproduct injections, similarly for $t_{X\sqcup Y}, \textcolor{e-color}{<_{X\sqcup Y}}, \textcolor{closed-color}{<_{X\sqcup Y}}, \consistency_{X\sqcup Y}$.
	We will omit labels as they are irrelevant to pushout construction.
	Consider relations
	\begin{align*}
		S_{V} & = \{
		\;(x_i,y_j) \in (V_{X} \sqcup  V_{Y}) \times (V_{X} \sqcup  V_{Y})\; |                                               \\
		      & \;\exists z \in V_{Z} \; . \; x_i = f_{V} \fatsemi \iota_{1,V}(z) \text{ and } y_j = g_{V} \fatsemi \iota_{2,V}(z) \\
		      & \;\text{ where $x_i \in V_{X}$ and $y_j \in V_{Y}$ }
		\}                                                                                                       \\
		      & \cup                                                                                             \\
		      & \{
		\;(y_j,x_j) \in (V_{X} \sqcup  V_{Y}) \times (V_{X} \sqcup  V_{Y})\; |                                               \\
		      & \;\exists z \in V_{Z} \; . \; x_i = f_{V} \fatsemi \iota_{1,V}(z) \text{ and } y_j = g_{V} \fatsemi \iota_{2,V}(z) \\
		      & \;\text{ where $x_i \in V_{X}$ and $y_j \in V_{Y}$ }
		\}                                                                                                       \\
		      & \cup                                                                                             \\
		      & \{
		\;(x,x)\;\text{where}\; x \in V_{X} \sqcup  V_{Y}
		\}                                                                                                       \\
		\\
		S_{E} & = \{(x,x) \text{ where } x \in E_{X} \sqcup  E_{Y}\}
	\end{align*}
	and let relations $R_{V},R_{E}$ be their transitive closures respectively.

	We then quotient the set vertices and edges in $X \sqcup  Y$ by $R_{V}$ and $R_{E}$ respectively.

	\begin{align*}
		\sfrac{X \sqcup  Y}{\sim (R_{V},R_{E})} = & (
		\sfrac{V_{X} \sqcup  V_{Y}}{\sim (R_{V},R_{E})},                                                             \\
		                                    & \sfrac{E_{X} \sqcup  E_{Y}}{\sim (R_{V},R_{E})},                       \\
		                                    & \sfrac{s_{X\sqcup Y}}{\sim (R_{V},R_{E})},                             \\
		                                    & \sfrac{t_{X,Y}}{\sim (R_{V},R_{E})},                             \\
		                                    & \consistency_{\sfrac{X\sqcup Y}{\sim (R_{V},R_{E})}},                  \\
		                                    & \textcolor{e-color}{<_{\sfrac{X \sqcup  Y}{\sim (R_{V},R_{E})}}}       \\
		                                    & \textcolor{closed-color}{<_{\sfrac{X \sqcup  Y}{\sim (R_{V},R_{E})}}})
	\end{align*}

	We will refer to $\sim (R_{V},R_{E})$ just as $\sim$, \textit{e.g.}, by writing $\sfrac{V_{X} \sqcup  V_{Y}}{\sim}$ and concrete relation will be clear from the context.
	Where necessary we will refer to $S_{V}$ and $S_{R}$ as $\sim_{S}$.
	We have
	\[
		\sfrac{s_{X\sqcup Y}}{\sim} : \sfrac{E_{X} \sqcup  E_{Y}}{\sim} \to (\sfrac{V_{X} \sqcup  V_{Y}}{\sim})^{*}
	\]
	and there is an obvious surjective function $[]_{V} : (V_{X} \sqcup  V_{Y}) \to (\sfrac{V_{X} \sqcup  V_{Y}}{\sim})$ that maps elements to their equivalence classes and $[]_{V}^{*}$ is its extension to sequences.
	There is also $[] : E_{X} \sqcup  E_{Y} \to \sfrac{E_{X} \sqcup  E_{Y}}{\sim}$ and we will omit subscripts as the correct type will be clear from the argument.
	We then define
	\[
		\sfrac{s_{X\sqcup Y}}{\sim}([e]) = s_{X\sqcup Y} \fatsemi []^{*}(e) = [s_{X\sqcup Y}(e)]^{*}
	\]
	We will also use subscripts when it is important to tell if an element of $E_{X} \sqcup  E_{Y}$ has a pre-image in either $E_{X}$ or $E_{Y}$ by writing $e_{x}$ or $e_{y}$.
	Similarly, we will write $v_{x}$ to refer to a vertex with a pre-image in $V_{X}$.
	Likewise,
	\[
		t_{\sfrac{X\sqcup Y}{\sim}}([e]) = [t_{X\sqcup Y}(e)]^{*}
	\]

	These definitions make $([]_{V},[]_{E})$ automatically a homomorphism with respect to source and target maps.
	These maps are also automatically well-defined functions since $[e_1] = [e_2]$ implies $e_1 = e_2$ since the mappings for edges are mono.

	Recall that $\textcolor{e-color}{<}$ is essentially a transitive closure of $\textcolor{e-color}{<^{\mu}}$ and hence we can first define
	\[\textcolor{e-color}{<_{\sfrac{X\sqcup Y}{\sim}}^{\mu}}(\iota_{\sfrac{E_{X} \sqcup  E_{Y}}{\sim}}[e])
	\] and
	\[
		\textcolor{e-color}{<_{\sfrac{X\sqcup Y}{\sim}}^{\mu}}(\iota_{\sfrac{V_{X} \sqcup  V_{Y}}{\sim}}[u])
	\] where $\iota_*$ are injections into $\sfrac{V_{X} \sqcup  V_{Y}}{\sim} \sqcup  \sfrac{E_{X} \sqcup  E_{Y}}{\sim}$.
	We will consider several cases.
	First, assume that $u$ is an element of $V_{X\sqcup Y}$ and
	\begin{enumerate}
		\item If there exists $v$ such that $\textcolor{e-color}{<_{X\sqcup Y}^{\mu}}(\iota_{V_{X} \sqcup  V_{Y}}(v))$ is defined and $u \sim v$
		      we let
		      \[
			      \textcolor{e-color}{<_{\sfrac{X\sqcup Y}{\sim}}^{\mu}}(\iota_{\sfrac{V_{X} \sqcup  V_{Y}}{\sim}}([u])) = [\textcolor{e-color}{<_{X\sqcup Y}^{\mu}}(\iota_{V_{X} \sqcup  V_{Y}}(v))]
		      \]
		\item \label{def:child_respects_connectivity} $u$ has no pre-image in $V_{Z}$ and $<_{X\sqcup Y}^{\mu}(\iota_{V_{X} \sqcup  V_{Y}}(u))$ is undefined (\textit{i.e.}, $[u) = \varnothing$).
		      If there exists $v$ such that $\textcolor{e-color}{<_{X\sqcup Y}^{\mu}}(\iota_{V_{X} \sqcup  V_{Y}}(v)) = e'$ and such that there is an \textit{undirected} path from $[u]$ to $[v]$, then we define
		      \[
			      \textcolor{e-color}{<_{\sfrac{X\sqcup Y}{\sim}}^{\mu}}(\iota_{\sfrac{V_{X} \sqcup  V_{Y}}{\sim}}[u]) = <_{\sfrac{X\sqcup Y}{\sim}}^{\mu}(\iota_{\sfrac{V_{X} \sqcup  V_{Y}}{\sim}}[v])
		      \]
	\end{enumerate}

	Otherwise, we leave $\textcolor{e-color}{<_{\sfrac{X\sqcup Y}{\sim}}^{\mu}}(\iota_{\sfrac{V_{X} \sqcup  V_{Y}}{\sim}}([u]))$ undefined.
	Now, assume that $e$ is an element of $E_{X\sqcup Y}$.
	Consider two cases.
	\begin{enumerate}
		\item  $\textcolor{e-color}{<_{X\sqcup Y}^{\mu}}(\iota_{E_{X} \sqcup  E_{Y}}(e)) = e'$.
		      Then we define
		      \[
			      \textcolor{e-color}{<_{\sfrac{X\sqcup Y}{\sim}}^{\mu}}(\iota_{\sfrac{E_{X} \sqcup  E_{Y}}{\sim}}[e]) = [\textcolor{e-color}{<_{X\sqcup Y}^{\mu}}(\iota_{E_{X} \sqcup  E_{Y}}(e))]
		      \]
		\item $[e) = \varnothing$ and there exists $v$ such that $<_{X\sqcup Y}^{\mu}(\iota_{V_{X} \sqcup  V_{Y}}(v)) = e'$ and such that there is an undirected path from $[e]$ to $[v]$, we define
		      \[
			      \textcolor{e-color}{<_{\sfrac{X\sqcup Y}{\sim}}^{\mu}}(\iota_{\sfrac{E_{X} \sqcup  E_{Y}}{\sim}}[e]) = \textcolor{e-color}{<_{\sfrac{X\sqcup Y}{\sim}}^{\mu}}(\iota_{\sfrac{V_{X} \sqcup  V_{Y}}{\sim}}[v])
		      \]
	\end{enumerate}
	Otherwise we leave $\textcolor{e-color}{<_{\sfrac{X\sqcup Y}{\sim}}^{\mu}}(\iota_{E_{X} \sqcup  E_{Y}}{\sim}([e]))$ undefined.
	Clearly, all the cases above are disjoint.
	We similarly define $\textcolor{closed-color}{<_{\sfrac{X\sqcup Y}{\sim}}^{\mu}}$.

	Let's now define $\consistency_{\sfrac{X\sqcup Y}{\sim}}$.
	We can consider the consistency relation from the coproduct as a function
	\[
		\consistency_{X\sqcup Y} : (V_{X} \sqcup  V_{Y}) \sqcup  (E_{X} \sqcup  E_{Y}) \to 2^{(V_{X} \sqcup  V_{Y}) \sqcup  (E_{X} \sqcup  E_{Y})}
	\]
	Quotienting the values of the function gives us
	\[
		\consistency_{X\sqcup Y}' : V_{X} \sqcup  V_{Y} \sqcup  E_{X} \sqcup  E_{Y} \to 2^{(\sfrac{V_{X} \sqcup  V_{Y}}{\sim} \sqcup  \sfrac{(E_{X} \sqcup  E_{Y})}{\sim})}
	\]
	which is essentially $[ []_{V} \fatsemi \iota_{\sfrac{V_{X} \sqcup  V_{Y}}{\sim}}, []_{E} \fatsemi \iota_{\sfrac{E_{X} \sqcup  E_{Y}}{\sim}}]^{*}$ (a copairing extended to sequences that we will further denote as $[ []_{V}^{\consistency} []_{E}^{\consistency}]$) applied to the return value of $\consistency_{X\sqcup Y}$.
	We then first define an auxiliary relation $\consistency^{\hashtag}$ similarly to how $<^{\mu}_{\sfrac{X\sqcup Y}{\sim}}$ was defined. We begin with defining $\consistency^{\hashtag}_{\sfrac{X\sqcup Y}{\sim}}$ for vertices.

	\begin{enumerate}
		\item If there exists $v$ such that $\consistency_{X\sqcup Y}(\iota_{V_{X} \sqcup  V_{Y}}(v)) \not = \varnothing$ and $u \sim v$, we let
		      \ifdefined \ONECOLUMN
			      \[
				      \consistency_{\sfrac{X\sqcup Y}{\sim}}^{\hashtag}(\iota_{\sfrac{V_{X} \sqcup  V_{Y}}{\sim}}([u]))
				      =
				      [[]_{V}^{\consistency},[]_{E}^{\consistency}]^{*}(\consistency_{X\sqcup Y}(\iota_{V_{X} \sqcup  V_{Y}}(v)))
			      \]
		      \else
			      \begin{align*}
				       & \consistency_{\sfrac{X\sqcup Y}{\sim}}^{\hashtag}(\iota_{\sfrac{V_{X} \sqcup  V_{Y}}{\sim}}([u]))           \\
				       & =                                                                                               \\
				       & [[]_{V}^{\consistency},[]_{E}^{\consistency}]^{*}(\consistency_{X\sqcup Y}(\iota_{V_{X} \sqcup  V_{Y}}(v)))
			      \end{align*}
		      \fi
		\item $u$ has no pre-image in $V_{Z}$ and $\consistency_{X\sqcup Y}(\iota_{V_{X} \sqcup  V_{Y}}(u)) = \varnothing$ and there exists $v$ such that $\consistency_{X\sqcup Y}(\iota_{V_{X} \sqcup  V_{Y}}(v)) \not = \varnothing$ and such that there is an undirected path from $[u]$ to $[v]$.
		      Then we let
		      \[
			      \consistency_{\sfrac{X\sqcup Y}{\sim}}^{\hashtag}(\iota_{\sfrac{V_{X} \sqcup  V_{Y}}{\sim}}([u])) = \consistency_{\sfrac{X\sqcup Y}{\sim}}^{\hashtag}(\iota_{\sfrac{V_{X} \sqcup  V_{Y}}{\sim}}([v]))
		      \]
	\end{enumerate}

	Next we define $\consistency_{\sfrac{X\sqcup Y}{\sim}}^{\hashtag}$ for edges.

	\begin{enumerate}
		\item If $\consistency_{X\sqcup Y}(\iota_{E_{X} \sqcup  E_{Y}}(e)) \not = \varnothing$, then
		      \begin{align*}
			      \consistency_{\sfrac{X\sqcup Y}{\sim}}^{\hashtag}(\iota_{\sfrac{E_{X} \sqcup  E_{Y}}{\sim}}([e]_{E})) =
			      [[]_{V}^{\consistency}, []_{E}^{\consistency}]^{*}(\consistency_{X\sqcup Y}(\iota_{E_{X} \sqcup  E_{Y}}(e)))
		      \end{align*}
		\item $\consistency_{X\sqcup Y}(\iota_{E_{X} \sqcup  E_{Y}}(e)) = \varnothing$ and there exists $v$ such that $\consistency_{X\sqcup Y}(\iota_{V_{X} \sqcup  V_{Y}}(v)) \not = \varnothing$ and such that there is an undirected path from $[e]$ to $[v]$
		      then
		      \[
			      \consistency_{\sfrac{X\sqcup Y}{\sim}}^{\hashtag}(\iota_{\sfrac{E_{X} \sqcup  E_{Y}}{\sim}}([e]_{E})) = \consistency_{\sfrac{X\sqcup Y}{\sim}}^{\hashtag}(\iota_{V_{X} \sqcup  V_{Y}}([v]))
		      \]
	\end{enumerate}

	The well-definedness of this construction follows by the same argument as the well-definedness of $<_{\sfrac{X\sqcup Y}{\sim}}^{\mu}$.
	Then we define
	\ifdefined \ONECOLUMN
		\[
			\consistency_{\sfrac{X\sqcup Y}{\sim}}(\iota_{\sfrac{V_{X} \sqcup  V_{Y}}{\sim}}([v])) \qquad \text{and} \qquad \consistency_{\sfrac{X\sqcup Y}{\sim}}(\iota_{\sfrac{E_{X} \sqcup  E_{Y}}{\sim}}([e]))
		\]
	\else
		\[
			\consistency_{\sfrac{X\sqcup Y}{\sim}}(\iota_{\sfrac{V_{X} \sqcup  V_{Y}}{\sim}}([v]))
		\]
		and
		\[
			\consistency_{\sfrac{X\sqcup Y}{\sim}}(\iota_{\sfrac{E_{X} \sqcup  E_{Y}}{\sim}}([e]))
		\]
	\fi

	as closures of $\consistency^{\hashtag}_{\sfrac{X\sqcup Y}{\sim}}$ as below.
	\[
		(\consistency^{\hashtag}_{\sfrac{X\sqcup Y}{\sim}}(\iota_{\sfrac{*}{\sim}}([x])))^{c}
	\]
	where $c$ denotes a closure and $\iota_{\sfrac{*}{\sim}}$ is $\iota_{\sfrac{V_{X} \sqcup  V_{Y}}{\sim}}$ or $\iota_{\sfrac{E_{X} \sqcup  E_{Y}}{\sim}}$ depending on whether $[x_i]$ comes from $V_{X} \sqcup  V_{Y}$ or $E_{X} \sqcup  E_{Y}$, is the smallest set such that
	\begin{itemize}
		\item $[x] \in (\consistency^{\hashtag}_{\sfrac{X\sqcup Y}{\sim}}(\iota_{\sfrac{*}{\sim}}([x])))^{c}$ if $\consistency^{\hashtag}_{\sfrac{X\sqcup Y}{\sim}}(\iota_{\sfrac{*}{\sim}}([x])) \not = \varnothing$
		\item if $[y] \in \consistency^{\hashtag}_{\sfrac{X\sqcup Y}{\sim}}(\iota_{\sfrac{*}{\sim}}([x]))$ then
		      \[
			      [x] \in (\consistency^{\hashtag}_{\sfrac{X\sqcup Y}{\sim}}(\iota_{\sfrac{*}{\sim}}([y])))^{c}
		      \]
		\item for any sequence $([x_1], \ldots, [x_n])$ such that
		      \ifdefined \ONECOLUMN
			      \[[x_i] \in \consistency^{\hashtag}_{\sfrac{X\sqcup Y}{\sim}}(\iota_{\sfrac{*}{\sim}}([x_{i\sqcup 1}])) \qquad \text{or} \qquad [x_{i\sqcup 1}] \in \consistency^{\hashtag}_{\sfrac{X\sqcup Y}{\sim}}(\iota_{\sfrac{*}{\sim}}([x_{i}]))\]
		      \else
			      \[
				      [x_i] \in \consistency^{\hashtag}_{\sfrac{X\sqcup Y}{\sim}}(\iota_{\sfrac{*}{\sim}}([x_{i\sqcup 1}]))
			      \] or
			      \[
				      [x_{i\sqcup 1}] \in \consistency^{\hashtag}_{\sfrac{X\sqcup Y}{\sim}}(\iota_{\sfrac{*}{\sim}}([x_{i}]))
			      \]
		      \fi
		      for $i < n$ both
		      \ifdefined \ONECOLUMN
			      \[
				      [x_1] \in (\consistency^{\hashtag}_{\sfrac{X\sqcup Y}{\sim}}(\iota_{\sfrac{*}{\sim}}([x_{n}])))^{c} \qquad \text{and} \qquad [x_n] \in (\consistency^{\hashtag}_{\sfrac{X\sqcup Y}{\sim}}(\iota_{\sfrac{*}{\sim}}([x_{1}])))^{c}
			      \]
		      \else
			      \[
				      [x_1] \in (\consistency^{\hashtag}_{\sfrac{X\sqcup Y}{\sim}}(\iota_{\sfrac{*}{\sim}}([x_{n}])))^{c}
			      \]
			      and
			      \[
				      [x_n] \in (\consistency^{\hashtag}_{\sfrac{X\sqcup Y}{\sim}}(\iota_{\sfrac{*}{\sim}}([x_{1}])))^{c}
			      \]
		      \fi
	\end{itemize}

	This construction is analogous to how a pushout was constructed by Tiurin et al.~\cite{tiurin2025equivalencehypergraphsdporewriting} and we refer the reader for more detailed proofs there.
	The only addition is $\textcolor{closed-color}{<_{\sfrac{X\sqcup Y}{\sim}}^{\mu}}$ that is constructed analogously to $\textcolor{e-color}{<_{\sfrac{X\sqcup Y}{\sim}}^{\mu}}$.
\end{proof}
The construction can be illustrated with the example in~\autoref{fig:pushout_example}.
Vertices and edges that are being identified as they have a common preimage inherit the relations from the vertex or edge that has it defined either directly or via a path.
In particular, a coproduct of two feet contains\begin{tikzpicture}
	\adjustbox{scale=0.8}{
		\begin{pgfonlayer}{nodelayer}
			\node [style=node, label={above:$v_1$}] (0) at (-5, 3.5) {};
			\node [style=node, label={above:$v_2$}] (1) at (-4.5, 3.5) {};
			\node [style=node, label={above:$v_3$}] (2) at (-4, 3.5) {};
			\node [style=node, label={above:$v'_1$}] (3) at (-3.5, 3.5) {};
			\node [style=node, label={above:$v'_2$}] (4) at (-3, 3.5) {};
			\node [style=node, label={above:$v'_3$}] (5) at (-2.5, 3.5) {};
		\end{pgfonlayer}
	}
\end{tikzpicture}
such that $\textcolor{closed-color}{<_{X\sqcup Y}}$ is only defined for $v_1,v_2,v_3$.
Therefore, we assign $\textcolor{closed-color}{<_{\sfrac{X\sqcup Y}{\sim}}}(\iota_{\sfrac{V_{X} \sqcup  V_{Y}}{\sim}}[v_1']) = [\textcolor{closed-color}{<_{X\sqcup Y}}(\iota_{V_{X} \sqcup  V_{Y}}(v_1))]$ and so on.
This motivates the requirement for conditions (2), (3), and (4) of~\autoref{th:existence_of_pushouts}, otherwise such inheritance would be ambiguous.

\begin{figure}[t]
	\[
		\adjustbox{scale=0.6}{
			\begin{tikzpicture}
	\begin{pgfonlayer}{nodelayer}
		\node [style=none] (0) at (1.5, 4) {};
		\node [style=none] (1) at (2, 4.5) {};
		\node [style=none] (2) at (5.75, 4.5) {};
		\node [style=none] (3) at (6.25, 4) {};
		\node [style=none] (4) at (6.25, -1.5) {};
		\node [style=none] (5) at (5.75, -2) {};
		\node [style=none] (6) at (2, -2) {};
		\node [style=none] (7) at (1.5, -1.5) {};
		\node [style=node] (15) at (4, -3.25) {};
		\node [style=none] (16) at (4, -2) {};
		\node [style=none] (17) at (3.5, 4.5) {};
		\node [style=node] (18) at (3.5, 6.25) {};
		\node [style=node, label={above:$v_1$}] (19) at (-6.25, 3) {};
		\node [style=node, label={above:$v_2$}] (20) at (-4.75, 3) {};
		\node [style=node, label={above:$v_3$}] (21) at (-5.5, 1.25) {};
		\node [style=hedge] (25) at (-5.5, -14.75) {$f$};
		\node [style=none] (26) at (-2.25, 1.5) {};
		\node [style=none] (27) at (-0.5, 1.5) {};
		\node [style=none] (28) at (-5.5, -5.25) {};
		\node [style=none] (29) at (-5.5, -7.25) {};
		\node [style=none] (30) at (-2.25, -14.5) {};
		\node [style=none] (31) at (-0.5, -14.5) {};
		\node [style=none] (32) at (4, -5.25) {};
		\node [style=none] (33) at (4, -7.25) {};
		\node [style=none] (34) at (2, -11.75) {};
		\node [style=none] (35) at (2.5, -11.25) {};
		\node [style=none] (36) at (6.25, -11.25) {};
		\node [style=none] (37) at (6.75, -11.75) {};
		\node [style=none] (38) at (6.75, -17.25) {};
		\node [style=none] (39) at (6.25, -17.75) {};
		\node [style=none] (40) at (2.5, -17.75) {};
		\node [style=none] (41) at (2, -17.25) {};
		\node [style=node] (45) at (4.5, -19) {};
		\node [style=none] (46) at (4.5, -17.75) {};
		\node [style=none] (47) at (4, -11.25) {};
		\node [style=node] (48) at (4, -9.5) {};
		\node [style=node, label={above:$v_1$}] (53) at (-6.5, -13.25) {};
		\node [style=node, label={above:$v_2$}] (54) at (-4.5, -13.25) {};
		\node [style=node, label={below:$v_3$}] (55) at (-5.5, -16.5) {};
		\node [style=hedge] (56) at (4.5, -14.5) {$f$};
		\node [style=node, label={above:$v_1$}] (57) at (3.5, -13) {};
		\node [style=node, label={above:$v_2$}] (58) at (5.5, -13) {};
		\node [style=node, label={below:$v_3$}] (59) at (4.5, -16.25) {};
		\node [style=node, label={above:$v_1$}] (61) at (3, 2.75) {};
		\node [style=node, label={above:$v_2$}] (62) at (5, 2.75) {};
		\node [style=node, label={above:$v_3$}] (63) at (4, -0.5) {};
	\end{pgfonlayer}
	\begin{pgfonlayer}{edgelayer}
		\draw [style=lambda box] (1.center)
			 to (2.center)
			 to [in=90, out=0, looseness=1.75] (3.center)
			 to (4.center)
			 to [in=0, out=-90, looseness=1.75] (5.center)
			 to (6.center)
			 to [in=-90, out=-180, looseness=1.75] (7.center)
			 to (0.center)
			 to [in=180, out=90, looseness=1.75] cycle;
		\draw (16.center) to (15);
		\draw (18) to (17.center);
		\draw [style=diredge] (26.center) to (27.center);
		\draw [style=diredge] (28.center) to (29.center);
		\draw [style=diredge] (30.center) to (31.center);
		\draw [style=diredge] (32.center) to (33.center);
		\draw [style=lambda box] (35.center) to (36.center);
		\draw [style=lambda box, in=90, out=0, looseness=1.75] (36.center) to (37.center);
		\draw [style=lambda box] (37.center) to (38.center);
		\draw [style=lambda box, in=0, out=-90, looseness=1.75] (38.center) to (39.center);
		\draw [style=lambda box] (39.center) to (40.center);
		\draw [style=lambda box, in=-90, out=-180, looseness=1.75] (40.center) to (41.center);
		\draw [style=lambda box] (41.center) to (34.center);
		\draw [style=lambda box, in=180, out=90, looseness=1.75] (34.center) to (35.center);
		\draw (46.center) to (45);
		\draw (48) to (47.center);
		\draw [in=-90, out=120] (25) to (53);
		\draw [in=-90, out=60] (25) to (54);
		\draw (25) to (55);
		\draw [in=-90, out=120] (56) to (57);
		\draw [in=-90, out=60] (56) to (58);
		\draw (56) to (59);
	\end{pgfonlayer}
\end{tikzpicture}
		}
	\]
	\caption{Pushout in $\catname{EHyp}(\Sigma)$}
	\label{fig:pushout_example}
\end{figure}

\section{Interpretation of $\Sigma$-terms as cospans of hypergraphs}
\label{sec:appendix:interpretation}
\begin{figure*}
	\[
		\adjustbox{width=\textwidth}{
			\begin{tikzpicture}
	\begin{pgfonlayer}{nodelayer}
		\node [style=none] (0) at (-6.5, 3) {\Large $[\textbf{id}_{I}]$};
		\node [style=none] (1) at (-4.5, 3) {$\Coloneqq$};
		\node [style=none] (2) at (-3, 6) {};
		\node [style=none] (3) at (-1, 6) {};
		\node [style=none] (4) at (-3, 5) {};
		\node [style=none] (5) at (-1, 5) {};
		\node [style=none] (6) at (-3, 3.5) {};
		\node [style=none] (7) at (-1, 3.5) {};
		\node [style=none] (8) at (-1, 2.5) {};
		\node [style=none] (9) at (-3, 2.5) {};
		\node [style=none] (10) at (-3, 1) {};
		\node [style=none] (11) at (-1, 1) {};
		\node [style=none] (12) at (-3, 0) {};
		\node [style=none] (13) at (-1, 0) {};
		\node [style=none] (14) at (-2, 2.25) {};
		\node [style=none] (15) at (-2, 1.25) {};
		\node [style=none] (16) at (-2, 4.75) {};
		\node [style=none] (17) at (-2, 3.75) {};
		\node [style=none] (18) at (2.25, 3) {\Large $[\textbf{id}_{1}]$};
		\node [style=none] (19) at (4.25, 3) {$\Coloneqq$};
		\node [style=none] (20) at (5.75, 6.5) {};
		\node [style=none] (21) at (7.75, 6.5) {};
		\node [style=none] (22) at (5.75, 5) {};
		\node [style=none] (23) at (7.75, 5) {};
		\node [style=none] (24) at (5.75, 3.75) {};
		\node [style=none] (25) at (7.75, 3.75) {};
		\node [style=none] (26) at (7.75, 2.25) {};
		\node [style=none] (27) at (5.75, 2.25) {};
		\node [style=none] (28) at (5.75, 1) {};
		\node [style=none] (29) at (7.75, 1) {};
		\node [style=none] (30) at (5.75, -0.5) {};
		\node [style=none] (31) at (7.75, -0.5) {};
		\node [style=none] (32) at (6.75, 2) {};
		\node [style=none] (33) at (6.75, 1.25) {};
		\node [style=none] (34) at (6.75, 4.75) {};
		\node [style=none] (35) at (6.75, 4) {};
		\node [style=node, label={above:$v_{1}$}] (36) at (6.75, 5.5) {};
		\node [style=node, label={above:$v_{1}$}] (39) at (6.75, 2.75) {};
		\node [style=node, label={above:$v_{1}$}] (40) at (6.75, 0) {};
		\node [style=none] (41) at (11.75, 3) {\Large $[\textbf{sym}_{1,1}]$};
		\node [style=none] (42) at (13.75, 3) {$\Coloneqq$};
		\node [style=none] (43) at (15.25, -0.5) {};
		\node [style=none] (44) at (18.75, -0.5) {};
		\node [style=none] (45) at (15.25, 1) {};
		\node [style=none] (46) at (18.75, 1) {};
		\node [style=none] (47) at (15.25, 2.25) {};
		\node [style=none] (48) at (18.75, 2.25) {};
		\node [style=none] (49) at (18.75, 3.75) {};
		\node [style=none] (50) at (15.25, 3.75) {};
		\node [style=none] (51) at (15.25, 5) {};
		\node [style=none] (52) at (18.75, 5) {};
		\node [style=none] (53) at (15.25, 6.5) {};
		\node [style=none] (54) at (18.75, 6.5) {};
		\node [style=none] (55) at (17, 4) {};
		\node [style=none] (56) at (17, 4.75) {};
		\node [style=none] (57) at (17, 1.25) {};
		\node [style=none] (58) at (17, 2) {};
		\node [style=node, label={below:$v_{1}$}] (59) at (16, 0.5) {};
		\node [style=node, label={above:$v_{1}$}] (60) at (16, 2.75) {};
		\node [style=node, label={above:$v_{1}$}] (61) at (18, 5.5) {};
		\node [style=node, label={below:$v_{2}$}] (62) at (18, 0.5) {};
		\node [style=node, label={above:$v_{2}$}] (65) at (18, 2.75) {};
		\node [style=node, label={above:$v_{2}$}] (66) at (16, 5.5) {};
		\node [style=none] (67) at (21.5, 3) {\Large $[o]$};
		\node [style=none] (68) at (23.5, 3) {$\Coloneqq$};
		\node [style=none] (69) at (25, -2.25) {};
		\node [style=none] (70) at (28.5, -2.25) {};
		\node [style=none] (71) at (25, -0.75) {};
		\node [style=none] (72) at (28.5, -0.75) {};
		\node [style=none] (73) at (25, 0.5) {};
		\node [style=none] (74) at (28.5, 0.5) {};
		\node [style=none] (75) at (28.5, 5.75) {};
		\node [style=none] (76) at (25, 5.75) {};
		\node [style=none] (77) at (25, 7) {};
		\node [style=none] (78) at (28.5, 7) {};
		\node [style=none] (79) at (25, 8.5) {};
		\node [style=none] (80) at (28.5, 8.5) {};
		\node [style=none] (81) at (26.75, 6) {};
		\node [style=none] (82) at (26.75, 6.75) {};
		\node [style=none] (83) at (26.75, -0.5) {};
		\node [style=none] (84) at (26.75, 0.25) {};
		\node [style=node, label={below:$v_{1}$}] (85) at (25.75, -1.25) {};
		\node [style=node, label={below:$v_{1}$}] (86) at (25.75, 1.75) {};
		\node [style=node, label={above:$w_{1}$}] (87) at (25.75, 7.5) {};
		\node [style=node, label={below:$v_{m}$}] (88) at (27.75, -1.25) {};
		\node [style=node, label={below:$v_{m}$}] (89) at (27.75, 1.75) {};
		\node [style=node, label={above:$w_{n}$}] (90) at (27.75, 7.5) {};
		\node [style=none] (91) at (27, 9.5) {$\textit{where } o : v_1 \otimes \ldots \otimes v_m \to w_1 \otimes \ldots \otimes w_n \in \Sigma_{M} $};
		\node [style=hedge] (92) at (26.75, 3.25) {$o$};
		\node [style=none] (93) at (26.75, 1.75) {$\ldots$};
		\node [style=none] (94) at (26.75, -1.25) {$\ldots$};
		\node [style=none] (95) at (26.75, 7.5) {$\ldots$};
		\node [style=node, label={above:$w_{1}$}] (96) at (25.75, 4.75) {};
		\node [style=node, label={above:$w_{n}$}] (97) at (27.75, 4.75) {};
		\node [style=none] (98) at (26.75, 4.75) {$\ldots$};
	\end{pgfonlayer}
	\begin{pgfonlayer}{edgelayer}
		\draw [style=boundary frame] (2.center)
			 to (3.center)
			 to (5.center)
			 to (4.center)
			 to cycle;
		\draw [style=graph frame] (8.center)
			 to (9.center)
			 to (6.center)
			 to (7.center)
			 to cycle;
		\draw [style=boundary frame] (10.center)
			 to (11.center)
			 to (13.center)
			 to (12.center)
			 to cycle;
		\draw [style=diredge] (15.center) to (14.center);
		\draw [style=diredge] (16.center) to (17.center);
		\draw [style=boundary frame] (20.center)
			 to (21.center)
			 to (23.center)
			 to (22.center)
			 to cycle;
		\draw [style=graph frame] (26.center)
			 to (27.center)
			 to (24.center)
			 to (25.center)
			 to cycle;
		\draw [style=boundary frame] (28.center)
			 to (29.center)
			 to (31.center)
			 to (30.center)
			 to cycle;
		\draw [style=diredge] (33.center) to (32.center);
		\draw [style=diredge] (34.center) to (35.center);
		\draw [style=boundary frame] (43.center)
			 to (44.center)
			 to (46.center)
			 to (45.center)
			 to cycle;
		\draw [style=graph frame] (49.center)
			 to (50.center)
			 to (47.center)
			 to (48.center)
			 to cycle;
		\draw [style=boundary frame] (51.center)
			 to (52.center)
			 to (54.center)
			 to (53.center)
			 to cycle;
		\draw [style=diredge] (56.center) to (55.center);
		\draw [style=diredge] (57.center) to (58.center);
		\draw [style=boundary frame] (69.center)
			 to (70.center)
			 to (72.center)
			 to (71.center)
			 to cycle;
		\draw [style=graph frame] (75.center)
			 to (76.center)
			 to (73.center)
			 to (74.center)
			 to cycle;
		\draw [style=boundary frame] (77.center)
			 to (78.center)
			 to (80.center)
			 to (79.center)
			 to cycle;
		\draw [style=diredge] (82.center) to (81.center);
		\draw [style=diredge] (83.center) to (84.center);
		\draw [in=90, out=-120] (92) to (86);
		\draw [in=-60, out=90] (89) to (92);
		\draw [in=-240, out=-90] (96) to (92);
		\draw [in=-90, out=60, looseness=0.75] (92) to (97);
	\end{pgfonlayer}
\end{tikzpicture}
		}
	\]
	\caption{Base cases for $[-] : \textbf{SMT}(\Sigma) \to \MdaCospans$}
	\label{fig:base_cases}
\end{figure*}

The base cases for interpretation function $[-] : \textbf{S}(\Sigma) \to \MdaCospans$ are given in~\autoref{fig:base_cases}.
Then it is defined inductively by letting $[f \otimes g] \Coloneqq [f] \otimes [g]$ and $[f \fatsemi g] \Coloneqq [f] \fatsemi [g]$ and by freeness induces the corresponding functor.
We obtain the interpretation function $\llbracket - \rrbracket : \Sigma^{+} \to \WellTypedMdaEcospans$ by combining the above with the interpretations of $\llbracket \textbf{ev}_{A,B} \rrbracket$ and $\llbracket \Lambda_{A,B,C} \rrbracket$ shown in~\autoref{fig:ev_and_lambda} and by interpreting $\llbracket f + g \rrbracket$ as shown in~\autoref{fig:f+g}.
The rest again follows by induction and freeness.
Note that unlike $[-]$, $\llbracket - \rrbracket$ does not get promoted to a functor $\catname{CS}^{+}(\Sigma) \to \WellTypedMdaEcospans$; this is because target of the to-be functor is not a \emph{closed} symmetric monoidal $\catname{SLat}$-category, although we will turn it into a symmetric monoidal $\catname{SLat}$-category.
This is achieved by quotienting morphisms of $\WellTypedMdaEcospans$ by EDPOI rewrite rules that arise from the interpretation of $\catname{SLat}$ SMC axioms.
For example, the EDPOI rule corresponding to the distributivity of $( \fatsemi )$ over $+$ is given by the following rewrite rule.
\[
	\adjustbox{width=\linewidth}{
		\input{./figures/semilattice_rule_1_2.tikz}~.
	}
\]
Additionally, we quotient by the distributivity law (\autoref{law:distributivity}).
The rewrite (schema) rule corresponding to it looks as follows.

\[
\adjustbox{width=\linewidth}{
		\input{./figures/semilat_distributivity_law1.tikz}~.
	}
\]
The collection of these rules is denoted $\mathcal{S}$ and defines the corresponding category $\WellTypedMdaEcospans / \mathcal{S}$ where morphisms (cospans) are identified if they can be rewritten into each other.

\begin{figure}
	\begin{subfigure}[c]{0.45\linewidth}
		\[
			\adjustbox{scale=0.6}{
				\begin{tikzpicture}
	\begin{pgfonlayer}{nodelayer}
		\node [style=none] (49) at (-0.5, -2.5) {};
		\node [style=none] (50) at (-5.5, -2.5) {};
		\node [style=none] (51) at (-5.5, 2.75) {};
		\node [style=none] (52) at (-0.5, 2.75) {};
		\node [style=none] (0) at (-10.25, -0.5) {\Large$\llbracket\textsf{ev}_{A,B}\rrbracket$};
		\node [style=none] (1) at (-6.5, -0.5) {$\Coloneqq$};
		\node [style=none] (2) at (-2, -0.5) {};
		\node [style=none] (3) at (-4, -0.5) {};
		\node [style=none] (4) at (-2.25, 0.25) {};
		\node [style=none] (5) at (-3.75, 0.25) {};
		\node [style=none] (6) at (-3, 0.5) {};
		\node [style=node] (8) at (-3.5, -1.5) {};
		\node [style=node, label={below:$B$}] (9) at (-2.5, -1.5) {};
		\node [style=none] (10) at (-3.5, -0.5) {};
		\node [style=none] (11) at (-2.5, -0.5) {};
		\node [style=node] (12) at (-3, 1.75) {};
		\node [style=node, label={above:$B $}] (13) at (-3, 1.75) {};
		\node [style=node, label={below:$B$}] (17) at (-2.25, -7.75) {};
		\node [style=none] (18) at (-1, -8.75) {};
		\node [style=none] (19) at (-5, -8.75) {};
		\node [style=none] (20) at (-5, -7.25) {};
		\node [style=none] (21) at (-1, -7.25) {};
		\node [style=node, label={below:$B$}] (23) at (-2.25, -4.5) {};
		\node [style=none] (24) at (-1, -5.75) {};
		\node [style=none] (25) at (-5, -5.75) {};
		\node [style=none] (26) at (-5, -4) {};
		\node [style=none] (27) at (-1, -4) {};
		\node [style=node, label={above:$B $}] (29) at (-3, 8.25) {};
		\node [style=none] (30) at (-2.25, 7.75) {};
		\node [style=none] (31) at (-3.75, 7.75) {};
		\node [style=none] (32) at (-3.75, 9.25) {};
		\node [style=none] (33) at (-2.25, 9.25) {};
		\node [style=none] (40) at (-3, 7.25) {};
		\node [style=none] (41) at (-3, 6) {};
		\node [style=none] (42) at (-3, 6.5) {};
		\node [style=none] (43) at (-3, 4) {};
		\node [style=none] (44) at (-3, 3.25) {};
		\node [style=none] (45) at (-3, -7) {};
		\node [style=none] (46) at (-3, -6.25) {};
		\node [style=none] (47) at (-3, -3.75) {};
		\node [style=none] (48) at (-3, -3) {};
		\node [style=node, label={above:$B $}] (53) at (-3, 5) {};
		\node [style=none] (54) at (-2.25, 4.5) {};
		\node [style=none] (55) at (-3.75, 4.5) {};
		\node [style=none] (56) at (-3.75, 6) {};
		\node [style=none] (57) at (-2.25, 6) {};
		\node [style=none] (58) at (-4, -2.125) {$A \multimap B$};
		\node [style=node] (59) at (-3.5, -4.5) {};
		\node [style=none] (60) at (-4, -5.125) {$A \multimap B$};
		\node [style=node] (61) at (-3.5, -7.75) {};
		\node [style=none] (62) at (-4, -8.425) {$A \multimap B$};
	\end{pgfonlayer}
	\begin{pgfonlayer}{edgelayer}
		\draw [style=graph frame] (51.center)
			 to (52.center)
			 to (49.center)
			 to (50.center)
			 to cycle;
		\draw [style={lambda_unit}] (5.center)
			 to [in=180, out=45] (6.center)
			 to [in=135, out=0, looseness=0.75] (4.center)
			 to [in=90, out=-45] (2.center)
			 to (3.center)
			 to [in=-135, out=90, looseness=0.75] cycle;
		\draw [style={lambda_unit}] (9) to (11.center);
		\draw [style={lambda_unit}] (10.center) to (8);
		\draw [style={lambda_unit}] (13) to (6.center);
		\draw [style=boundary frame] (18.center)
			 to (19.center)
			 to (20.center)
			 to (21.center)
			 to cycle;
		\draw [style=boundary frame] (24.center)
			 to (25.center)
			 to (26.center)
			 to (27.center)
			 to cycle;
		\draw [style=boundary frame] (30.center)
			 to (31.center)
			 to (32.center)
			 to (33.center)
			 to cycle;
		\draw [style=diredge] (40.center) to (42.center);
		\draw [style=diredge] (43.center) to (44.center);
		\draw [style=diredge] (45.center) to (46.center);
		\draw [style=diredge] (47.center) to (48.center);
		\draw [style=boundary frame] (54.center)
			 to (55.center)
			 to (56.center)
			 to (57.center)
			 to cycle;
	\end{pgfonlayer}
\end{tikzpicture}
			}
		\]
	\end{subfigure}
	\hfill
	\begin{subfigure}[c]{0.45\linewidth}
		\[
			\adjustbox{scale=0.6}{
				\begin{tikzpicture}
	\begin{pgfonlayer}{nodelayer}
		\node [style=none] (57) at (-4, 4) {};
		\node [style=none] (58) at (2.25, 4) {};
		\node [style=none] (59) at (2.25, -7.5) {};
		\node [style=none] (60) at (-4, -7.5) {};
		\node [style=none] (10) at (-3.25, -4.25) {};
		\node [style=none] (11) at (-2.75, -4.75) {};
		\node [style=none] (12) at (1, -4.75) {};
		\node [style=none] (13) at (1.5, -4.25) {};
		\node [style=none] (14) at (1.5, 1.25) {};
		\node [style=none] (15) at (1, 1.75) {};
		\node [style=none] (16) at (-2.75, 1.75) {};
		\node [style=none] (17) at (-3.25, 1.25) {};
		\node [style=none] (0) at (-9, -1.5) {\Large$\llbracket\Lambda_{A,B,C}(f)\rrbracket$};
		\node [style=none] (1) at (-5.5, -1.5) {$\Coloneqq$};
		\node [style=medium box] (2) at (-0.75, -1.5) {$\llbracket f \rrbracket$};
		\node [style=node, label={below:$A$}] (3) at (-1.25, -3.75) {};
		\node [style=node, label={below:$B$}] (4) at (-0.25, -3.75) {};
		\node [style=none] (5) at (-1.25, -2) {};
		\node [style=none] (6) at (-0.25, -2) {};
		\node [style=node, label={above:$C$}] (7) at (-0.75, 0.75) {};
		\node [style=none] (8) at (-0.75, -0.75) {};
		\node [style=none] (9) at (-2.25, 3.25) {$B \multimap C$};
		\node [style=node] (18) at (-0.75, 3) {};
		\node [style=none] (19) at (-0.75, 1.75) {};
		\node [style=none] (20) at (-1.25, -4.75) {};
		\node [style=node, label={below:$A$}] (21) at (-1.25, -6.5) {};
		\node [style=node, label={below:$A$}] (22) at (-1, -9.25) {};
		\node [style=node, label={below:$B$}] (23) at (0, -9.25) {};
		\node [style=node, label={below:$A$}] (24) at (-2, -9.25) {};
		\node [style=none] (25) at (-3, -10.25) {};
		\node [style=none] (26) at (1, -10.25) {};
		\node [style=none] (27) at (1, -8.75) {};
		\node [style=none] (28) at (-3, -8.75) {};
		\node [style=node, label={below:$A$}] (29) at (-1, -12.5) {};
		\node [style=none] (32) at (1, -13.5) {};
		\node [style=none] (33) at (-3, -13.5) {};
		\node [style=none] (34) at (-3, -12) {};
		\node [style=none] (35) at (1, -12) {};
		\node [style=none] (36) at (-1, -11.5) {};
		\node [style=none] (37) at (-1, -10.75) {};
		\node [style=none] (38) at (-1, -8.5) {};
		\node [style=none] (39) at (-1, -7.75) {};
		\node [style=node] (40) at (-1, 10) {};
		\node [style=none] (41) at (-2, 10.5) {$B \multimap C$};
		\node [style=node] (42) at (-1, 6.5) {};
		\node [style=none] (43) at (-2, 7.125) {$B \multimap C$};
		\node [style=node, label={above:$C$}] (44) at (0.25, 6.5) {};
		\node [style=none] (45) at (-3, 6) {};
		\node [style=none] (46) at (1, 6) {};
		\node [style=none] (47) at (1, 7.5) {};
		\node [style=none] (48) at (-3, 7.5) {};
		\node [style=none] (49) at (1, 9.5) {};
		\node [style=none] (50) at (-3, 9.5) {};
		\node [style=none] (51) at (-3, 10.75) {};
		\node [style=none] (52) at (1, 10.75) {};
		\node [style=none] (53) at (-1, 8.75) {};
		\node [style=none] (54) at (-1, 8) {};
		\node [style=none] (55) at (-1, 5.25) {};
		\node [style=none] (56) at (-1, 4.5) {};
	\end{pgfonlayer}
	\begin{pgfonlayer}{edgelayer}
		\draw [style=graph frame] (59.center)
			 to (58.center)
			 to (57.center)
			 to [in=90, out=-90] (60.center)
			 to cycle;
		\draw [style=lambda box] (11.center)
			 to (12.center)
			 to [in=-90, out=0, looseness=1.75] (13.center)
			 to (14.center)
			 to [in=0, out=90, looseness=1.75] (15.center)
			 to (16.center)
			 to [in=90, out=180, looseness=1.75] (17.center)
			 to (10.center)
			 to [in=180, out=-90, looseness=1.75] cycle;
		\draw (3) to (5.center);
		\draw (4) to (6.center);
		\draw (8.center) to (7);
		\draw (19.center) to (18);
		\draw (21) to (20.center);
		\draw [style=boundary frame] (28.center)
			 to (25.center)
			 to (26.center)
			 to (27.center)
			 to cycle;
		\draw [style=boundary frame] (35.center)
			 to (32.center)
			 to (33.center)
			 to (34.center)
			 to cycle;
		\draw [style=diredge] (36.center) to (37.center);
		\draw [style=diredge] (38.center) to (39.center);
		\draw [style=boundary frame] (45.center)
			 to (46.center)
			 to (47.center)
			 to (48.center)
			 to cycle;
		\draw [style=boundary frame] (50.center)
			 to (51.center)
			 to (52.center)
			 to (49.center)
			 to cycle;
		\draw [style=diredge] (53.center) to (54.center);
		\draw [style=diredge] (55.center) to (56.center);
	\end{pgfonlayer}
\end{tikzpicture}
			}
		\]
	\end{subfigure}
	\caption{$\llbracket \textsf{ev} \rrbracket$ (left) and $\llbracket \Lambda(f) \rrbracket$ (right)}
	\label{fig:ev_and_lambda}
\end{figure}

\begin{figure}
	\[
		\adjustbox{scale=0.5}{
			\input{./figures/f_plus_g_new_2.tikz}
		}
	\]
	\captionsetup{belowskip=-1ex}
	\caption{$+$ of two morphisms in $\WellTypedMdaEcospans$}
	\label{fig:f+g}
\end{figure}

\section{EDPOI rewriting}%
\label{sec:appendix:dpoi}

Consider an example of EDPOI rewriting in~\autoref{fig:dpoi-example} as defined in~\autoref{def:dpoi-e}.
The rule in the top half of the diagram is an instance of a $\beta$-reduction rule as defined for string diagrams~\cite{ghica2024stringdiagramslambdacalculifunctional}.
Note that the match satisfies all the conditions of the boundary complement as all vertices from the interface have the same parent.
First, to compute the boundary complement, we remove from the source e-hypergraph everything that has no pre-image in the interface above.
The mda-cospan for the complement then will be
\[
	[u_1] \to [u_1, v_1, v_2] \sqcup  [v_3] \to \mathcal{L}^{\bot} \xleftarrow{} [v_3] \sqcup  [v_1,v_2] \xleftarrow{} [u_2] ~.
\]
Note how we removed strict internal interfaces of $\mathcal{L}$ from internal interfaces of $\mathcal{G}$, yielding $[u_1,v_1,v_2]$ and similarly for output interfaces.
After computing the pushout for the second half of the diagram we add strict internal interfaces of $\mathcal{R}$ but since there are none, we are done as per condition (2) of~\autoref{def:dpoi-e}.

\begin{figure}
		\[
			\adjustbox{width=0.6\linewidth}{
				\input{./figures/dpo_example.tikz}
			}
		\]
	\caption{$\beta$-rule EDPOI example.}
	\label{fig:dpoi-example}
\end{figure}

\section{Modelling $\lambda_{\text{lsub}}$-calculus in slotted e-graphs}%
\label{sec:e-graph-explicit-subst}


\begin{listing}
	\begin{minted}{rust}
pub enum Lambda {
   Lam(Bind<AppliedId>) = "lam",
   App(AppliedId, AppliedId) = "app",
   // subst $x t u := t[x / u]
   Subst(Bind<AppliedId>, AppliedId) = "subst",
   Var(Slot) = "var",
}
	\end{minted}
	\caption{$\lambda_{\text{lsub}}$ syntax definition in slotted e-graphs}
	\label{listing:define_language}
\end{listing}

\begin{listing}
	\begin{minted}{rust}
fn beta() -> Rewrite<Lambda> {
    let pat = "(app (lam $1 ?b) ?t)";
    let outpat = "(subst $1 ?b ?t)";
    Rewrite::new("beta", pat, outpat)
}
	\end{minted}
	\caption{$\lambda_{\text{lsub}}$ $d\beta$-rule}
	\label{listing:lsub_beta_rule}
\end{listing}

Recall the grammar, the reduction rules and the equational theory of $\lambda_{\text{lsub}}$ from~\autoref{sec:application}.
These can be directly encoded in the framework of slotted e-graphs of Schneider et al.~\cite{slotted-egraphs} as can be seen in \autoref{listing:define_language}, \autoref{listing:lsub_beta_rule}, \autoref{listing:lsub_ls_rules}, \autoref{listing:gc_rule}, and \autoref{listing:graphical_equivalence}.
We use \mintinline{rust}{"subst"} to denote explicit substitution and use the built-in substitution mechanism provided by slotted e-graphs to encode $C\llbracket x \rrbracket[x / u] \to C\llbracket u \rrbracket[x/u]$ rules.
Notably, since equivalence is symmetric, the corresponding rules have two versions --- \textit{left} and \textit{right}.
We check for the side conditions using conditional rewrites provided by the slotted e-graphs framework.

To show the effect of those particular sets of rewrites when performing equality saturation for $\lambda_{\text{lsub}}$ we apply certain subsets of these rewrites to a simple term that encodes $2 + 2$ using Church encoding.
When doing so, we keep track of the number of e-nodes in the saturated e-graph as well as the number of times each rewrite was applied during saturation (we count only the applications that change the e-graph).
We also note whether a particular rewrite is absorbed by using string-diagrammatic (e-hypergraph) representation of $\lambda_{\text{lsub}}$-terms.
The idea is to show how much workload can be absorbed by using the above-mentioned representation that represents certain equivalences on the nose.
\autoref{tbl:slotted} shows the number of e-nodes after saturation with different sets of rewrite rules associated with $\lambda_{\text{lsub}}$ enabled.
\texttt{lambda\_subst} is the distributivity of substitution over lambda abstraction (encoded using two rewrite rules) which is the middle equation of graphical equivalence from~\autoref{sec:application}, \texttt{subst\_comm} is the commutativity of substitutions which corresponds to the first equation and \texttt{app\_subst} is the third equation.
We can see that enabling all those rules --- that is truly performing rewriting modulo graphical equivalence --- significantly increases the number of e-nodes in the e-graph.
Since \texttt{subst\_comm} and \texttt{app\_subst} are absorbed by the e-hypegraph representation we theoretically can get away with fewer nodes if we used this representation.

Another metric that we use is the number of times each rewrite was applied during saturation which is shown in \autoref{tbl:slotted-rewrite-applications}.
$gc$ rule is encoded using \texttt{subst\_no\_var} rewrite rule.
We can see that the rules that are absorbed by the e-hypergraph representation (highlighted in grey) comprise more than half of all the rewrite applications during saturation when all the rules are enabled.
So again, using a more suitable representation can significantly reduce the workload of the equality saturation in this case.

\begin{remark}
Even though $gc$ rule is not absorbed by the e-hypergraph representation (despite being a part of string diagram equivalence), we can argue that this rule can be disregarded when using e-hypergraphs.
Consider a substitution $y[x / v]$ which corresponds to the following composition in the internal language
\[
\inferrule*{\Gamma_1,y : T, \Gamma_2 \vdash v : T' \qquad \Gamma_1, x : T', y : T, \Gamma_2 \vdash y : T}{
	\Gamma_1, y : T, \Gamma_2 \vdash y[x / v] : T \equiv y
}
\]
This composition is represented using the string diagram in~\autoref{fig:subst_no_var_string}.
While string diagrammatically this composition equals the string diagram for just $y$ since we are working modulo the equations for copy and delete, they are not isomorphic as e-hypergraphs.
However, since the output of $v$ gets deleted anyway, we can ignore any rewrites that involve $v$ and its inputs.
The deletion morphisms for $\Gamma_1, \Gamma_2$ that follow the copy can be absorbed by the e-hypergraph representation by following the approach of Milosavljevi\'{c} et al.~\cite{zanassi_comonoid}.
\end{remark}

\begin{listing}
	\begin{minted}{rust}
fn subst_no_var() -> Rewrite<Lambda> {
   let pat = "(subst $x ?t1 ?t2)";
   let outpat = "?t1";
   Rewrite::new_if("subst_no_var", pat, outpat, |subst, _ | {
      !subst["t1"].slots().contains(&Slot::named("x"))
   })
}
	\end{minted}
	\caption{$\lambda_{\text{lsub}}$ $gc$-rule}
	\label{listing:gc_rule}
\end{listing}

\begin{figure}
	\centering
	\begin{tikzpicture}
	\begin{pgfonlayer}{nodelayer}
		\node [style=none] (0) at (-5.5, -0.75) {};
		\node [style=vertex] (1) at (-5.5, 0.75) {};
		\node [style=none] (2) at (-6, 0) {$\Gamma_1$};
		\node [style=none] (3) at (-4.75, -0.75) {};
		\node [style=vertex] (4) at (-4.75, 0.75) {};
		\node [style=none] (5) at (-4.25, 0) {$y$};
		\node [style=none] (6) at (-3.25, -0.75) {};
		\node [style=vertex] (7) at (-3.25, 0.75) {};
		\node [style=none] (8) at (-2.75, 0) {$\Gamma_2$};
		\node [style=round box] (9) at (-3.25, 3) {$v$};
		\node [style=none] (10) at (-3.5, 2.5) {};
		\node [style=none] (11) at (-3.25, 2.5) {};
		\node [style=none] (12) at (-3, 2.5) {};
		\node [style=vertex] (13) at (-6.5, 4.5) {};
		\node [style=none] (14) at (-5.5, 4.5) {};
		\node [style=vertex] (15) at (-4.5, 4.5) {};
		\node [style=none] (16) at (-3.25, 3.5) {};
		\node [style=vertex] (17) at (-3.25, 4.5) {};
		\node [style=none] (18) at (-3.5, 2) {};
		\node [style=none] (19) at (-3.25, 2) {};
		\node [style=none] (20) at (-3, 2) {};
	\end{pgfonlayer}
	\begin{pgfonlayer}{edgelayer}
		\draw (1) to (0.center);
		\draw (4) to (3.center);
		\draw (7) to (6.center);
		\draw [in=-90, out=120, looseness=0.75] (1) to (13);
		\draw [in=-90, out=120, looseness=0.75] (4) to (14.center);
		\draw [in=-90, out=120, looseness=0.75] (7) to (15);
		\draw (16.center) to (17);
		\draw (10.center) to (18.center);
		\draw (11.center) to (19.center);
		\draw (12.center) to (20.center);
		\draw [bend left=15] (20.center) to (7);
		\draw [bend left=45, looseness=0.75] (19.center) to (4);
		\draw [in=-90, out=30, looseness=0.75] (1) to (18.center);
	\end{pgfonlayer}
\end{tikzpicture}
	\caption{String diagram for $y[x/v]$}
	\label{fig:subst_no_var_string}
\end{figure}

\begin{table*}
	\begin{tabular}{lcccc}
		Set of equations & $d\beta$ + $ls$ + $gc$ & $d\beta$ + $ls$ + $gc$    & $d\beta$ + $ls$ + $gc$     & $d\beta$ + $ls$ + $gc$\\
		                 &                        & + \texttt{lambda\_subst}  & + \texttt{lambda\_subst}   & + \texttt{lambda\_subst}\\
		                 &                        &                           & + \texttt{subst\_comm}     & + \texttt{subst\_comm}\\
						 &                        &                           &                            & + \texttt{subst\_app}\\						 
		\# of e-nodes    & 88                     & 185                       & 244                        & 846
	\end{tabular}
	\caption{Number of e-nodes after saturating the e-graph for $2 + 2$ with different sets of rewrite rules for $\lambda_{\text{lsub}}$}
	\label{tbl:slotted}
\end{table*}

\begin{table*}
	\begin{tabular}{l|ccccc}
		Rewrite rule & \# of applications     & \# of applications         & \# of applications      & \# of applications       & Absorbed?\\
		             &($d\beta$ + $ls$ + $gc$)&($d\beta$ + $ls$ + $gc$     &($d\beta$ + $ls$ + $gc$  &($d\beta$ + $ls$ + $gc$   &\\
		             &                        & + \texttt{lambda\_subst})  & + \texttt{lambda\_subst}& + \texttt{lambda\_subst} &\\
		             &                        &                            & + \texttt{subst\_comm}) & + \texttt{subst\_comm}   &\\
		             &                        &                            &                         & + \texttt{subst\_app})   &\\
		\hline
		\texttt{beta} & 14 & 57 & 49 & 80 & No\\
		\texttt{subst\_no\_var} & 35 & 113 & 118 & 197 & No\\
		\rowcolor{gray!30}
		\texttt{ls\_rule\_*} & 48 & 141 & 140 & 302 & Yes\\
		\texttt{subst\_lambda} & --- & 76 & 109 & 336 & No\\
		\rowcolor{gray!30}
		\texttt{subst\_comm} & --- & --- & 104 & 321 & Yes\\
		\rowcolor{gray!30}
		\texttt{subst\_app} & --- & --- & --- & 263 & Yes\\
	\end{tabular}
	\caption{Number of times each rewrite rule was applied during saturation of the e-graph for $2 + 2$}
	\label{tbl:slotted-rewrite-applications}
\end{table*}


\clearpage

\begin{listing*}
	\centering
\begin{minipage}{0.8\textwidth}
	\begin{minted}{rust}
fn ls_rule_var() -> Rewrite<Lambda> {
   let pat = "(subst $x (var $x) ?t)";
   let outpat = "(subst $x (var $x)[(var $x) := ?t] ?t)";
   Rewrite::new("ls_var", pat, outpat)
}

fn ls_rule_app_l() -> Rewrite<Lambda> {
   let pat = "(subst $x (app ?u ?v) ?t)";
   let outpat = "(subst $x (app ?u[(var $x) := ?t] ?v) ?t)";
   Rewrite::new_if("ls_app_l", pat, outpat, |subst, _| {
      subst["u"].slots().contains(&Slot::named("x")) 
	  && !subst["v"].slots().contains(&Slot::named("x"))
    })
}

fn ls_rule_app_r() -> Rewrite<Lambda> {
   let pat = "(subst $x (app ?u ?v) ?t)";
   let outpat = "(subst $x (app ?u ?v[(var $x) := ?t]) ?t)";
   Rewrite::new_if("ls_app_r", pat, outpat, |subst, _| {
      !subst["u"].slots().contains(&Slot::named("x")) 
	  && subst["v"].slots().contains(&Slot::named("x"))
   })
}

fn ls_rule_lam() -> Rewrite<Lambda> {
   let pat = "(subst $x (lam $y ?b) ?t)";
   let outpat = "(subst $x (lam $y ?b[(var $x) := ?t]) ?t)";
   Rewrite::new_if("ls_lam", pat, outpat, |subst, _| {
      subst["b"].slots().contains(&Slot::named("x"))
   })
}

fn ls_rule_subst_l() -> Rewrite<Lambda> {
   let pat = "(subst $x (subst $y ?u ?t1) ?t2)";
   let outpat = "(subst $x (subst $y ?u ?t1[(var $x) := ?t2]) ?t2)";
   Rewrite::new_if("ls_subst_l", pat, outpat, |subst, _| {
      subst["t1"].slots().contains(&Slot::named("x"))
   })
}

fn ls_rule_subst_r() -> Rewrite<Lambda> {
    let pat = "(subst $x (subst $y ?u ?t1) ?t2)";
    let outpat = "(subst $x (subst $y ?u[(var $x) := ?t2] ?t1) ?t2)";
    Rewrite::new_if("ls_subst_r", pat, outpat, |subst, _| {
          subst["u"].slots().contains(&Slot::named("x"))
    })
}
	\end{minted}
\end{minipage}
	\caption{$\lambda_{\text{lsub}}$ $ls$-rules}
	\label{listing:lsub_ls_rules}
\end{listing*}

\clearpage

\begin{listing*}
	\centering
	\begin{minipage}{0.8\textwidth}
	\begin{minted}{rust}
fn subst_comm() -> Rewrite<Lambda> {
   let pat = "(subst $y (subst $x ?u ?t1) ?t2)";
   let outpat = "(subst $x (subst $y ?u ?t2) ?t1)";
   Rewrite::new_if("subst_comm", pat, outpat, |subst,_| {
      !subst["t2"].slots().contains(&Slot::named("x")) 
      && !subst["t1"].slots().contains(&Slot::named("y"))
   })
}

fn lambda_subst_l() -> Rewrite<Lambda> {
   let pat = "(subst $x (lam $y ?b) ?t)";
   let outpat = "(lam $y (subst $x ?b ?t))";
   Rewrite::new("lambda_subst_l", pat, outpat)
}

fn lambda_subst_r() -> Rewrite<Lambda> {
   let pat = "(lam $y (subst $x ?b ?t))";
   let outpat = "(subst $x (lam $y ?b) ?t)";
   Rewrite::new_("lambda_subst_r", pat, outpat, |subst, _| {
      !subst["t"].slots().contains(&Slot::named("y"))
   })
}

fn subst_app_l() -> Rewrite<Lambda> {
   let pat = "(subst $x (app ?t1 ?t2) ?t3)";
   let outpat = "(app (subst $x ?t1 ?t3) ?t2)";
   Rewrite::new_if("subst_app_2", pat, outpat, |subst,_| {
      !subst["t2"].slots().contains(&Slot::named("x"))
   })
}
fn subst_app_r() -> Rewrite<Lambda> {
   let pat = "(app (subst $x ?t1 ?t3) ?t2)";
   let outpat = "(subst $x (app ?t1 ?t2) ?t3)";
   Rewrite::new_if("subst_app_2_2", pat, outpat, |subst,_| {
      !subst["t2"].slots().contains(&Slot::named("x"))
   })
}
	\end{minted}
\end{minipage}
	\caption{Graphical equivalence rules for $\lambda_{\text{lsub}}$}
	\label{listing:graphical_equivalence}
\end{listing*}







\end{document}